\documentclass[11pt]{article}
\usepackage{amssymb,amsmath}
\usepackage{epsfig}
\usepackage{hyperref}
\usepackage{cite} 

\usepackage{times}

\newcommand{\qedsymb}{\hfill{\rule{2mm}{2mm}}}
\newenvironment{proof}[1][]{\begin{trivlist}
\item[\hspace{\labelsep}{\bf\noindent Proof#1:\/}] }{\qedsymb\end{trivlist}}

\mathchardef\ordinarycolon\mathcode`\:     
\mathcode`\:=\string"8000
\def\vcentcolon{\mathrel{\mathop\ordinarycolon}} \begingroup
\catcode`\:=\active \lowercase{\endgroup \let :\vcentcolon }

\usepackage[usenames]{color}

\newcommand{\ignore}[1]{}

\newtheorem{theorem}{Theorem}[section]

\newtheorem{definition}[theorem]{Definition}

\newtheorem{claim}[theorem]{Claim}
\newtheorem{lemma}[theorem]{Lemma}

\newtheorem{coro}[theorem]{Corollary}

\newcommand{\smfrac}[2]{\mbox{$\frac{#1}{#2}$}}

\newcommand{\ket}[1]{|#1\rangle}
\newcommand{\bra}[1]{\langle#1|}

\newcommand{\ketbra}[2]{|#1\rangle\langle#2|}

\newcommand{\integer}{\mathbb{Z}}

\newcommand{\QMA}{{\sf{QMA}}}
\newcommand{\QMAEXP}{{\sf{QMA_{EXP}}}}

\newcommand{\BQEXP}{{\sf{BQEXP}}}
\newcommand{\NEXP}{{\sf{NEXP}}}
\newcommand{\EXP}{{\sf{EXP}}}

\newcommand{\BQP}{{\sf{BQP}}}
\newcommand{\NTIME}{{\sf{NTIME}}}
\newcommand{\Pclass}{{\sf{P}}}

\newcommand{\NP}{{\sf{NP}}}

\newcommand{\PSPACE}{{\sf{PSPACE}}}
\newcommand{\EXPSPACE}{{\sf{EXPSPACE}}}

\newcommand{\mysymbol}[1]{{\mbox{\raisebox{-0.3em}{\epsfysize=1.2em\epsfbox{#1}}}}}

\newcommand{\leftend}{\mysymbol{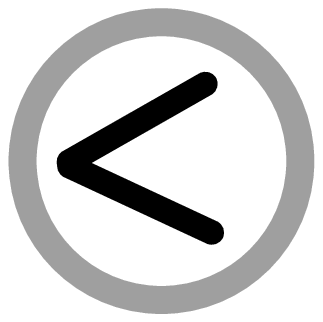}}
\newcommand{\rightend}{\mysymbol{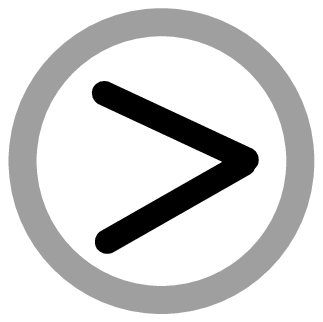}}
\newcommand{\variable}{\mysymbol{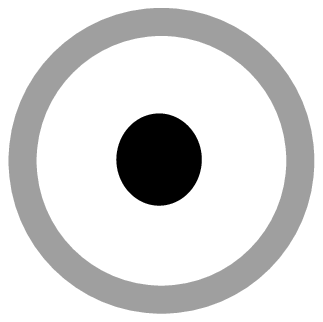}}

\newcommand{\arrR}{\mysymbol{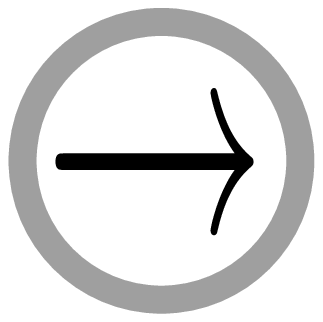}}
\newcommand{\arrRzero}{\mysymbol{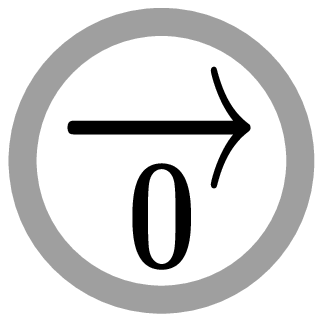}}
\newcommand{\arrRone}{\mysymbol{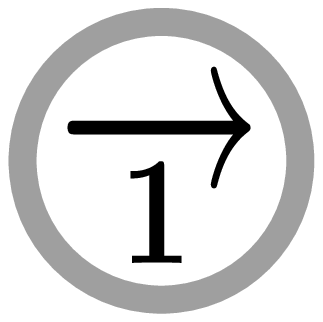}}

\newcommand{\arrRtwo}{\mysymbol{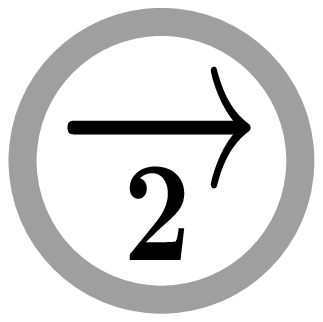}}

\newcommand{\arrL}{\mysymbol{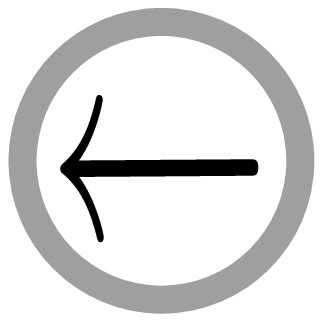}}
\newcommand{\arrLzero}{\mysymbol{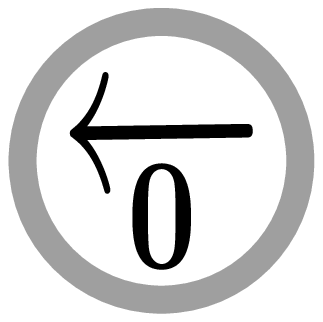}}
\newcommand{\arrLone}{\mysymbol{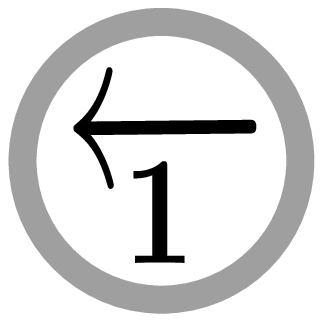}}
\newcommand{\arrLtwo}{\mysymbol{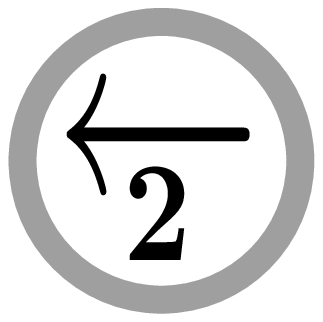}}

\newcommand{\zero}{\mysymbol{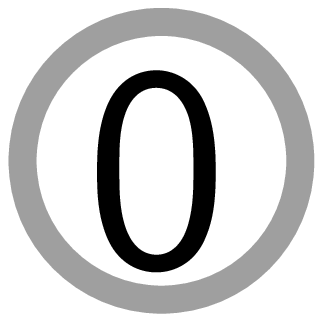}}
\newcommand{\zeroB}{\mysymbol{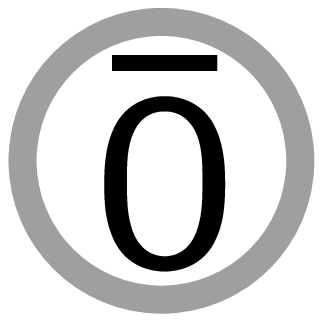}}
\newcommand{\one}{\mysymbol{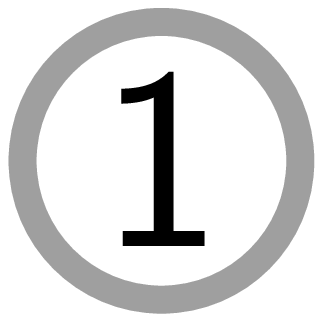}}
\newcommand{\oneB}{\mysymbol{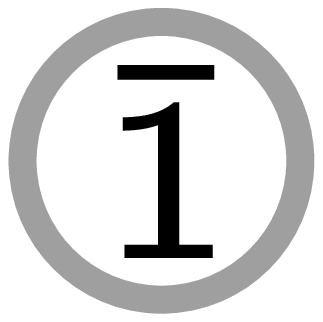}}
\newcommand{\two}{\mysymbol{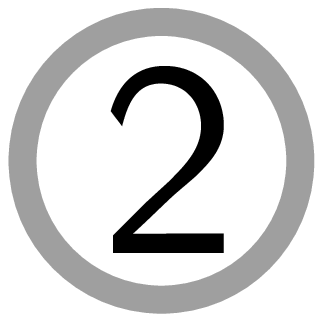}}

\newcommand{\blankR}{\mysymbol{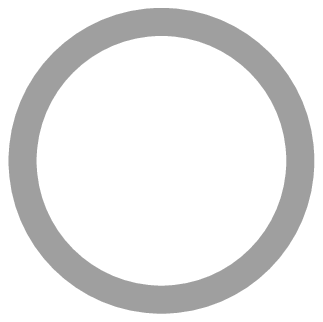}}
\newcommand{\blankL}{\mysymbol{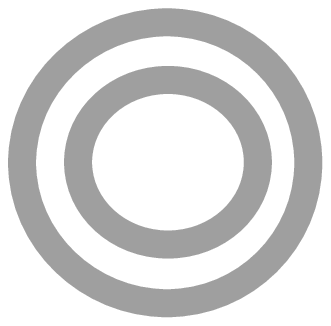}}

\newcommand{\stateS}{\mysymbol{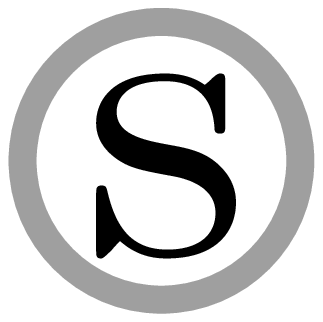}}

\newcommand{\stateA}{\mysymbol{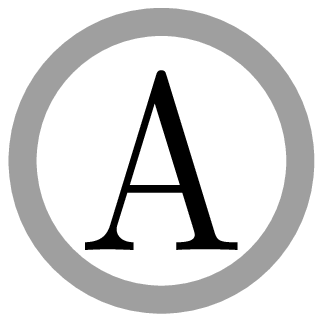}}
\newcommand{\stateB}{\mysymbol{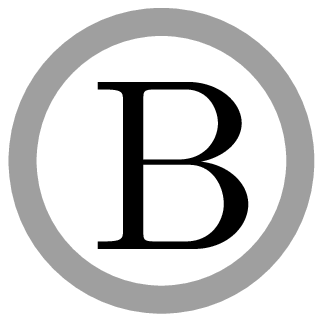}}
\newcommand{\stateLine}{\mysymbol{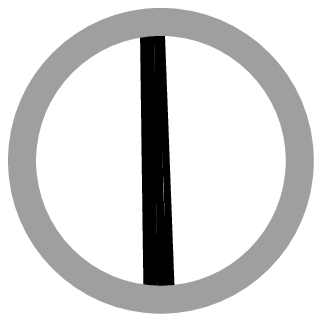}}

\newcommand{\arrRA}{\mysymbol{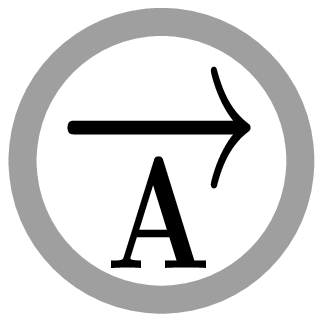}}
\newcommand{\arrRB}{\mysymbol{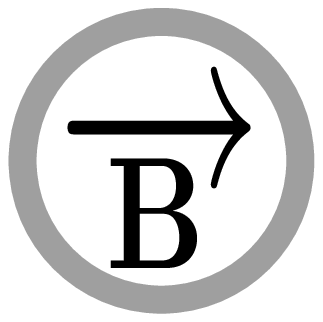}}
\newcommand{\arrLA}{\mysymbol{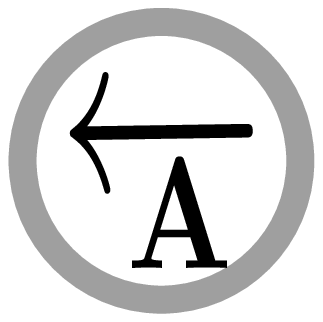}}
\newcommand{\arrLB}{\mysymbol{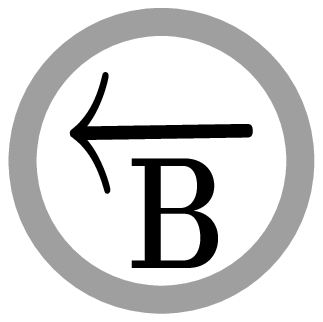}}
\newcommand{\arrLAzero}{\mysymbol{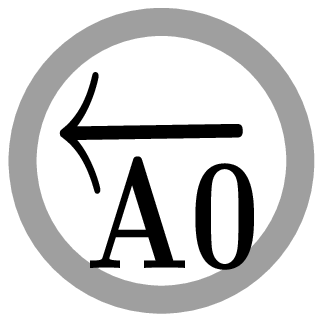}}
\newcommand{\arrLBzero}{\mysymbol{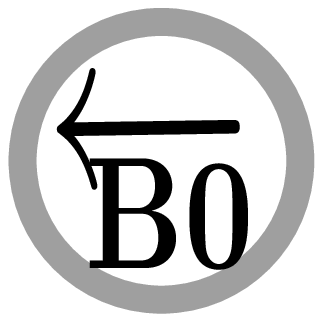}}
\newcommand{\arrLAone}{\mysymbol{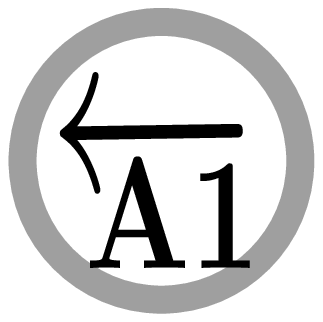}}
\newcommand{\arrLBone}{\mysymbol{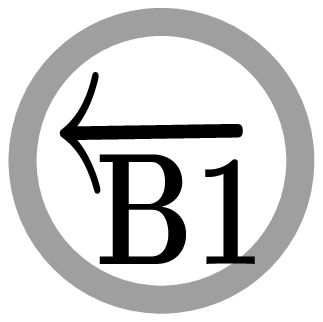}}
\newcommand{\arrLAtwo}{\mysymbol{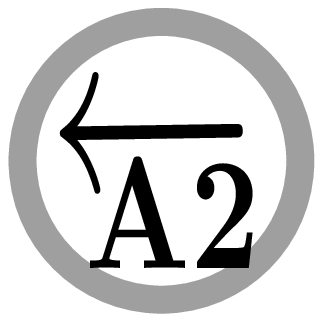}}
\newcommand{\arrLBtwo}{\mysymbol{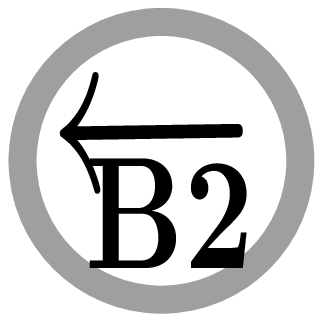}}
\newcommand{\arrRAzero}{\mysymbol{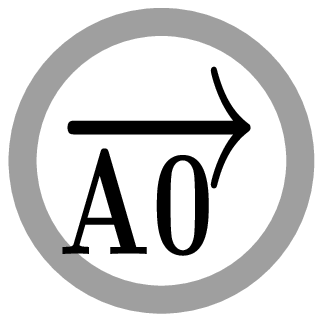}}
\newcommand{\arrRBzero}{\mysymbol{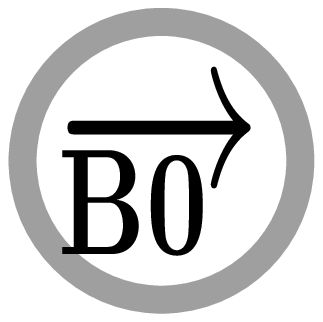}}

\newcommand{\arrRAtwo}{\mysymbol{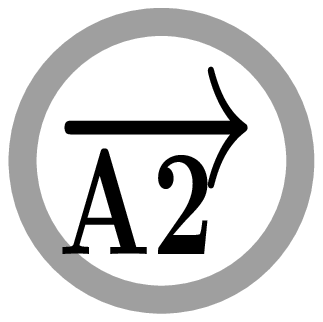}}

\newcommand{\zeroVerA}{\mysymbol{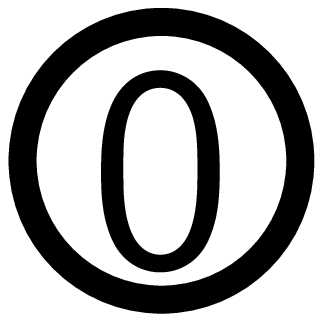}}
\newcommand{\oneVerA}{\mysymbol{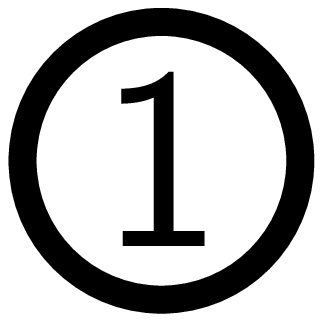}}
\newcommand{\twoVerA}{\mysymbol{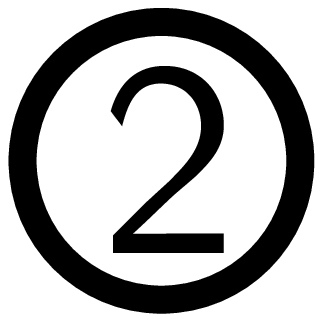}}
\newcommand{\zeroVerB}{\mysymbol{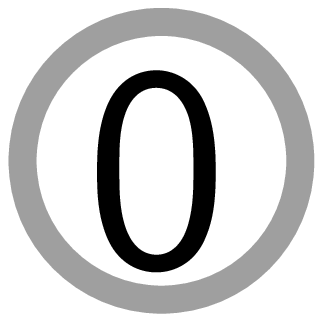}}
\newcommand{\oneVerB}{\mysymbol{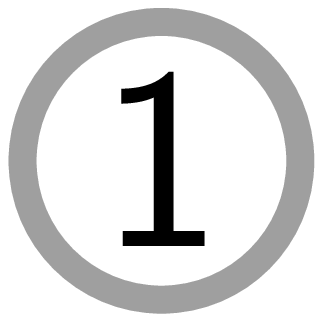}}
\newcommand{\twoVerB}{\mysymbol{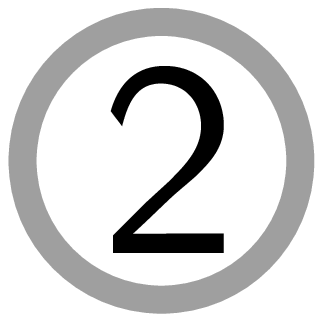}}

\newcommand{\tileC}{\mysymbol{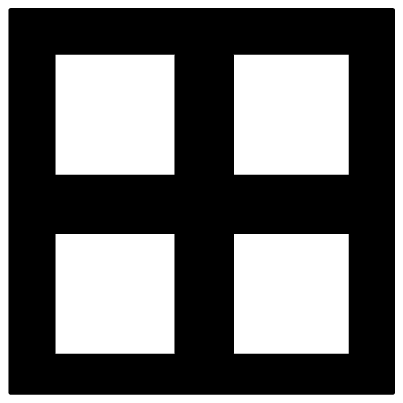}}       
\newcommand{\tileV}{\mysymbol{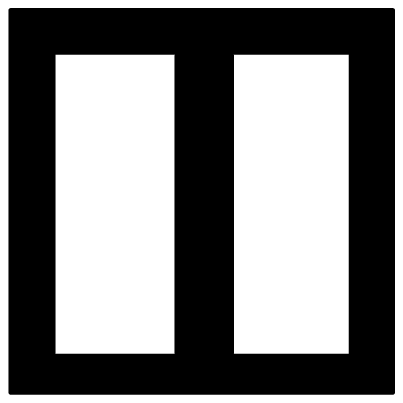}}       
\newcommand{\tileH}{\mysymbol{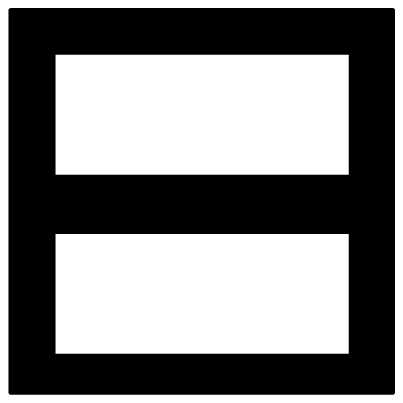}}       

\newcommand{\tileW}{\mysymbol{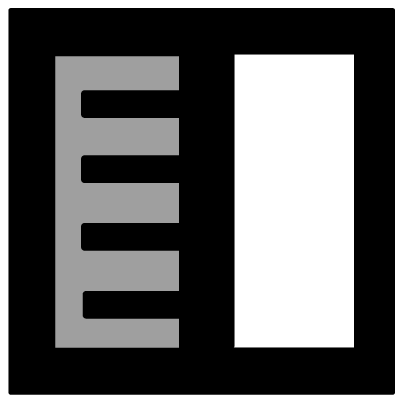}}      
\newcommand{\tileN}{\mysymbol{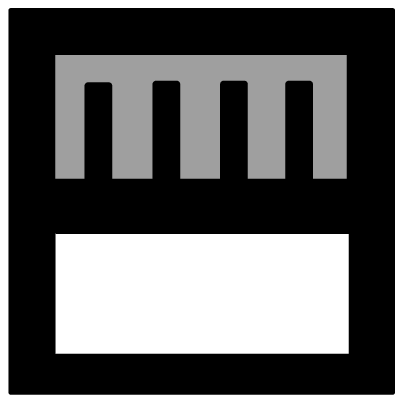}}       
\newcommand{\tileS}{\mysymbol{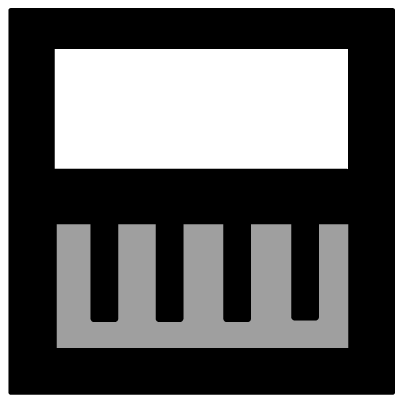}}       
\newcommand{\tileE}{\mysymbol{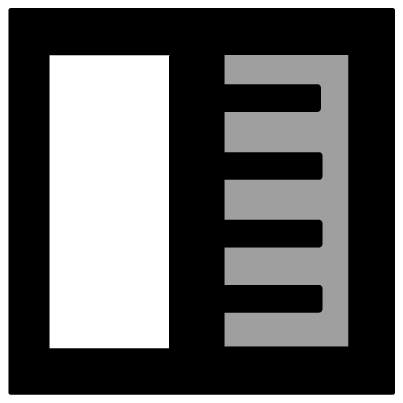}}       
\newcommand{\tileNW}{\mysymbol{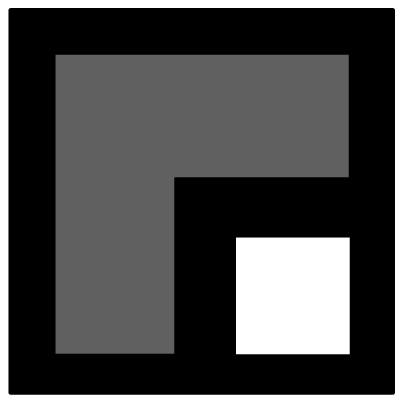}}     
\newcommand{\tileNE}{\mysymbol{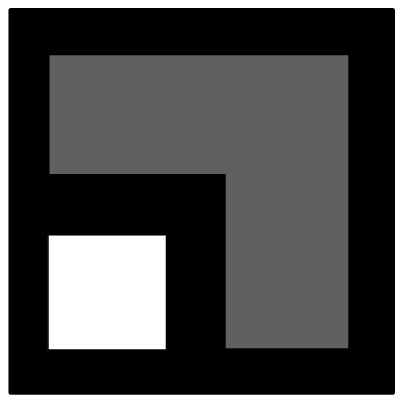}}     
\newcommand{\tileSE}{\mysymbol{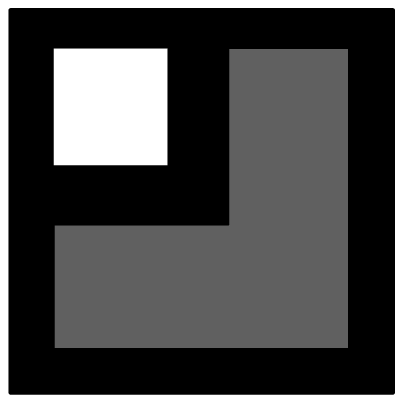}}      
\newcommand{\tileSW}{\mysymbol{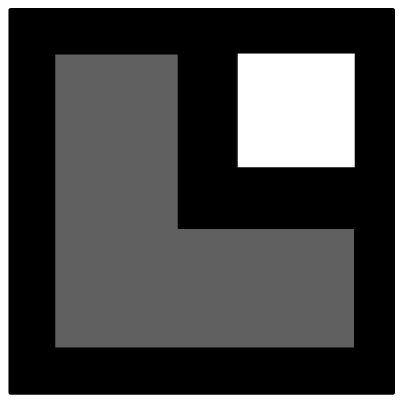}}      

\newcommand{\tileWhite}{\mysymbol{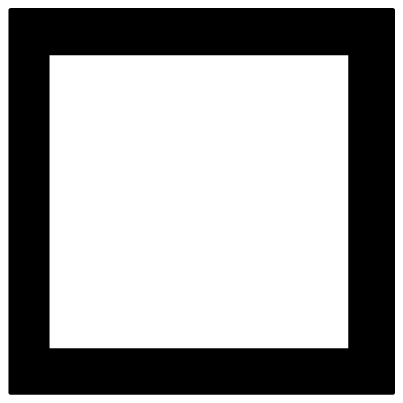}}   
\newcommand{\tileBlack}{\mysymbol{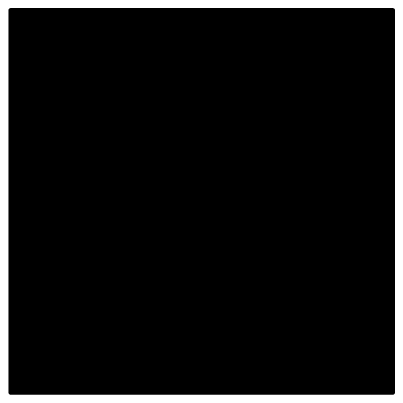}}   
\newcommand{\tileLight}{\mysymbol{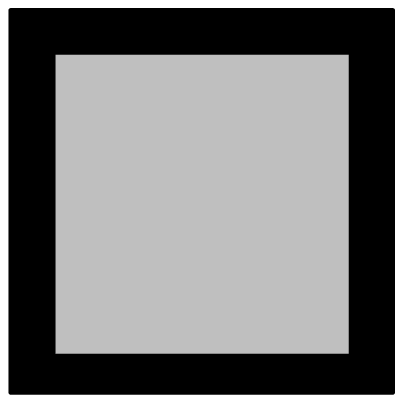}}   
\newcommand{\tileDark}{\mysymbol{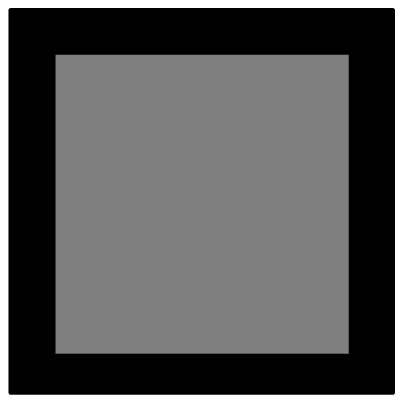}}    

\newcommand{\tileBlackD}{\mysymbol{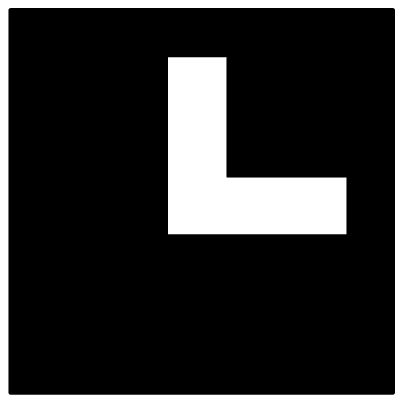}}  
\newcommand{\tileWhiteU}{\mysymbol{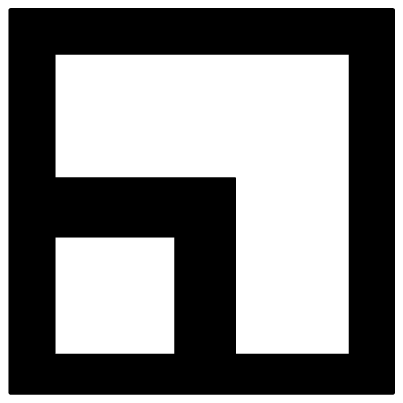}}  

\newcommand{\tileHvar}{\mysymbol{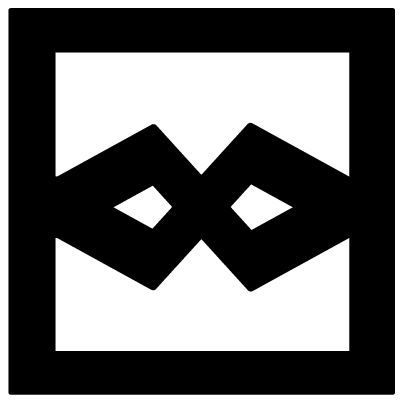}}  
\newcommand{\tileHrev}{\mysymbol{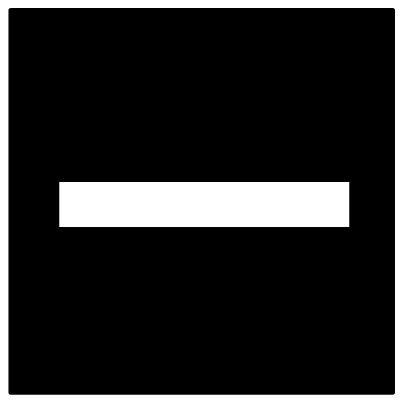}}       
\newcommand{\tileHvarrev}{\mysymbol{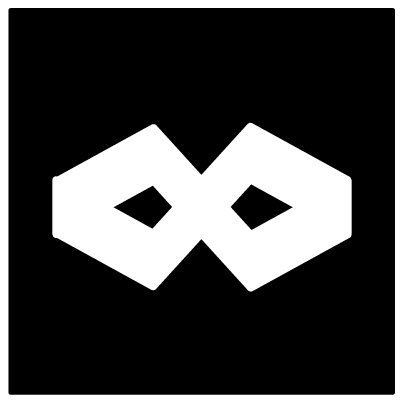}}  

\newcommand{\tileVvar}{\mysymbol{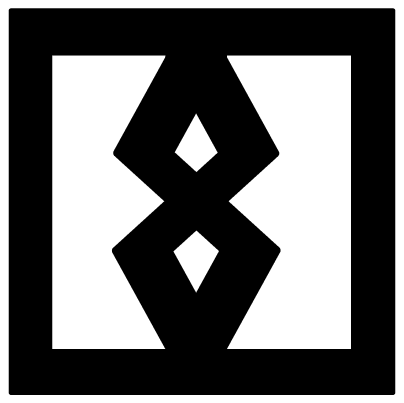}}  
\newcommand{\tileVrev}{\mysymbol{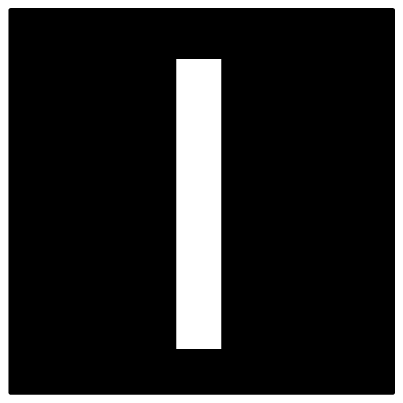}}       
\newcommand{\tileVvarrev}{\mysymbol{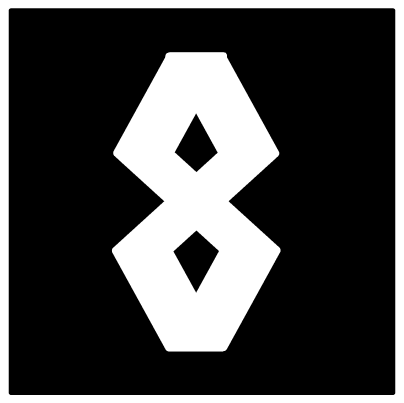}}  

\newcommand{\tileHH}{\mysymbol{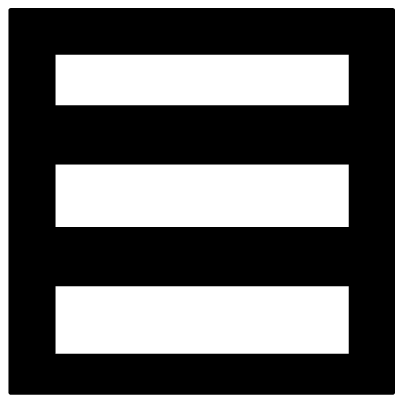}}  
\newcommand{\tileHHrev}{\mysymbol{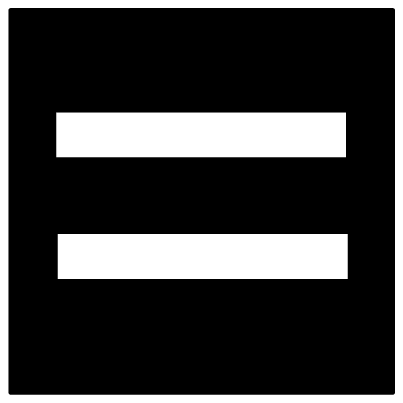}}  
\newcommand{\tileVV}{\mysymbol{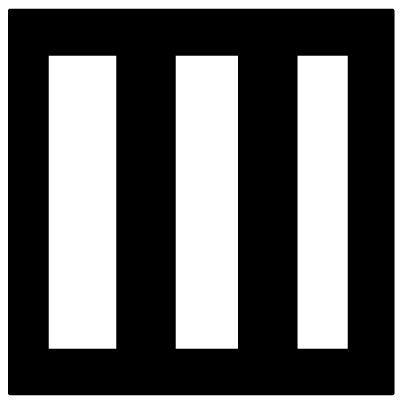}}  
\newcommand{\tileVVrev}{\mysymbol{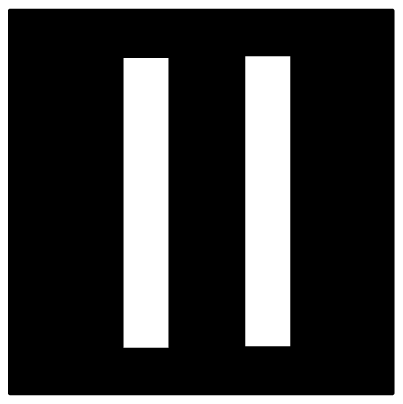}}  

\newcommand{\tileCrev}{\mysymbol{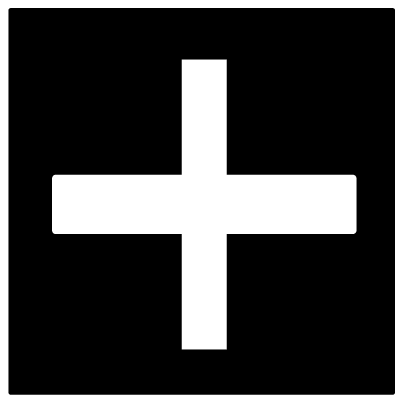}}       
\newcommand{\tileCrosshatch}{\mysymbol{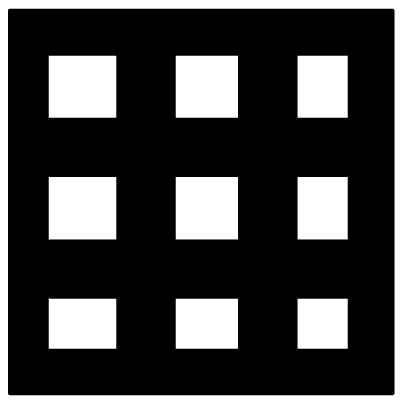}}  
\newcommand{\tileCrosshatchrev}{\mysymbol{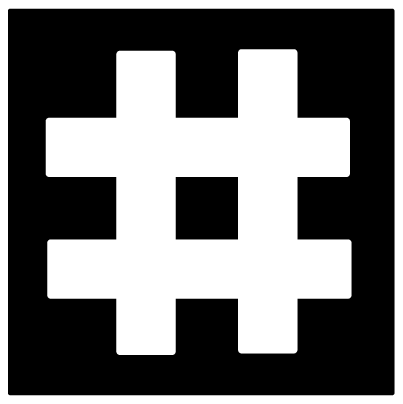}}  

\newcommand{\tileHVV}{\mysymbol{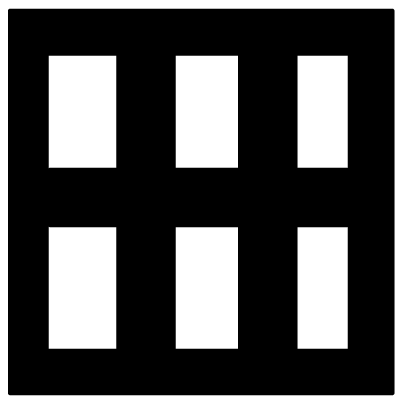}}  
\newcommand{\tileHVVrev}{\mysymbol{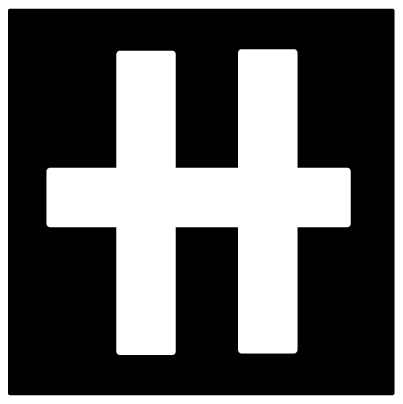}}  
\newcommand{\tileHHV}{\mysymbol{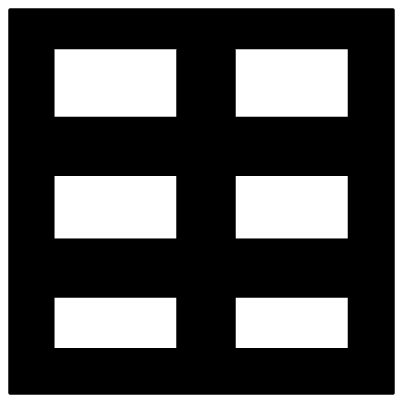}}  
\newcommand{\tileHHVrev}{\mysymbol{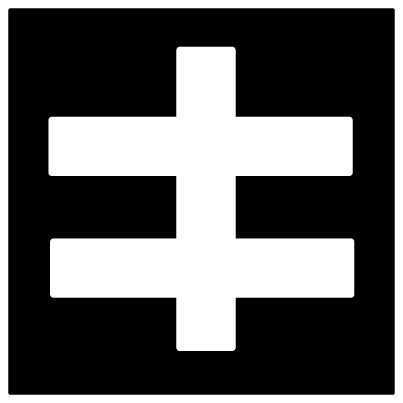}}  

\newcommand{\tileCirc}{\mysymbol{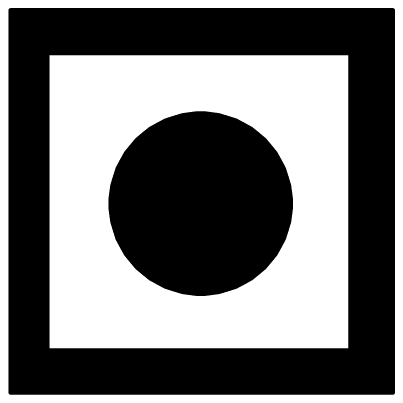}}  
\newcommand{\tileCircrev}{\mysymbol{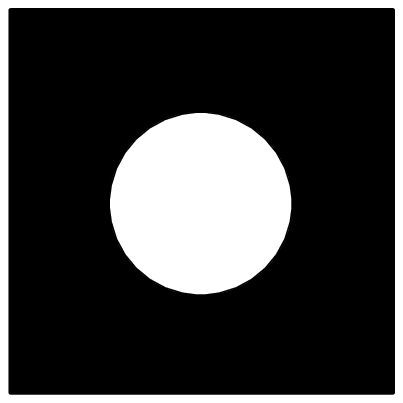}}  
\newcommand{\tileRing}{\mysymbol{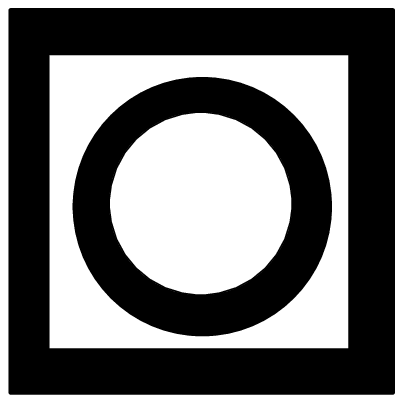}}  
\newcommand{\tileRingrev}{\mysymbol{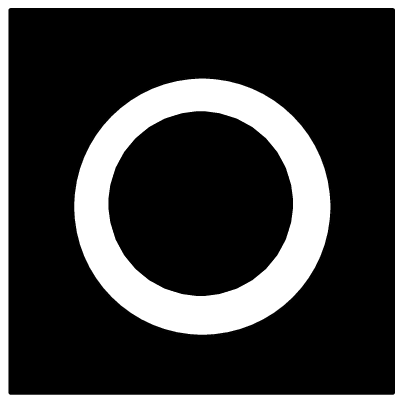}}  

\newcommand{\tileD}{\mysymbol{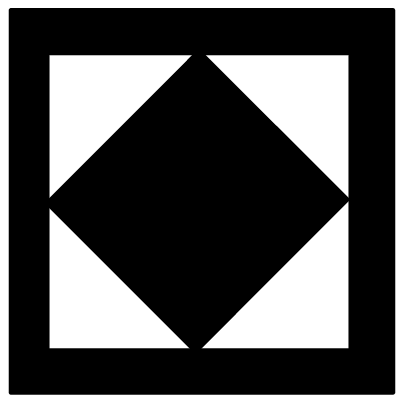}}  
\newcommand{\tileDrev}{\mysymbol{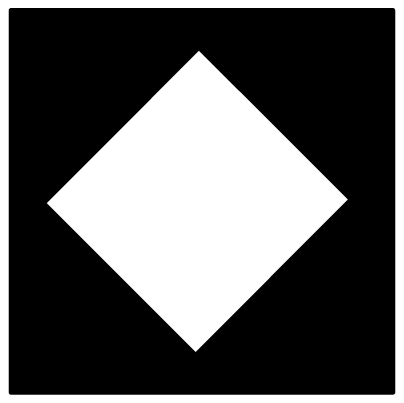}}  
\newcommand{\tileDHol}{\mysymbol{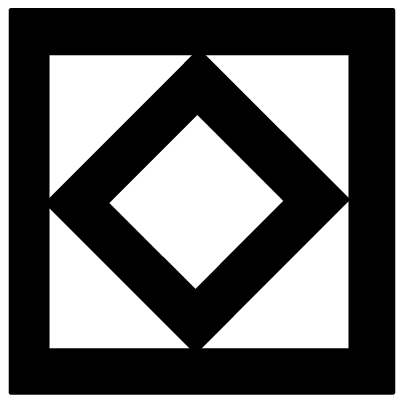}}  
\newcommand{\tileDHolrev}{\mysymbol{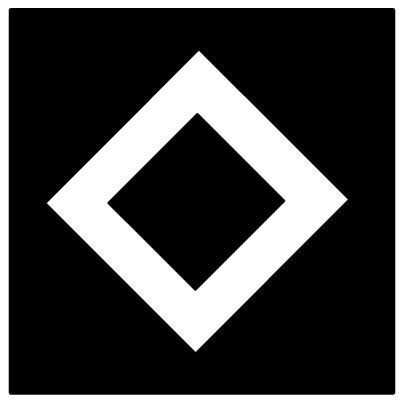}}  

\newcommand{\tileSquare}{\mysymbol{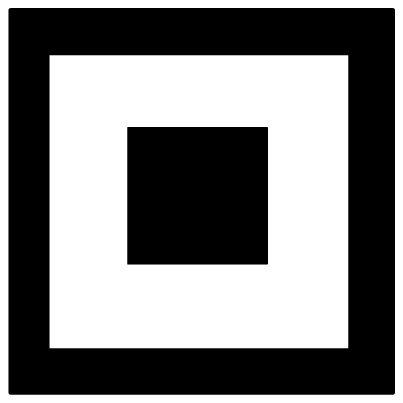}}  
\newcommand{\tileSquarerev}{\mysymbol{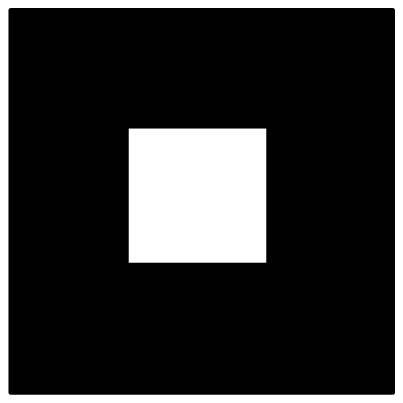}}  
\newcommand{\tileX}{\mysymbol{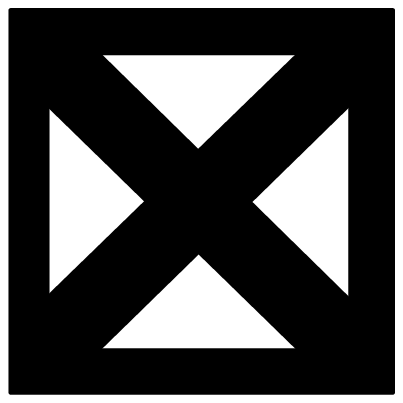}}  
\newcommand{\tileXrev}{\mysymbol{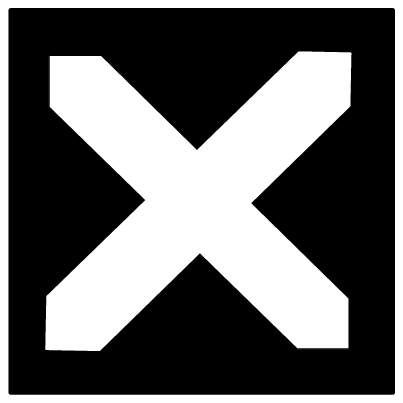}}  

\newcommand{\tilevariable}{*}       

\newcommand{\threecellsL}[2]{\begin{array}{|@{}c@{}|@{}c@{}|} \hline  \leftend &
	\begin{array}{@{}c@{}} #1 \\ \hline  #2 \\ \end{array}\\ \hline  \end{array}}
\newcommand{\threecellsR}[2]{\begin{array}{|@{}c@{}|@{}c@{}|} \hline
	\begin{array}{@{}c@{}} #1 \\ \hline  #2  \end{array} & \rightend \\ \hline  \end{array}}
\newcommand{\fourcells}[4]{\begin{array}{|@{}c@{}|@{}c@{}|} \hline  #1 & #3 \\ \hline #2 & #4 \\ \hline   \end{array}}

\newcommand{\threecellsLRefl}[2]{\begin{array}{|@{}c@{}|@{}c@{}|} \hline  \stateLine &
	\begin{array}{@{}c@{}} #1 \\ \hline  #2 \\ \end{array}\\ \hline  \end{array}}
\newcommand{\threecellsRRefl}[2]{\begin{array}{|@{}c@{}|@{}c@{}|} \hline
	\begin{array}{@{}c@{}} #1 \\ \hline  #2 \\ \end{array} & \stateLine \\ \hline  \end{array}}

\newcommand{\eightcells}[8]{\begin{tabular}{|@{}c@{}| @{}c@{}| @{}c@{}|@{}c@{}|@{}c@{}|@{}c@{}|} \hline
\leftend &
	\begin{tabular}{@{}c@{}} #1 \\ \hline  #5   \end{tabular} &
	\begin{tabular}{@{}c@{}} #2 \\ \hline  #6   \end{tabular} &
	\begin{tabular}{@{}c@{}} #3 \\ \hline  #7   \end{tabular} &
	\begin{tabular}{@{}c@{}} #4 \\ \hline  #8   \end{tabular} &
	 \rightend \\ \hline  \end{tabular}}

\newcommand{\sixcells}[6]{\begin{array}{|@{}c@{}| @{}c@{}|} \hline
          #1 & #2 \\ \hline
          #3 & #4 \\ \hline
          #5 & #6 \\ \hline
          \end{array}}

\newcommand{\twotilesvert}[2]{\begin{array}{@{}c@{}}  #1 \\ #2 \\ \end{array}}

\def\calH{{\cal{H}}}
\def\calS{{\cal{S}}}
\def\calT{{\cal{T}}}

\def\mns{{\mbox{-}}}

\newcommand{\rTIH}[1]{${#1}$-DIM TIH}
\newcommand{\ITIH}{ITIH}

\newcommand{\tiling}{TILING}

\title{The Quantum and Classical Complexity of Translationally Invariant Tiling and Hamiltonian Problems}

\author{
Daniel Gottesman
\and Sandy Irani
}

\begin{document}

\maketitle


\begin{abstract}
We study the complexity of a class of problems involving satisfying
constraints which remain the same under translations in one or more
spatial directions.  In this paper, we show hardness of a classical tiling
problem on an $N \times N$ $2$-dimensional grid and a quantum problem involving
finding the ground state energy of a $1$-dimensional quantum system of
$N$ particles.  In both cases, the only input is $N$, provided in
binary.  We show that the classical problem is $\NEXP$-complete and
the quantum problem is $\QMAEXP$-complete.  Thus, an algorithm
for these problems which runs in time polynomial in $N$ (exponential
in the input size) would imply that
$\EXP = \NEXP$ or $\BQEXP = \QMAEXP$, respectively. Although tiling in general
is already known to be $\NEXP$-complete, to our knowledge, all
previous reductions require that either the set of tiles and their constraints
or some varying boundary conditions be given as part of the input.
In the problem considered here, these are fixed, constant-sized parameters
of the problem. Instead, the problem instance is encoded solely in the
size of the system.
\end{abstract}

\tableofcontents

\section{Introduction}

One perennial difficulty with practical applications of hardness
results is that the practically interesting instances of a hard language may not
themselves form a hard class.  One approach to solving this problem is the
difficult theory of average-case complexity \cite{levin86,bcg89},
in which one can show that ``typical'' cases of some language are
hard.  In this paper we take a different approach.  In many cases,
practically interesting instances possess some shared property, such
as a symmetry, that distinguish them from the general instance and
might, in principle, make those instances easier.  We will study
such an example and show that, even in a system possessing a great
deal of symmetry, it is still possible to prove a hardness result.

Specifically, we consider the related problems of determining
whether there is a possible tiling of an $r$-dimensional grid
with some fixed set of classical tiles and of finding the lowest energy
state (or {\em ground state}) of a quantum system involving interactions
only between neighboring particles on an $r$-dimensional grid.
The ground state energy of a system is considered one of the
basic properties of a physical system, and over the last few decades,
physicists have developed a number of heuristics that have been successful
in finding the ground state energy in many special cases.  On the
other hand, in earlier work~\cite{focsVersion,QMA1D}, we have shown that in the most
general case, even in a $1$-dimensional quantum system, finding the ground
state is a computationally difficult problem (modulo the usual complexity-theoretic
assumptions).  However, the construction presented in~\cite{focsVersion} involves
a system which is completely unnatural from a physical point of view.  The most
interesting physical systems frequently possess an additional symmetry: translational
invariance.  In this paper, we will show that even a $1$-dimensional translationally-invariant system can be hard.

One interesting feature of our proof which may have more general
applicability is that the only free parameter for the language we
consider is the size of the system.  This is frequently the case
for interesting systems: there is a basic set of rules of constant
size, and we wish to study the effect of those rules when the system
to which the rules apply becomes large.  In practice, many such systems
seem difficult to solve, but it is hard to see how to prove a
complexity-theoretic hardness result, since that requires reducing
a general problem in some complexity class to the language under consideration,
and there doesn't seem to be room in the language to fit all the needed
instances.
Usually, this difficulty is circumvented by modifying the problem slightly, to
add additional parameters in which we can encode the description of the
instance we wish to simulate.

To illustrate, let us present the classical tiling problem we study in
this paper: We are given a set of square tiles which come
in a variety of colors. The area to be tiled is a square
area whose size is an integer multiple of the length of a tile.
We are given horizontal constraints indicating
which pairs of colors can be placed next to each other in the horizontal
direction and another set of constraints in the vertical direction.
We specify a particular color which must go in the four corners of the grid.
The description of the tile colors, placement constraints
and boundary conditions
are fixed for all inputs of the problems.
The input is just a number $N$ written in binary and
we wish to know whether an $N \times N$ grid can be properly tiled given
these constraints.  We show that this problem is $\NEXP$-complete.  Note
that the input in this case is size $\log N$, so an
algorithm to solve our tiling problem that runs in time polynomial
in $N$ would imply that $\NEXP = \EXP$.
While it is possible that $\Pclass \neq \NP$ and yet $\NEXP = \EXP$, this seems
unlikely to be the case.

This version of tiling is equivalent to the more common Wang Tiles
\cite{wang} in that
any set of tiles can be transformed into a set of Wang Tiles (and vice versa)
such that there is
a one-to-one correspondence between valid tilings on an $N \times N$ grid.
For an {\em infinite} grid, the problem is undecidable.  Intuitively, it
makes sense that it is also hard (for some sets of tiles) for a finite grid,
since there are exponentially many possible tilings, and it is impossible to tell
locally whether a given partial tiling can be extended indefinitely.  Indeed,
there are prior results showing that related tiling problems are $\NEXP$-complete,
but to our knowledge, all previous reductions require that either the set of tiles
and their constraints or some varying boundary conditions be given as part of the input
\cite{lp97,boas97}.
For instance, one may specify the placement of some number of tiles and ask whether
it is possible to extend that partial tiling to a tiling of the full square.
Even though that problem had been proven hard, the more natural problem of
whether it is possible to efficiently find a tiling of the empty grid remained
open.

Many $\NEXP$-complete problems are succinct versions of familiar
combinatorial problems \cite{GW83, PY86} in which the input has
some special structure which allows for a more compact representation.
For example, consider the problem of finding an independent set in a graph
where the graph is specified by indicating the number of
nodes in binary and providing a compact rule (or circuit) to determine if
two nodes are connected. Traditionally, the rule is included
as part of the input, which potentially allows for a more expressive language.
The analog to our work would be for the rule
to have a constant-sized description, fixed for
all inputs.  A philosophically similar approach has been taken by Valiant~\cite{valiant79},
who considers counting objects (for instance graphs) of a variable size $N$ with
a fixed rule (such as containing a certain fixed set of subgraphs).

The basic idea of our construction is to reduce from an instance $x$ of some
language in $\NEXP$ by
encoding $x$ in the binary expansion of $N$, the size of the grid.
It is well known that a finite set of tiling rules can be used to
implement a universal Turing Machine.  We need some way to express the
program for the Turing Machine to run, and that program must grow with
the size of $x$.  Previous constructions managed this by resorting to either
polylog~$N$ different tile types or varying boundary conditions to encode
$x$, but those are both fixed, constant-sized parameters in our version of
the problem. Instead, we use the tiles to implement a binary counter which
converts $N$ into binary and then uses it as an input to a universal Turing
Machine.

The other problem we consider is finding the ground state energy of a
quantum system.  The state of a quantum system with $N$ qubits
is a vector in a Hilbert space of dimension $2^N$. We will be considering a
slightly more general version in which an individual particle has its state
in a space of dimension $d$, in which case the state of
a system of $N$ such particles is a vector in a $d^N$-dimensional
Hilbert space.
One of the postulates of quantum mechanics states that any physical
property of a system  that can be
measured (e.g. location, momentum, energy)
corresponds to a linear operator. For an $N$-particle system,
it can be expressed as a $d^N \times d^N$ matrix over the complex numbers.
If the property is measured, then the outcome must be an eigenvalue of
the corresponding linear operator and the state of the system after
the measurement is in the eigenspace corresponding to
the outcome. Thus, the problem of finding the energy
for the lowest energy state is the same as determining the lowest
eigenvalue for the  energy operator (also
called the {\em Hamiltonian} for the system).
The difficulty, of course, is that the Hamiltonian matrix is exponentially
large in the size $N$ of the system.

We are typically interested in systems whose Hamiltonians are {\em local}
in that they can be expressed as a sum of terms each of which acts
non-trivially only on a constant-sized subset of the particles in the system.
Although the term ``local'' does not imply anything about the physical location
of the particles, it is motivated by the idea that particles only interact
when they are physically close to each other. We are therefore interested
in extending this even further and examining particles that are located
in a geometrically $r$-dimensional space where only particles within a fixed distance
can interact.
A particularly natural model to consider, then, is a system of particles on
an $r$-dimensional grid, where the terms of the Hamiltonian operate only on
neighboring pairs of particles in the grid.
Note that although the full matrix representation of a Hamiltonian is exponentially
large in the size of the system, a local Hamiltonian has a compact representation:
each term can be expressed as a constant-sized matrix, and there can
only be polynomially many such terms.

Kitaev introduced the class $\QMA$, the quantum analog of $\NP$,
and showed that the problem of determining the ground state energy of a system defined by
a local Hamiltonian is $\QMA$-hard~\cite{Kitaev:book}.
Thus, we do not hope to solve it even on a quantum computer.
With an additional promise, the problem is $\QMA$-complete: there exist two values
$a > b$, such that $a-b \ge 1/{\mathrm{poly}}(N)$, where it is guaranteed that the ground
state energy is at most $b$ or at least $a$, and one wants to determine only which
of the two alternatives holds.  In other words, we wish to determine the energy with precision $(a-b)/2$.
The problem is still hard even for two-dimensional systems
on qubits or one-dimensional systems of particles of constant Hilbert space
dimension~\cite{Oliveira:05a,focsVersion}.

Despite these worst-case results, numerical methods have been successful at
determining ground state energies for many quantum systems, especially in one dimension.
What are the differences between these hard $\QMA$-complete problems and the more tractable
systems studied by numerical physicists?
One feature of the $\QMA$-completeness constructions is that the individual terms
of the Hamiltonian are position-dependent. Essentially, the computation performed
by a quantum verifier circuit is encoded into the Hamiltonian so that a low
energy state exists if and only if there is a quantum witness that causes a
verifier to accept. Thus, the terms of the Hamiltonian encode, among other things,
individual gates in a quantum circuit. In contrast, many quantum systems of physical
interest are much more uniform in that they consist of a single Hamiltonian term that
is simultaneously applied to each pair of neighboring particles along a particular
dimension.  Such a system is called {\em translationally invariant}.

Since highly symmetric systems are rather natural, a number of researchers have
studied the computational power of translationally invariant
quantum systems.  For instance,
\cite{NWtranslation} gives a $20$-state translation-invariant modification of
the
construction from \cite{focsVersion} (improving on a $56$-state construction by \cite{JWZtranslation})
that can be used for universal $1$-dimensional adiabatic computation.
These modifications require that the system be initialized to a particular
configuration in which each particle is in a state that encodes some additional
information.
The terms of the Hamiltonian, although identical, act differently
on different particles depending on their state. The ground state is
therefore degenerate and one determines which ground state
is reached by ensuring that the system starts in a particular state.
Kay~\cite{Kaytranslation} gives a construction showing that
determining the ground state energy of a one dimensional nearest-neighbor
Hamiltonian is $\QMA$-complete even with all two-particle terms identical.
The construction does, however, require position-dependent one-particle terms.
Irani has demonstrated ground state complexity in one-dimensional
translationally-invariant systems by showing that such systems can
have ground states with a high degree of quantum entanglement \cite{irani09}.
While quantum entanglement is closely related to the performance of
numerical heuristics in practice, the particular states in this
construction are easy to compute.

In contrast, we show that there exist $1$-dimensional translationally-invariant
quantum systems with nearest-neighbor interactions for which finding the ground
state energy is complete for $\QMAEXP$, a quantum analogue of $\NEXP$.  As
with the classical result, the only parameter which varies in the language is
$N$, the number of particles, and we must use $N$ to encode the instance from
which we wish to reduce.  The quantum result uses a similar idea to the classical
result: we arrange for a control particle to shuttle between the ends of the system
and count the number of particles.  The binary encoding for the number of
particles is then used as an input to
a quantum Turing Machine.

One consequence of our result is that it is now possible to talk about the hardness
of a specific Hamiltonian term rather than the hardness of a class of Hamiltonians.
Since a system with a computationally difficult Hamiltonian cannot find its own ground
state, it is likely that such a system will behave like a spin glass at low temperatures.
The usual models of spin glasses have randomly chosen coefficients, causing a breakdown
of translational invariance.  The systems we construct in this paper are different, with
a completely ordered, translationally-invariant Hamiltonian, even though the ground states
are quite complicated.  Any disorder in the system is emergent rather than put in by hand,
a property that these spin glasses would share with structural glasses.  Some other systems
with ``self-induced disorder'' have been introduced previously, although not in the context
of computational complexity~\cite{BM94}.
In all of our hardness constructions, we construct specific two-particle terms designed
to let us prove that the resulting Hamiltonian problems are hard, but one can imagine going the
other direction, and studying the complexity of a particular Hamiltonian term given
to you.  However, we currently have no techniques for doing so.

It is worth noting that the one-dimensional version of the classical tiling problem
is very easy: it is in $\Pclass$
(see section~\ref{sec:oneDclassical} for the algorithm).  That is, it can be solved
in a time polylog~$N$, whereas it appears the quantum problem takes time $\exp(N)$,
even on a quantum computer (unless $\QMAEXP = \BQEXP$, where $\BQEXP$ is like $\BQP$,
but with exponential circuits).  Translational invariance does seem to
simplify the $1$-dimensional classical case, reducing poly($N$) time to polylog($N$) time, but it doesn't help very much in the quantum case.

Note that the classical tiling problem is a special case of the ground state energy
problem for quantum systems where the Hamiltonian is diagonal in the standard basis
with only $1$ or $0$ entries. Any ground state of such a system
is a classical state in which the state of each particle is specified by
one of the $d$ possible standard basis states, which correspond to the possible
tile colors.  A pair of tiles $(t_i,t_j)$ is allowed by the tiling rules iff the corresponding
$\ket{t_i t_j} \bra{t_i t_j}$ term of the Hamiltonian is $0$, so that allowed tilings have
$0$ total energy, whereas a forbidden tiling has energy at least $1$.

\section{Problems and Results}

The paper contains a variety of different but related results involving classical tiling and quantum Hamiltonian problems with translational invariance.  In this section, we will summarize the different variants, giving the proofs and more detailed discussion of each variant in later sections.  While there are certain recurring techniques, the details of the different cases vary considerably.  As a consequence, the proof of each major result is largely self-contained.

\begin{definition}
{\bf \tiling}

\noindent
{\bf Problem Parameters:} A set of tiles $T = \{t_1,\ldots,t_m\}$.
A set of horizontal constraints $H \subseteq T \times T$ such that if
$t_i$ is placed to the left of $t_j$, then it must be the case that
$(t_i,t_j) \in H$. A set of vertical constraints $V \subseteq T \times T$ such that if
$t_i$ is placed below $t_j$, then it must be the case that
$(t_i,t_j) \in V$. A designated tile $t_1$ that must be placed in the four corners of the grid.

\noindent
{\bf Problem Input:} Integer $N$, specified in binary.

\noindent
{\bf Output:} Determine whether there is a valid tiling of an $N \times N$ grid.
\end{definition}

\begin{theorem}
\label{th:unweighted-tiling}
\tiling\ is \NEXP-complete.
\end{theorem}

We give the proof in section~\ref{sec:tiling}. The basic idea is that the
corner tiles are used to create a border around the perimeter of the grid which
allows us to implement special rules at the top and bottom rows.
The interior of the grid is tiled in two layers, each of which implements
the action of a Turing machine. The first TM proceeds from top to bottom
on layer 1 and the second proceeds from bottom to top on layer 2.
The first TM takes no input and acts
as a binary counter for $N$ steps. The bottom row
of  the first layer then
holds a binary number that is $\Theta(N^{1/k})$. The rules for the lower boundary are then used
to copy the output from the binary counter to the bottom row of layer 2, which
acts as the input to a generic non-deterministic Turing machine.
The rules for the top boundary check whether
the final configuration on layer 2 is an accepting  state.

Note that it is important that we chose to have the input $N$ provided in binary.
If it were instead given in unary, there would only be one instance per problem size,
and the problem would be trivially in $\Pclass$/poly.  Thus, in order to prove
a meaningful hardness result, we are forced to move up the exponential hierarchy
and prove the problem is $\NEXP$-complete rather than $\NP$-complete.

A common convention for this tiling problem
is to only specify the boundary condition tile in a single
corner of the grid. This does not work in our case, so we instead use specified
tiles in all four corners to mark out the boundary of the grid to be tiled.
We have considered other versions of the classical translationally-invariant
tiling problem to understand to what extent the precise definition of the
problem is important.  The boundary conditions, as noted above, are a critical component.
As well as fixing the tiles at the $4$ corners of the square, we have
considered periodic boundary conditions (so we are actually tiling a torus)
and open boundary conditions, where any tile is allowed at the edges of the
square.  The case of periodic boundary conditions is particularly interesting because
it is truly translationally invariant, unlike our usual formulation
where the boundaries break the translational symmetry.  We show this
case is also hard, but with a more complicated reduction than in our
standard \tiling\ problem.
Another variant is to make the problem more similar to the quantum
Hamiltonian problem by assigning a cost to any pair of adjacent tiles, and
allowing the costs to be different from $0$ or $1$.  This is like a weighted version
of tiling and corresponds to a
Hamiltonian which is diagonal in the standard basis but does not have any
other constraints.

We have also considered problems with additional symmetry
beyond the translational invariance.
If we have {\em reflection symmetry},
then if $(t_i, t_j) \in H$, then $(t_j, t_i) \in H$ as well, and if $(t_i, t_j) \in V$,
then $(t_j, t_i) \in V$ also.  That is, the tiling constraints to the left and right are
the same, as are the constraints above and below.  However, if we only have reflection
symmetry, there can still be a difference between the horizontal and vertical directions.
If we have {\em rotation symmetry}, we have reflection symmetry and also $(t_i, t_j) \in H$
iff $(t_i, t_j) \in V$.  Now the direction does not matter either.
These additional symmetries are well motivated from a physical point of view since many
physical systems exhibit reflection or rotation symmetry.
Finally, we have studied the one-dimensional
version of the problem as well as the two-dimensional version.  See Table~\ref{table:variants}
for a summary of our results. Proofs are given in Section \ref{sec:variants}.

Some variants of \tiling\ we consider are easy but in a strange non-constructive
sense in that
there exists $N_0 \in \integer^+ \cup \{\infty\}$ such
that if $N < N_0$, there exists a valid tiling, and if $N \geq N_0$,
then there does not exist a tiling (or sometimes the other way around).  However, $N_0$ is uncomputable as
a function of $(T, H, V)$.
These cases are denoted as ``$\Pclass$, uncomputable'' in Table~\ref{table:variants} (with a question
mark if we have not been able to prove whether $N_0$ is computable or not).
Note
that this does not exclude the existence of a (potentially slower) algorithm to solve
particular instances; indeed, all the classes in Table~\ref{table:variants} are included
in $\NEXP$.  For these variants, we know that there is an efficient algorithm, so a
hardness result can be ruled out, but since the algorithm depends on an uncomputable
parameter, it may be that the problem remains hard in practice.

\begin{table*}[!t]
\begin{centering}
\begin{tabular}{rllll}
                            & $2$-D, no symmetry    & $2$-D, reflection sym.        & $2$-D, rotation sym.     & $1$-D \\
\multicolumn{1}{l}{BC on all corners} &             &                               &                          &   \\
            unweighted      &   $\NEXP$-complete    & $\Pclass$, uncomputable?      & $\Pclass$                & $\Pclass$ \\
            weighted        &   $\NEXP$-complete    & $\NEXP$-complete              & $\Pclass$                & $\Pclass$ \\
\multicolumn{1}{l}{BC on $0$ or $1$ corner \ \ } &  &                               &                          &   \\
            unweighted      & $\Pclass$, uncomputable &         $\Pclass$           & $\Pclass$                & $\Pclass$ \\
            weighted        &   $\NEXP$-complete    & $\NEXP$-complete              & $\Pclass$         	   & $\Pclass$ \\
\multicolumn{1}{l}{Periodic BC}       &             &                               &                          &   \\
            unweighted      &   $\NEXP$-complete*   & $\Pclass$, uncomputable?      & $\Pclass$                & $\Pclass$ \\
            weighted        &   $\NEXP$-complete    & $\NEXP$-complete              & $\Pclass$                & $\Pclass$ \\
\end{tabular}
\caption{Summary of the variants of \tiling.  ``BC'' is short for ``boundary condition.''
``$\Pclass$, uncomputable'' means that the associated problem is in $\Pclass$, but
an essential parameter of the efficient algorithm we found is uncomputable (with a question mark if we are not sure whether it is uncomputable).  ``$\NEXP$-complete*'' means complete under an expected poly-time randomized reduction or a deterministic polyspace reduction.}
\label{table:variants}
\end{centering}
\end{table*}

Now we turn to the quantum problem.  First we need to define the class $\QMAEXP$.
It will be a bit more convenient to work with quantum Turing Machines than quantum
circuits.  The definition is the same as $\QMA$ except that the witness and the
length of the computation for the verifier (which is a quantum Turing Machine)
can be of size $2^{n^k}$ on an input of length $n$.

\begin{definition}
A language $L$ is in $\QMAEXP$ iff there exists a $k$ and a Quantum
Turing Machine $M$ such that
for each instance $x$ and any $\ket{\psi}$ on $O(2^{|x|^k})$ qubits,
on input $(x,\ket{\psi})$,
$M$ halts in $O(2^{|x|^k})$ steps. Furthermore,
 (a) if $x \in L_{\rm yes}$, $\exists\, \ket{\psi}$
 such that
$M$ accepts $(x,\ket{\psi})$ with probability at least $2/3$.
 (b) if $x
\in L_{\rm no}$, then $\forall\, \ket{\psi}$,
$M$ accepts $(x,\ket{\psi})$ with
probability at most $1/3$.
\end{definition}

\begin{definition}
{\bf \rTIH{r} (Translationally-Invariant Hamiltonian)}

\noindent
{\bf Problem Parameter:} $r$ Hamiltonian terms $H_1, \ldots, H_r$
that each operate on two finite dimensional particles, specified with a constant number of bits.
Two polynomials $p$ and $q$.

\noindent
{\bf Problem Input:} Integer $N$, specified in binary.

\noindent
{\bf Promise:} Consider an $N^r$-dimensional grid of particles and the Hamiltonian resulting from applying $H_i$
to each pair of neighboring particles along dimension $i$. The ground state energy of this system is either
at most $p(N)$ or at least $p(N)+1/q(N)$.

\noindent
{\bf Output:} Determine whether the ground state energy of the system is at most $p(N)$
or at least $p(N)+1/q(N)$.
\end{definition}

The following theorem is the main result for the quantum case and
shows that the problem will likely be, in general,  difficult.
Note that typically, one is willing to spend time that is polynomial in the
size of the system (which is in turn exponential in the size of the input).
It follows from the result that if there is a quantum algorithm that finds the ground state energy
in time that is polynomial in the size of the system then
$\QMAEXP = \BQEXP$.

\begin{theorem}
\label{th:qmaexp}
\rTIH{1} is $\QMAEXP$-complete.
\end{theorem}

The theorem immediately implies that \rTIH{r} is $\QMAEXP$-complete
for any $r \ge 1$ since we can always take $H_i=0$ for $i \ge 2$ which results in a system of
$N^{r-1}$ independent lines with $N$ particles.  We prove theorem~\ref{th:qmaexp} in
section~\ref{sec:quantum}.

As is common in $\QMA$-completeness results,
the construction for Theorem \ref{th:qmaexp} creates a Hamiltonian whose ground state is
a uniform superposition of a sequence of states which represent a particular process.
A portion of the Hilbert space for the system holds a clock which allows us to control
the length of the sequence and ensures that the states in the sequence are mutually orthogonal.
That is, the $t^{th}$ state has  the form $\ket{\phi_t}\ket{t}$, where $\ket{\phi_t}$
is the $t^{th}$ state in the process we wish to simulate, and the overall ground state
will be $\sum_t \ket{\phi_t}\ket{t}$.
The size of the system controls the number of time steps for which the clock runs.
In the case of the construction presented here, the process consists of two main phases.
The first phase is the execution of a Turing machine which simply
increments a binary counter. The clock ensures that this TM is run for $N-3$ steps
after which a number that is $\Theta(N^{1/k})$ is encoded in binary in the state of the
quantum system. This state is then used as  the input to an arbitrary quantum Turing
machine which is executed in the second phase. This QTM implements a verifier which
is also allowed a quantum witness of length $\Theta(N)$. Finally, there is an energy
term which penalizes any non-accepting computation of the verifier.

We can also consider variants of \rTIH{1}.  If we use periodic boundary conditions instead
of open boundary conditions, we get the same result (see section~\ref{sec:quantumperiodic}).
If we add reflection symmetry, the problem also remains $\QMAEXP$-complete with open or
periodic boundary conditions (see section~\ref{sec:quantumreflection}).

Another case of particular physical interest is the infinite chain. Of course, if the
Hamiltonian term is fixed and the chain is infinite, the ground state energy of the system is a single number and there is
not a computational problem with an infinite family of inputs to study.
Instead, we look at a variation where the two-particle Hamiltonian term is
the input to  the problem and the size of the input is the number of bits required
to specify the term. We then ask: if this term is applied to each pair of neighboring
particles in an infinite chain, what is the ground energy per particle?

\begin{definition}
{\bf \ITIH\ (Infinite-Translationally-Invariant Hamiltonian)}

\noindent
{\bf Problem Parameter:}
Three polynomials $p$, $q$, and $r$. $d$, the dimension of a particle.

\noindent
{\bf Problem Input:} A Hamiltonian $H$ for two $d$-dimensional particles, with matrix entries which are multiples of $1/r(N)$.

\noindent
{\bf Promise:} Consider an infinite chain
of particles and the Hamiltonian resulting from applying $H$
to each pair of neighboring particles. The ground state energy per
particle of this system is either
at most $1/p(N)$ or at least $1/p(N)+1/q(N)$.

\noindent
{\bf Output:} Determine whether the ground state energy
per particle of the system is at most $1/p(N)$
or at least $1/p(N)+1/q(N)$.
\end{definition}

We prove the following theorem:

\begin{theorem}
\label{th:iqmaexp}
\ITIH\ is $\QMAEXP$-complete.
\end{theorem}

As a corollary of Theorem~\ref{th:qmaexp}, the following version of $N$-REPRESENTABILITY~\cite{LCV} is also $\QMAEXP$-complete: Given a density matrix
$\rho$ on two $d$-state particles, is it within $\epsilon$ of a state $\rho'$ such that there exists a translationally-invariant pure state
$\ket{\psi}$ for $N$ particles arranged in a circle for which $\rho'$ is the marginal state of two adjacent particles?  $\rho$ is a parameter of the
problem, and $N$, given in binary, is the only input, as in our Hamiltonian problem.  We can reduce to this version of $N$-REPRESENTABILITY by starting
with \rTIH{1} on a circle.  Then there is always a translationally-invariant pure ground state $\ket{\psi}$ of the Hamiltonian $H$.  By breaking the
Hilbert space of two $d$-state particles up into small balls, we can get a finite set of density matrices $\rho$ to try.  For each one, if we can solve
$N$-REPRESENTABILITY, we can determine if $\rho$ can be extended to a candidate ground state $\ket{\phi}$, and if so we can determine the energy of
$\ket{\phi}$, since it is just equal to $N \mathrm{tr}(H_1 \rho)$.  Trying all possible $\rho$, we can thus find the ground state energy of $H$, up to
some $\epsilon$-dependent precision.

\section{Hardness of \tiling}
\label{sec:tiling}

The construction will make use of a
binary counter Turing machine $M_{BC}$
which
starts with a blank semi-infinite tape. The head begins in a designated start
state in the left-most position of the tape. $M_{BC}$ will generate all
binary strings in lexicographic order. More specifically, there is a
function $f: \mathbb{Z} \rightarrow \{0,1\}^*$ such that for some
constant $N_0$ and
every $N \ge N_0$, if $M_{BC}$ runs for $N$ steps, then the string $f_{BC}(N)$
will be written on the tape with the rest of the tape blank.
Moreover there are constants $c_1$ and $c_2$ such that
if $n$ is the length of the string $f_{BC}(N)$ and
$N \ge N_0$, then
$2^{c_1 n} \le N \le 2^{c_2 n}$.
We will also assume that for any binary string $x$, we can compute $N$
such that $f_{BC}(N) = x$ in time that is polynomial in the length of $x$.  In some of the variations of the problem we consider in
section~\ref{sec:variants}, we will need to put additional restrictions on $N$ (such as requiring $N$ to be odd), and in those cases, we still require
that we can find an $N$ with the appropriate restrictions such that $f_{BC}(N) = x$.

Using a standard padding argument, we can reduce any language in
$\NEXP$ to $\NTIME(2^{c_1 n})$. If $L$ is in $\NTIME(2^{n^k})$, the reduction
consists of padding an input $x$ so that its length is $|x|^k/c_1$ \cite{papa95}.
Thus, we will take an arbitrary non-deterministic Turing machine $M$ which accepts
a language $L$ in time $2^{c_1 n}$ and reduce it to \tiling.
The tiling rules and boundary conditions
will be specific to the Turing machine $M$ but will be independent
of any particular input. The reduction for Theorem \ref{th:unweighted-tiling}
then will take an input string $x$ and output
integer $N$ such that $f_{BC}(N-3)=x$. The tiling rules will have the property that
a string $x$ is in $L$ if and only if an $N \times N$ grid can be tiled according
to the tiling rules.

{\bf Proof of Theorem \ref{th:unweighted-tiling}}.
The boundary conditions for the $N \times N$ grid will be that the four corners of the grid must have a designated tile type $\tileC$.
(We actually only need to use two corners as described in Section \ref{sec:boundaryconditions}.)
First we will specify a set of boundary tiles and their constraints.
In addition to \tileC\,
there are four other kinds of boundary tiles: $\tileW$, $\tileN$, $\tileS$, $\tileE$.
We will call the rest of the tiles {\em interior} tiles.
The tiling rules for the boundary tiles are summarized in table~\ref{table:fourcorners}.
\begin{table}
\begin{centering}
\begin{tabular}{lc|cccccc}
     &             &                \multicolumn{6}{c}{Tile on right} \\
     &             & $\tileC$ & $\tileW$ & $\tileE$ & $\tileN$ & $\tileS$ & $\tilevariable$ \\
 \hline
     & $\tileC$    &     N    &    N     &    N     &     Y    &     Y    &    N   \\
Tile & $\tileW$    &     N    &    N     &    Y     &     N    &     N    &        \\
on   & $\tileE$    &     N    &    N     &    N     &     N    &     N    &    N   \\
left & $\tileN$    &     Y    &    N     &    N     &     Y    &     N    &    N   \\
     & $\tileS$    &     Y    &    N     &    N     &     N    &     Y    &    N   \\
 & $\tilevariable$ &     N    &    N     &          &     N    &     N    &        \\
\end{tabular}
\qquad
\begin{tabular}{lc|cccccc}
       &             &                \multicolumn{6}{c}{Tile on top} \\
       &             & $\tileC$ & $\tileW$ & $\tileE$ & $\tileN$ & $\tileS$ & $\tilevariable$ \\
 \hline
       & $\tileC$    &     N    &    Y     &    Y     &     N    &     N    &    N   \\
Tile   & $\tileW$    &     Y    &    Y     &    N     &     N    &     N    &    N   \\
on     & $\tileE$    &     Y    &    N     &    Y     &     N    &     N    &    N   \\
bottom & $\tileN$    &     N    &    N     &    N     &     N    &     N    &    N   \\
       & $\tileS$    &     N    &    N     &    N     &     Y    &     N    &        \\
   & $\tilevariable$ &     N    &    N     &    N     &          &     N    &        \\
\end{tabular}
\caption{The tiling rules for boundary tiles.  $\tilevariable$ represents any interior tile.  ``N'' indicates a disallowed pairing of tiles.  For two boundary tiles, ``Y'' represents an allowed pairing.  Some of the rules for interior tiles are not specified, because they depend on the specific interior tile.}
\label{table:fourcorners}
\end{centering}
\end{table}

\tileW\ will mark the left side of the grid, \tileN\ the top of the
grid, \tileS\ the bottom of the grid, and \tileE\ the right side
of the grid. (See figure~\ref{fig:fourcorners}.)
\begin{figure}
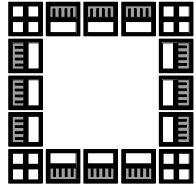

\begin{centering}
\begin{tabular}{c@{\extracolsep{0.1em}}c@{}c@{}c@{}c}
\tileC & \tileN & \tileN & \tileN & \tileC \\
\tileW & & & & \tileE \\
\tileW & & & & \tileE \\
\tileW & & & & \tileE \\
\tileC & \tileS & \tileS & \tileS & \tileC
\end{tabular}
\caption{The only allowed tiling of the sides of a $5 \times 5$ grid.}
\label{fig:fourcorners}
\end{centering}
\end{figure}
To show this, note the following facts:
\begin{itemize}
\item Nothing can go to the left of a $\tileW$ tile which means that
the only place a $\tileW$ tile could go is the left-most boundary.

\item Similarly, $\tileN$ tiles can only go in the
top row, $\tileS$ tiles can only go in the bottom row, and $\tileE$ tiles can only go in the right-most column.

\item No interior tile can border a $\tileC$ tile in any direction. Furthermore
a $\tileC$ cannot border on itself in any direction. This means that the
only possible locations for a $\tileC$ are the four corners since those
are the only places which can be surrounded by $\tileW$, $\tileN$, $\tileS$, or $\tileE$
tiles. Since the boundary conditions state that $\tileC$ tiles must go in
the corners, those are exactly the locations that will hold $\tileC$
tiles.

\item The only tiles that can go above or below a $\tileC$ tile are $\tileW$ and $\tileE$ tiles, and
$\tileE$ cannot go on the west boundary.  Thus, the tiles on the west boundary adjacent to the corners
must be $\tileW$ tiles.

\item Only $\tileW$ or $\tileC$ can go above or below a $\tileW$ tile, so the entire west boundary, except for the corners, will be $\tileW$ tiles.

\item Similar logic shows that the entire east boundary, except for the corners, will be $\tileE$ tiles.  Also, the entire south boundary, except for the corners, will be $\tileS$ tiles, and the entire north boundary, except for the corners, will be $\tileN$ tiles.

\end{itemize}

The remainder of the grid will be tiled in two layers.
The constraints on the two layers only interact at the bottom of the grid, so
we describe each layer separately. The actual type for an
interior tile is specified by a pair denoting its
layer 1 type and layer 2 type.
The bottom layer will be used to simulate the Turing machine $M_{BC}$.
The top boundary of the grid will be used to ensure that $M_{BC}$
begins with the proper initial conditions. Then the rules will enforce that
each row of the tiling going downwards advances the Turing machine $M_{BC}$
by one step. At the bottom of the grid, the output is copied onto layer 2.
Layer 2 is then used to simulate a generic non-deterministic Turing machine
on the input copied from layer 1. The lower left corner is used to
initialize the state of $M$ and the constraints enforce that
each row going upwards advances the Turing machine $M$ by one step.
Finally, the only states of $M$ that are allowed to be below an $\tileN$ tile
are accepting states.
Since each Turing machine only executes for $N-3$ steps and the grid has space
for $N-2$ tape symbols, the right end of the tape will never be reached.

Although it is well known that tiling rules are Turing complete \cite{berger66},
we review the ideas here in order to specify the details in our
construction. We will assume that the Turing machine $M$ is encoded
in a tiling that goes from bottom to top. This can easily be reversed for
$M_{BC}$ which goes from top to bottom.
The non-deterministic Turing machine $M$ is specified by
a triplet $(\Sigma, Q, \delta)$, with
designated blank symbol $\# \in \Sigma$, start state $q_0 \in Q$ and
accept state $q_A \in Q$.
There are three varieties of tiles, designated by elements of
$\Sigma$ (variety 1), $\Sigma \times Q \times \{r,l\}$ (variety 2) and
$\Sigma \times Q \times \{R,L\}$ (variety 3).  Variety 1 represents the state of the tape away from the Turing machine head.  Variety 2 represents the
state of the tape and head when the head has moved on to a location but before it has acted.  The $\{r, l\}$ symbol in a variety 2 tile tells us from which direction the head came in its last move.  Variety 3 represents the state of the tape and head after
the head has acted, and the $\{R,L\}$ symbol tells us which way the head moved.

The tiling rules are given in table~\ref{table:TM}.
\begin{table}
\begin{centering}
\begin{tabular}{lc|cccccc}
     &             &                \multicolumn{6}{c}{Tile on right} \\
     &             & $[b]$ & $[b,q',r]$ & $[b,q',l]$ & $[b,q',R]$ & $[b,q',L]$ & $\tileE$ \\
 \hline
     & $[a]$       &   Y   &     Y      &      N     &      Y     &     N      &   Y     \\
Tile & $[a,q,r]$   &   N   &     N      &      N     &      N     & If $q=q'$  &   N     \\
on   & $[a,q,l]$   &   Y   &     N      &      N     &      N     &     N      &   Y     \\
left & $[a,q,R]$   &   N   &     N      & If $q=q'$  &      N     &     N      &   N     \\
     & $[a,q,L]$   &   Y   &     N      &      N     &      N     &     N      &   Y     \\
     & $\tileW$    &   Y*  &     Y     & If $q'=q_0$ &      Y     &     N      &   Y     \\
\end{tabular}

\bigskip
\begin{tabular}{lc|ccccc}
       &             &                \multicolumn{5}{c}{Tile on top} \\
       &             & $[b]$    & $[b,q',r]$ & $[b,q',l]$                & $[b,q',R]$ & $[b,q',L]$ \\
 \hline
       & $[a]$       & If $a=b$ &  If $a=b$  &  If $a=b$, $q' \neq q_0$  &     N      &     N   \\
Tile   & $[a,q,r]$   &    N     &     N      &     N                     & If TM rule & If TM rule \\
on     & $[a,q,l]$   &    N     &     N      &     N                     & If TM rule & If TM rule \\
bottom & $[a,q,R]$   & If $a=b$ &  If $a=b$  &  If $a=b$, $q' \neq q_0$  &     N      &     N   \\
       & $[a,q,L]$   & If $a=b$ &  If $a=b$  &  If $a=b$, $q' \neq q_0$  &     N      &     N   \\
\end{tabular}
\caption{The tiling rules to simulate a Turing machine.  ``N'' indicates a disallowed pairing of tiles, and ``Y'' represents an allowed pairing.  Pairings with ``if'' statements are allowed only if the condition is satisfied.  ``If TM rule'' means the pairing is allowed only if $(a,q) \rightarrow (b,q',L/R)$ is one of the valid non-deterministic moves of the Turing machine being simulated.  The ``Y*'' entry will be modified later to get the correct starting configuration for the Turing machine. }
\label{table:TM}
\end{centering}
\end{table}
The table below gives an example of a section of tiles that
encodes the move $(a,q) \rightarrow (b,q',L)$ of the Turing machine from the bottom
row to the middle row and the move $(c,q') \rightarrow (f,q'',R)$ from
the middle row to the top row:

\[
\begin{array}{|c|c|c|}
\hline
[f,q'',R] & [b,q'',l] & [d] \\
\hline
[c,q',r] & [b,q',L] & [d] \\
\hline
[c] & [a,q,r] & [d,q,L]\\
\hline
\end{array}
\]

The lower row shows the head in the square with the $a$. The $[d,q,L]$
is from the previous TM move.
The tile $[b,q',L]$ in the middle row
enforces that the tiler is committing to executing the
step $(a,q) \rightarrow (b,q',L)$, although there may have been
other possible non-deterministic choices. The $[c,q',r]$ tile to the left of
the $[b,q',L]$ shows the new location and state of the head after the first move.
The $[b,q',L]$ tile now just acts as a $[b]$ tile for purposes of
the tiling above. The tile $[f,q'',R]$ enforces that the
tiler is committing to executing the step $(c,q') \rightarrow (f,q'',R)$.
The $[b,q'',l]$ tile to the right of
the $[f,q'',R]$ shows the new location and state of the head after the second move.

To be consistent with the horizontal tiling rules, a row must consist of some number of variety 1 tiles, along with adjacent pairs consisting of a variety 2 tile and a variety 3 tile.  The variety 2 and 3 tiles must be matched in the following sense: $q = q'$, and either we have $r$ on the left and $L$ on the right or $R$ on the left and $l$ on the right.
At the west edges of the row, we could possibly have just a single variety 2 tile by itself.
However, this can only happen if the tile is of the form $[a,q,l]$.
The rules enforce then that a \tileW\ tile can only go to the left of a variety 2 tile if
the Turing Machine state is $q_0$, which
is the starting head state of the machine.  We can assume without loss of
generality that the Turing machine never transitions back to the $q_0$ state, and
never has a transition that would move it left from the first location on the tape.
Given that we have one row of this form, the next row up must have the
same number of pairs of variety 2 and variety 3 tiles, and furthermore,
each pair must be shifted one position left or right, with the Turing
machine performing an allowed transition, as in the example.  If one row
has exactly one variety 2 tile $[a,q_0,l]$ in the leftmost position, then every row
above it in a valid tiling will also have exactly one variety 2 tile, and the $j$th
row up will correctly represent the state of the Turing machine tape and head after $j$ steps.

For our particular tiling, we will need additional rules for some pairs of boundary tiles and interior tiles.  These additional rules will be different for layers $1$ and $2$, and some of them will couple the two layers.  In addition, we will need to slightly modify some of the horizontal tiling rules.  The new and modified rules are summarized in table~\ref{table:TMboundaries}.  Also, for layer $1$, recall that we reverse ``above'' and ``below'' for the regular Turing machine simulation rules (table~\ref{table:TM}), so that time goes downwards instead of upwards.
\begin{table}
Additional layer $1$ rules:
\smallskip
\begin{centering}
\begin{tabular}{lc|ccccc}
\multicolumn{2}{l|}{Boundary Tile} &                \multicolumn{5}{c}{Adjacent interior tile} \\
\multicolumn{2}{l|}{Location} & $[a]$ & $[a,q,r]$ & $[a,q,l]$          & $[a,q,R]$ & $[a,q,L]$ \\
 \hline
Bottom & $\tileS$    &    Y        &     Y     &         Y             &    Y      &    Y     \\
Top    & $\tileN$    & If $a=\#$   &     N     & If $a=\#$ and $q=q_0$ &    N      &    N      \\
Left   & $\tileW$    & If $a \neq \#$ &  Y     & If $q=q_0$            &    Y      &    N      \\
\end{tabular}

\bigskip

\end{centering}
\noindent Additional layer $2$ rules:
\smallskip
\begin{centering}
\begin{tabular}{lc|ccccc}

\multicolumn{2}{l|}{Boundary Tile} &                \multicolumn{5}{c}{Adjacent interior tile} \\
\multicolumn{2}{l|}{Location} & $[a]$         & $[a,q,r]$  & $[a,q,l]$                            & $[a,q,R]$ & $[a,q,L]$ \\
 \hline
Bottom  & $\tileS$ & If $a$ matches layer $1$ &     N      & If $q=q_0$ and $a$ matches layer $1$ &    N      &    N      \\
Top     & $\tileN$ &   Y                      & If $q=q_A$ & If $q=q_A$                           &    Y      &    Y      \\
Left    & $\tileW$ & If $a \not\in \Sigma_{M_{BC}}$ &  Y   & If $q=q_0$                           &    Y      &    N      \\
\end{tabular}

\caption{Additional tiling rules between certain boundary and interior tiles for layers $1$ and $2$.  For layer $2$, the alphabet symbol $a \in \Sigma$ must match the alphabet symbol for the corresponding layer $1$ tile when on the bottom interior row.}
\label{table:TMboundaries}
\end{centering}
\end{table}

We would like to start out the
Turing machine $M_{BC}$ with $[\#, q_0, l]$ in the leftmost location
followed by $[\#]$ tiles.  This is ensured by the additional rules for layer $1$: The top interior row must consist of only these two types of tiles, and the leftmost location cannot be $[\#]$.  Since there cannot be any variety 3 tiles, there can only be one $[\#, q_0, l]$ in the leftmost location.
We can assume without loss of generality that
the Turing machine overwrites the leftmost $\#$ on the tape
and never writes a $\#$ there again.

The rest of the layer 1 rules just enforce the rules for the Turing machine
$M_{BC}$.

Now in order to copy the output from $M_{BC}$ to the input tape
for $M$, we restrict the kinds of tiles that can go above $\tileS$ tiles.
All the alphabet characters in the bottom interior row of layer 2 must
match the alphabet characters for layer 1.  That is, the output of $M_{BC}$ is
copied onto the input of $V$.

We also want to ensure that the starting configuration of $V$ has
only one head in the leftmost location.  To accomplish this,
we forbid a $\tileW$ to go next to an $[a]$ tile for $a \in \Sigma_{M_{BC}}$, so
the leftmost tile in the 2nd row from the bottom of layer 2 must be $[a,q_0,l]$.  Since there
can be no variety 3 tiles in the 2nd row, that must be the only variety 2 tile
in the row.  A little care must be taken to overwrite the leftmost input tape
character with something that is not in the alphabet of $M_{BC}$.
This is because we have forbidden having an $[a]$ tile to the right
of a $\tileW$ for any $a \in \Sigma_{M_{BC}}$, and this prohibition
applies to {\em all} rows.
The information encoded in the left-most tape symbol can be retained by
having a new $a'$ symbol in $\Sigma_M$ for every $a \in \Sigma_{M_{BC}}$.

Finally, the only variety 2 tiles on layer 2 which we allow below a
$\tileN$ tile must be of the form $[a,q_A, r/l]$, where $q_A$ is the accepting state.
Thus, there is a valid tiling if and only if the non-deterministic TM $M$ accepts on input $x$ in $N-3$ steps.

\section{Variants of the Classical Tiling Problem}
\label{sec:variants}

We now turn to studying variants of \tiling.  There are variety of different boundary
conditions we could consider.  We can also consider changing the absolute prohibitions
on certain adjacent tiles to a soft condition by assigning different weights to the
different adjacent pairs of tiles.  We can add additional reflection or rotation symmetry
to the tiling rules.
Also, we prove that in one dimension, \tiling\ and all the above variants are easy.

All of the hardness results are proven by reducing from \tiling, perhaps with a restriction
(usually straightforward) on the value of $N$.  For instance, when discussing WEIGHTED
\tiling\ with periodic boundary conditions, we require $N$ to be odd.  \tiling\ with odd $N$
is clearly still $\NEXP$-complete, as we can either use a universal Turing Machine $M$ that performs
some processing on the input $x$, or simply a different counting Turing Machine $M_{BC}$ that
counts more slowly, effectively ignoring the least significant bit of $N$.  The
other restrictions on $N$ we use similarly result in hard subclasses of \tiling.

To prove the variants hard, we have a main layer, which duplicates the tiling rules
given to us, with perhaps some small variations (e.g., adding extra tiles).  We also
add additional layers with tiling rules that force a structure that duplicates the
conditions of the standard \tiling\ problem.  The details of the additional layer differ
with each construction, and are provided in the discussion below.

\subsection{Boundary Conditions}
\label{sec:boundaryconditions}

Our choice of fixing the tiles in all four corners of the square is unusual.  See,
for instance, Papadimitriou~\cite{papa95}, who fixes just the tile in a single
corner.  If we had instead chosen that convention, the \tiling\ problem would
become easy, in an annoying non-constructive sense:
\begin{theorem}
Define a variant of \tiling\ where a designated tile $t_1$ is placed
in the upper right corner of the grid, and the other corners are
unconstrained.  Then there exists $N_0 \in \integer^+ \cup \{\infty\}$ such
that if $N < N_0$, there exists a valid tiling, and if $N \geq N_0$,
then there does not exist a tiling.  However, $N_0$ is uncomputable as
a function of $(T, H, V)$.
\end{theorem}

This theorem follows immediately from the observation that with this boundary
condition, a valid tiling for an $N \times N$ grid can be cropped by removing
the leftmost column and bottommost row to give a valid tiling for the
$(N-1) \times (N-1)$ grid.  We know $N_0$ must be uncomputable because the
question of whether there is an infinite tiling is uncomputable.  Still, if
we fix $(T, H, V)$, we know there exists a straightforward algorithm
to solve this variant of tiling: simply determine if $N < N_0$.  We just
do not know $N_0$, so we do not know precisely what algorithm to use.

On the other hand, if we fix boundary conditions in two corners, that is
already enough for hardness.  There are two cases to consider: when the
two corners are adjacent and when they are opposite.  The techniques used
for showing these two cases are hard are similar to those used in some other
variants, so we omit the details.  When the two corners are opposite, we can
create a boundary much like that of figure~\ref{fig:fourcorners}, but with two
new unique corner tiles.  When the two corners are adjacent, we can modify the
tiling rules given in section~\ref{sec:tiling} in order to avoid the need for the
right boundary.  We only need special tiling conditions at the right boundary to
prevent a new TM head from appearing there.  If instead, we insist that only $[\#]$
can appear to the right of $[\#]$, it is also impossible to create a new TM head on
the right end.

Another interesting case is when we have periodic boundary conditions.  That
is, we consider the top row to be adjacent to the bottom row, and the leftmost column is
adjacent to the rightmost column.  Essentially, we are tiling a torus.  This
case is particularly interesting because it is truly translationally invariant,
unlike our usual formulation, where the boundaries break the translational
invariance.
\begin{theorem}
\label{thm:periodic}
Define PERIODIC \tiling\ as a variant of \tiling\ with periodic boundary conditions
on an $N \times N$ grid.  PERIODIC \tiling\ is $\NEXP$-complete under an expected
poly-time randomized reduction or a deterministic polyspace reduction.
\end{theorem}

Since we are dealing with the class $\NEXP$, even an exponential-time reduction
is meaningful.  For instance, if there is an algorithm to solve PERIODIC \tiling\ which
takes time poly($N$), then, as a consequence of theorem~\ref{thm:periodic}, $\EXP = \NEXP$.

\begin{proof}
To prove this, we will show that, for appropriate $N$, we can introduce an effective
horizontal and vertical border at some point inside the square.  The actual location
of the borders will not be specified at all, but we will choose conditions so that there
is exactly one horizontal border and one vertical border.  Those borders will then
act like the usual edges of the grid for the standard $4$-corners boundary condition.

We will specialize to $N$ which is an odd prime.  We will add two additional layers of
tiles over those for the usual \tiling\ $\NEXP$-completeness result (section~\ref{sec:tiling}).
The new layer $1$ has $7$ different possible types of tile: \tileH, \tileV, \tileC, \tileWhite,
\tileWhiteU, \tileBlack, and \tileBlackD.  The new layer $2$ has $10$ types of tile: \tileN,
\tileS, \tileE, \tileW, \tileNW, \tileNE, \tileSW, \tileSE, \tileLight, and \tileDark.

The \tileH, \tileV, and \tileC\ tiles will create the vertical and horizontal borders.
\tileBlack\ and \tileWhite\ are used to make sure that there is at least one of each kind
of border, and the \tileBlackD\ and \tileWhiteU\ tiles will create a diagonal line within
the rectangle defined by the borders.  Layer $2$ is used to mark the directions next to
the border and make sure the diagonal line goes from the upper left corner of the
rectangle to the bottom right corner, which ensures that the rectangle is a square.  When
$N$ is prime, this means there can only be one horizontal and one vertical border.  The structure
of tiling that we would like to achieve is shown in figure~\ref{fig:periodicnoweights}.

\begin{figure}
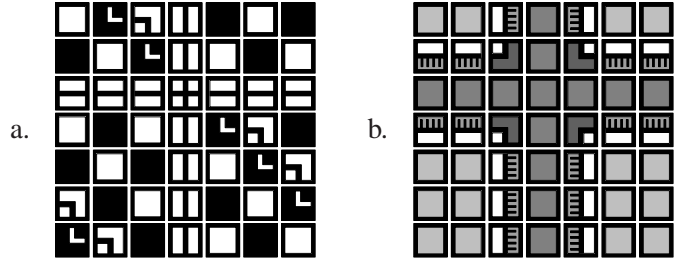

\begin{centering}
a.
\begin{tabular}{c@{\extracolsep{0.1em}}c@{}c@{}c@{}c@{}c@{}c}
\tileWhite  & \tileBlackD & \tileWhiteU & \tileV & \tileBlack  & \tileWhite  & \tileBlack \\
\tileBlack  & \tileWhite  & \tileBlackD & \tileV & \tileWhite  & \tileBlack  & \tileWhite \\
\tileH      & \tileH      & \tileH      & \tileC & \tileH      & \tileH      & \tileH \\
\tileWhite  & \tileBlack  & \tileWhite  & \tileV & \tileBlackD & \tileWhiteU & \tileBlack \\
\tileBlack  & \tileWhite  & \tileBlack  & \tileV & \tileWhite  & \tileBlackD & \tileWhiteU \\
\tileWhiteU & \tileBlack  & \tileWhite  & \tileV & \tileBlack  & \tileWhite  & \tileBlackD \\
\tileBlackD & \tileWhiteU & \tileBlack  & \tileV & \tileWhite  & \tileBlack  & \tileWhite
\end{tabular}
\quad
b.
\begin{tabular}{c@{\extracolsep{0.1em}}c@{}c@{}c@{}c@{}c@{}c}
\tileLight  & \tileLight & \tileE & \tileDark & \tileW  & \tileLight  & \tileLight \\
\tileS  & \tileS  & \tileSE & \tileDark & \tileSW  & \tileS  & \tileS \\
\tileDark      & \tileDark      & \tileDark      & \tileDark & \tileDark      & \tileDark      & \tileDark \\
\tileN  & \tileN  & \tileNE  & \tileDark & \tileNW & \tileN & \tileN \\
\tileLight  & \tileLight  & \tileE  & \tileDark & \tileW  & \tileLight & \tileLight \\
\tileLight  & \tileLight  & \tileE  & \tileDark & \tileW  & \tileLight & \tileLight \\
\tileLight  & \tileLight  & \tileE  & \tileDark & \tileW  & \tileLight & \tileLight
\end{tabular}
\caption{A possible arrangement of layers $1$ (a.) and $2$ (b.) for PERIODIC \tiling.}
\label{fig:periodicnoweights}
\end{centering}
\end{figure}

\medskip

{\bf Layer 1:~} The tiling rules for layer 1 are summarized in table~\ref{table:periodic1}.

\begin{table}
\begin{centering}
\begin{tabular}{lc|ccccccc}
     &               &                \multicolumn{7}{c}{Tile on right} \\
     &               & $\tileH$ & $\tileV$ & $\tileC$ & $\tileWhite$ & $\tileWhiteU$ & $\tileBlack$ & $\tileBlackD$ \\
 \hline
     & $\tileH$      &     Y    &    N     &    Y     &       N      &        N      &      N       &    N  \\
Tile & $\tileV$      &     N    &    N     &    N     &       Y      &        N      &      Y       &    Y   \\
on   & $\tileC$      &     Y    &    N     &    N     &       N      &        N      &      N       &    N   \\
left & $\tileWhite$  &     N    &    Y     &    N     &       N      &        N      &      Y       &    Y   \\
     & $\tileWhiteU$ &     N    &    Y     &    N     &       N      &        N      &      Y       &    N   \\
     & $\tileBlack$  &     N    &    Y     &    N     &       Y      &        N      &      N       &    N  \\
     & $\tileBlackD$ &     N    &    Y     &    N     &       N      &        Y      &      N       &    N  \\
\end{tabular}

\medskip

\begin{tabular}{lc|ccccccc}
      &               &                \multicolumn{7}{c}{Tile on top} \\
      &               & $\tileH$ & $\tileV$ & $\tileC$ & $\tileWhite$ & $\tileWhiteU$ & $\tileBlack$ & $\tileBlackD$ \\
 \hline
      & $\tileH$      &    N     &     N    &     N    &      Y       &      N        &       Y      &    Y  \\
Tile  & $\tileV$      &    N     &     Y    &     Y    &      N       &      N        &       N      &    N   \\
on    & $\tileC$      &    N     &     Y    &     N    &      N       &      N        &       N      &    N   \\
bottom & $\tileWhite$ &    Y     &     N    &     N    &      N       &      N        &       Y      &    Y   \\
      & $\tileWhiteU$ &    Y     &     N    &     N    &      N       &      N        &       Y      &    N   \\
      & $\tileBlack$  &    Y     &     N    &     N    &      Y       &      N        &       N      &    N  \\
      & $\tileBlackD$ &    Y     &     N    &     N    &      N       &      Y        &       N      &    N  \\
\end{tabular}
\caption{The tiling rules for layer 1 for PERIODIC \tiling.}
\label{table:periodic1}
\end{centering}
\end{table}

The rules for $\tileV$, \tileH, and \tileC\ imply that if we have a \tileV\ anywhere in a column, that column must only
contain \tileV\ and \tileC\ tiles, and if we have a \tileH\ anywhere in a row, that row can only
contain \tileH\ and \tileC\ tiles.  Furthermore, a \tileC\ tile must be surrounded by \tileV\ tiles
above and below and \tileH\ tiles to the left and right.
Thus, layer $1$ contains some number of vertical and horizontal
lines composed of \tileV\ and \tileH\ tiles intersecting at \tileC\ tiles.  No two horizontal or
vertical lines can be adjacent.

The next set of rules for layer 1 enforce that
the space between the lines must be filled by a checkerboard
of \tileWhite\ or \tileWhiteU\ alternating with \tileBlack\ or \tileBlackD\ tiles.  We ensure
this by forbidding \tileWhite\ and \tileWhiteU\ from being adjacent to themselves or each other
in any direction, and forbidding  \tileBlack\ and \tileBlackD\ from being adjacent to themselves or
each other in any direction. Since $N$ is odd, this ensures that it is not possible to tile the
entire torus with the checkerboard pattern and there must be at least one horizontal and one vertical
line.

Furthermore, if we ever
have a \tileBlackD\ tile anywhere, we want to be forced to have a diagonal line of \tileBlackD\
tiles continuing to the upper left and lower right, with a diagonal line of \tileWhiteU\ tiles
next to it (above and to the right), with both lines ending at a vertical or horizontal line.
We enforce this by requiring that a \tileWhiteU\ tile must have \tileBlackD\ both below it and to its left.
Above a \tileBlackD\ tile
we can have only \tileH\ or \tileWhiteU, and to the right of a \tileBlackD\ tile, we
must have either \tileV\ or \tileWhiteU.

\medskip

{\bf Layer 2:~}  The tiling rules for layer 2 are summarized in table~\ref{table:periodic2}.

\begin{table}
\begin{centering}
\begin{tabular}{lc|cccccccccc}
     &               &                \multicolumn{10}{c}{Tile on right} \\
     &               & $\tileN$ & $\tileS$ & $\tileE$ & $\tileW$ & $\tileNW$ & $\tileNE$ & $\tileSW$ & $\tileSE$ & $\tileLight$ & $\tileDark$ \\
 \hline
     & $\tileN$      &    Y     &    N     &     N    &     N    &     N     &     Y     &     N     &     N     &      N       &    N  \\
     & $\tileS$      &    N     &    Y     &     N    &     N    &     N     &     N     &     N     &     Y     &      N       &    N  \\
     & $\tileE$      &    N     &    N     &     N    &     N    &     N     &     N     &     N     &     N     &      N       &    Y  \\
Tile & $\tileW$      &    N     &    N     &     N    &     N    &     N     &     N     &     N     &     N     &      Y       &    N  \\
on   & $\tileNW$     &    Y     &    N     &     N    &     N    &     N     &     N     &     N     &     N     &      N       &    N  \\
left & $\tileNE$     &    N     &    N     &     N    &     N    &     N     &     N     &     N     &     N     &      N       &    Y  \\
     & $\tileSW$     &    N     &    Y     &     N    &     N    &     N     &     N     &     N     &     N     &      N       &    N  \\
     & $\tileSE$     &    N     &    N     &     N    &     N    &     N     &     N     &     N     &     N     &      N       &    Y  \\
     & $\tileLight$  &    N     &    N     &     Y    &     N    &     N     &     N     &     N     &     N     &      Y       &    N  \\
     & $\tileDark$   &    N     &    N     &     N    &     Y    &     Y     &     N     &     Y     &     N     &      N       &    Y  \\
\end{tabular}

\medskip

\begin{tabular}{lc|cccccccccc}
       &               &                \multicolumn{10}{c}{Tile on top} \\
       &               & $\tileN$ & $\tileS$ & $\tileE$ & $\tileW$ & $\tileNW$ & $\tileNE$ & $\tileSW$ & $\tileSE$ & $\tileLight$ & $\tileDark$ \\
 \hline
       & $\tileN$      &    N     &     N    &     N    &     N    &     N     &      N    &     N     &     N     &      N       &    Y  \\
       & $\tileS$      &    N     &     N    &     N    &     N    &     N     &      N    &     N     &     N     &      Y       &    N  \\
       & $\tileE$      &    N     &     N    &     Y    &     N    &     N     &      Y    &     N     &     N     &      N       &    N  \\
Tile   & $\tileW$      &    N     &     N    &     N    &     Y    &     Y     &      N    &     N     &     N     &      N       &    N  \\
on     & $\tileNW$     &    N     &     N    &     N    &     N    &     N     &      N    &     N     &     N     &      N       &    Y  \\
bottom & $\tileNE$     &    N     &     N    &     N    &     N    &     N     &      N    &     N     &     N     &      N       &    Y  \\
       & $\tileSW$     &    N     &     N    &     N    &     Y    &     N     &      N    &     N     &     N     &      N       &    N  \\
       & $\tileSE$     &    N     &     N    &     Y    &     N    &     N     &      N    &     N     &     N     &      N       &    N  \\
       & $\tileLight$  &    Y     &     N    &     N    &     N    &     N     &      N    &     N     &     N     &      Y       &    N  \\
       & $\tileDark$   &    N     &     Y    &     N    &     N    &     N     &      N    &     Y     &     Y     &      N       &    Y  \\
\end{tabular}
\caption{The tiling rules for layer 2 for PERIODIC \tiling.}
\label{table:periodic2}
\end{centering}
\end{table}

We wish \tileN\ to mark the top side of the rectangles delineated by the layer $1$
\tileH\ and \tileV\ tiles. \tileE\ will mark the right side, \tileW\ the left side,
and \tileS\ the bottom side, with \tileNW, \tileNE, \tileSW, and \tileSE\ marking
the upper left, upper right, lower left, and lower right corners, respectively.  \tileLight\
will be inside the rectangle and \tileDark\ will go over the layer $1$ \tileH\ and \tileV\
tiles.

Looking at the horizontal tiling rules, we see that if we have a \tileN\ anywhere, there must be a
horizontal line of \tileN\ that either goes all the way around the grid (due to the periodic boundary
conditions) or ends at \tileNW\ on the left and \tileNE\ on the right.  If it ends, then, because of the vertical tiling rules, below
the \tileNW\ tile is a vertical line of \tileW\ tiles which end at a \tileSW\ tile, and below the
\tileNE\ tile is a vertical line of \tileE\ tiles which end at a \tileSE\ tile.  The \tileSE\
and \tileSW\ tiles must be in the same row, and between them is a line of \tileS\ tiles, forming a closed rectangle.
Following
this line of logic for the other tiles, we find that layer $2$ must be a collection of rectangles,
horizontal stripes, and vertical stripes.

The remaining rules enforce that we have \tileDark\ on the outside of the rectangles and
\tileLight\ on the inside. That is, if any of the border tiles has a solid color on one edge,
it must be adjacent on that side to a matching solid color tile.
These rules imply that inside the
rectangle delineated by the \tileN, \tileS, \tileE, and \tileW\ tiles are only \tileLight\ tiles,
and immediately outside it are \tileDark\ tiles.
Each rectangle is outlined by \tileN, \tileS, \tileE, and \tileW\
tiles on the sides, with \tileNW, \tileNE, \tileSW, and \tileSE\ on the corners, and is full of
\tileLight\ inside.  A vertical stripe has \tileW\ on its left and \tileE\ on its right, and
a horizontal stripe has \tileN\ above it and \tileS\ below it.  Both vertical and horizontal stripes
have only \tileLight\ inside them.  The rectangles and stripes cannot be adjacent and are separated
by \tileDark\ tiles.

\medskip

{\bf Interactions between layers 1 and 2:~}  The layer $1$ and layer $2$ tiles must pair up as indicated in
table~\ref{table:periodiccompatibility}.

\begin{table}
\begin{centering}
\begin{tabular}{lc|cccccccccc}
        &               &                \multicolumn{10}{c}{Layer 2 tile} \\
        &               & $\tileN$ & $\tileS$ & $\tileE$ & $\tileW$ & $\tileNW$ & $\tileNE$ & $\tileSW$ & $\tileSE$ & $\tileLight$ & $\tileDark$ \\
 \hline
        & $\tileH$      &    N     &    N     &    N     &     N    &     N     &     N     &     N     &     N     &      N       &    Y  \\
        & $\tileV$      &    N     &    N     &    N     &     N    &     N     &     N     &     N     &     N     &      N       &    Y  \\
Layer 1 & $\tileC$      &    N     &    N     &    N     &     N    &     N     &     N     &     N     &     N     &      N       &    Y  \\
tile    & $\tileWhite$  &    Y     &    Y     &    Y     &     Y    &     N     &     Y     &     Y     &     N     &      Y       &    N  \\
        & $\tileWhiteU$ &    Y     &    Y     &    Y     &     Y    &     N     &     Y     &     Y     &     N     &      Y       &    N  \\
        & $\tileBlack$  &    Y     &    Y     &    Y     &     Y    &     N     &     Y     &     Y     &     N     &      Y       &    N  \\
        & $\tileBlackD$ &    N     &    N     &    N     &     N    &     Y     &     N     &     N     &     Y     &      Y       &    N  \\
\end{tabular}
\caption{The allowed pairs of layer $1$ and layer $2$ tiles for PERIODIC \tiling.}
\label{table:periodiccompatibility}
\end{centering}
\end{table}

These conditions tell us that the horizontal and vertical lines on layer $1$ formed by \tileH, \tileV, and \tileC\ must match exactly the locations of the \tileDark\ tiles on layer $2$.  This implies that in fact layer $2$
can only contain rectangles which must be lined up with the space between the horizontal and
vertical lines on layer $1$.  On layer $1$, each rectangle delineated by the horizontal and
vertical lines must have a \tileBlackD\ in the upper left corner, which starts a diagonal line
of \tileBlackD\ tiles extending towards to the bottom right.  Since it must end at the border
of the rectangle, but a \tileBlackD\ tile cannot be in the same spot as a layer $2$ \tileE\
or \tileS\ tile, the only place the diagonal line can end is at the lower right corner.  Thus,
the rectangle must actually be a square.  See figure~\ref{fig:periodicnoweights} for an example
of an allowed tiling.

However, the only way for all the rectangles formed by the horizontal and vertical layer $1$
lines to be squares is if the spacing between them is equal.  They then form $M \times M$
squares arrayed in a $k \times k$ grid for a total of $k^2$ squares.  It follows that $N = k (M+1)$.
But when $N$ is prime, then $k$ must be $1$. ($M$ cannot be $0$ since the horizontal and
vertical lines cannot be adjacent.)  Thus, the only allowed tiling is to produce a single
$(N-1) \times (N-1)$ square on layer $1$.  We can then consider some data layers of
tiles which implement the rules from section~\ref{sec:tiling}.  The layer $2$ \tileNW, \tileNE,
\tileSW, and \tileSE\ tiles mark out the corner of the square, so we can put a condition on
the data layers that enforce the corner boundary conditions on those locations.

\medskip

We do need to be careful of one aspect, however, since we are now restricted to a size of
square which is $1$ less than a prime number.  For any input $x$, we need to choose a prime
$N$.  $N$ should be not much bigger than $x$ ($\log N = \text{poly}(\log x)$), and it must
be possible for the universal TM implemented by the data layers to deduce $x$ in a reasonable
time.  We will show a method of finding a prime $N$ such that
the $1/3$ most significant bits represent $x$.
Let $n_0 (x) = 2 \lceil \log x \rceil$ (that is, basically twice the
number of bits in $x$) and let $N_0 (x)= x 2^{n_0(x)}$ (that is, the binary
expansion of $N_0(x)$ is that of $x$ followed by $(n_0(x) - 1)$ $0$s).
We want to show that there exists a prime number in the range $[N_0(x), N_0(x+1))$.
In fact, we will show that there is a prime in a somewhat narrower range so that it
can be found by exhaustive search in polynomial space. Furthermore, the primes in
this range are sufficiently plentiful that the expected time to find
a prime number by random selection will be polynomial.

It is known that there exists a constant $\theta < 1$ such that
\begin{equation}
\lim_{y \rightarrow \infty} [\pi(y+y^\theta) - \pi(y)] = \frac{y^\theta}{\log y},
\label{eq:primedensity}
\end{equation}
where $\pi(y)$ is the number of primes less than or equal to $y$~\cite{ingham}.
(For our purposes, it is sufficient to take $\theta = 2/3$, but smaller $\theta$ is
possible.)
That is, for sufficiently large $y$, the number of primes in the interval $[y, y+y^\theta]$
is approximately $1/\log y$ times the size of the interval.%
\footnote{This is the result one might expect from the Prime Number Theorem, but that
theorem is not strong enough, as it is compatible with having a large interval with a
low density of primes.}
Then
\begin{equation}
N_0(x)^\theta = x^\theta 2^{\theta n_0(x)} \leq 2^{\theta (\lceil \log x \rceil + n_0(x))} \leq 2^{n_0(x)}.
\end{equation}

In particular, it follows that $N_0(x) + N_0(x)^\theta \leq N_0 (x+1)$.  If we can find
a prime $N$ in the interval $I(x) = [N_0(x), N_0(x) + N_0(x)^\theta)$, the TM can thus easily
deduce $x$ by looking at the most significant bits.  There is at least one prime in $I(x)$, so we
can certainly find one with an exhaustive search, which can be done with polynomial space.
Furthermore, by (\ref{eq:primedensity}), if we choose a random $N \in I(x)$, for sufficiently
large $x$, there is a probability about $1/\log(N_0(x))$ that $N$ is prime, which we can
verify in $\text{poly}(\log x)$ time.  Thus, we get a randomized reduction to PERIODIC \tiling\
which runs in expected time $\text{poly}(\log x)$.

\end{proof}

\subsection{Weighted Constraints}
\label{sec:weighted}

Now we consider the case where the constraints count different amounts.

\begin{definition}
{\bf WEIGHTED \tiling}

\noindent
{\bf Problem Parameters:} A set of tiles $T = \{t_1,\ldots,t_m\}$.
A set of horizontal weights $w_H: T \times T \rightarrow \integer$, such that if
$t_i$ is placed to the left of $t_j$, there is a contribution of $w_H (t_i, t_j)$ to
the total cost of the tiling.  A set of vertical weights
$w_V: T \times T \rightarrow \integer$, such that if
$t_i$ is placed below $t_j$, there is a contribution of $w_V (t_i, t_j)$ to
the total cost of the tiling. A polynomial $p$.  Boundary conditions (a tile
to be placed at all four corners, open boundary conditions, or periodic
boundary conditions).

\noindent
{\bf Problem Input:} Integer $N$, specified in binary.

\noindent
{\bf Output:} Determine whether there is a tiling of an $N \times N$ grid such
that the total cost is at most $p(N)$.
\end{definition}

Note that we can shift all the weights by a constant $r$: $w_H' (t_i, t_j) = w_H (t_i, t_j) - r$, $w_V' (t_i, t_j) = w_V (t_i, t_j) - r$.  This has the effect of shifting the total cost by $2rN(N-1)$.  Therefore, we may assume without loss of generality that $p(N) = o(N^2)$.

\begin{theorem}
\label{thm:weighted}
WEIGHTED \tiling\ is $\NEXP$-complete, with any of the three choices of boundary conditions.
\end{theorem}

\begin{proof}

The case of the four-corners boundary condition follows immediately from the result for unweighted \tiling.

{\bf WEIGHTED \tiling\ with open boundary conditions:}
It is sufficient to take all the weights to be $0$, $+2$, $+4$, or $-1$,
and $c=-4$.  We will reduce from \tiling; the main point is to show we can fix
the boundary conditions on the corners.  To do so, we take the \tiling\ instance
and add a new layer of tiles consisting of five special types of tile, \tileNE, \tileNW, \tileSW, \tileSE, and \tileWhite.
The weights are summarized in table~\ref{table:weightedopen}.
\begin{table}
\begin{centering}
\begin{tabular}{lc|rrrrr}
     &               &                \multicolumn{5}{c}{Tile on right} \\
     &               & $\tileNW$ & $\tileNE$ & $\tileSW$ & $\tileSE$ & $\tileWhite$ \\
 \hline
     & $\tileNW$     &     +4    &     +4    &     +4    &     +4    &      -1    \\
Tile & $\tileNE$     &     +4    &     +4    &     +4    &     +4    &      +2    \\
on   & $\tileSW$     &     +4    &     +4    &     +4    &     +4    &      -1    \\
left & $\tileSE$     &     +4    &     +4    &     +4    &     +4    &      +2    \\
     & $\tileWhite$  &     +2    &     -1    &     +2    &     -1    &       0    \\
\end{tabular}
\qquad
\begin{tabular}{lc|rrrrr}
       &               &                \multicolumn{5}{c}{Tile on top} \\
       &               & $\tileNW$ & $\tileNE$ & $\tileSW$ & $\tileSE$ & $\tileWhite$ \\
 \hline
       & $\tileNW$     &    +4     &    +4     &    +4     &    +4     &     +2     \\
Tile   & $\tileNE$     &    +4     &    +4     &    +4     &    +4     &     +2     \\
on     & $\tileSW$     &    +4     &    +4     &    +4     &    +4     &      0     \\
bottom & $\tileSE$     &    +4     &    +4     &    +4     &    +4     &      0     \\
       & $\tileWhite$  &     0     &     0     &    +2     &    +2     &      0     \\
\end{tabular}
\caption{The tiling weights for WEIGHTED \tiling\ with open boundary conditions.}
\label{table:weightedopen}
\end{centering}
\end{table}

Suppose we have any configuration which contains the tile \tileNW, and suppose we replace that
tile by \tileWhite, leaving the rest of the configuration the same.  If there is any tile to the
left of the replaced tile, the cost of that edge decreases by at least $2$.  Similarly, if there
is any tile above the replaced tile, the cost of the above edge decreases by at least $2$.  The only way
for the overall cost of the tiling to increase is if the replaced \tileNW\ tile was located in the upper left
corner.  Similarly, if we replace \tileNE\ by \tileWhite, the cost decreases unless possibly the \tileNE\
tile is located in the upper right corner.  Similarly, we can more cheaply replace \tileSW\ and \tileSE\
by \tileWhite\ tiles unless the \tileSW\ and \tileSE\ tiles are located in the lower left and lower right
corners, respectively.  Thus, the cheapest tilings have \tileWhite\ tiles everywhere but the corners,
and it is easy to verify that the optimal tiling of the new layer has total cost $-4$, with \tileNE, \tileNW, \tileSW, and \tileSE\ in
the proper corners, and \tileWhite\ elsewhere.  Figure~\ref{fig:openBCweighted} shows the optimal
configuration of the new layer.
\begin{figure}
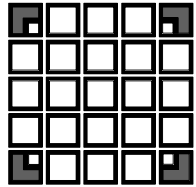

\begin{centering}
\begin{tabular}{c@{\extracolsep{0.1em}}c@{}c@{}c@{}c}
\tileNW & \tileWhite & \tileWhite & \tileWhite & \tileNE \\
\tileWhite & \tileWhite & \tileWhite & \tileWhite & \tileWhite \\
\tileWhite & \tileWhite & \tileWhite & \tileWhite & \tileWhite \\
\tileWhite & \tileWhite & \tileWhite & \tileWhite & \tileWhite \\
\tileSW & \tileWhite & \tileWhite & \tileWhite & \tileSE \\
\end{tabular}
\caption{The preferred arrangement of the new layer in WEIGHTED \tiling\ with open boundary conditions.}
\label{fig:openBCweighted}
\end{centering}
\end{figure}

Finally, we insist that \tileNE, \tileNW, \tileSW, or \tileSE\
in the special layer must correspond to the usual corner tile $t_1$ in the main
layer.  Then the optimal tiling overall has the main layer constrained in exactly the way it would be
in the standard \tiling\ problem.

\medskip

{\bf WEIGHTED \tiling\ with periodic boundary conditions:} We can of course
apply Theorem~\ref{thm:periodic}, but when we allow weights, there is a simpler solution
which avoids the caveats about the reduction.

We now consider only odd $N$, not necessarily prime.  We will use the weights $0$, $+1$, and $+3$,
and set $c=+2$.  Thus, to have a good enough tiling, we can have two pairings with weight $+1$,
and all the others must have weight $0$.  The weights we use are summarized in table~\ref{table:weightedperiodic}.

\begin{table}
\begin{centering}
\begin{tabular}{lc|rrrrr}
       &               &                \multicolumn{5}{c}{Tile on right} \\
       &               & $\tileWhite$ & $\tileBlack$ & $\tileH$ & $\tileV$ & $\tileC$ \\
 \hline
       & $\tileWhite$  &     +3       &       0      &    +3    &     0    &    +3    \\
Tile   & $\tileBlack$  &      0       &      +3      &    +3    &     0    &    +3    \\
on     & $\tileH$      &     +3       &      +3      &     0    &    +3    &    +1    \\
left   & $\tileV$      &      0       &       0      &    +3    &    +3    &    +3    \\
       & $\tileC$      &     +3       &      +3      &    +1    &    +3    &    +3    \\
\end{tabular}
\qquad
\begin{tabular}{lc|rrrrr}
       &               &                \multicolumn{5}{c}{Tile on top} \\
       &               & $\tileWhite$ & $\tileBlack$ & $\tileH$ & $\tileV$ & $\tileC$ \\
 \hline
       & $\tileWhite$  &      +3      &       0      &     0    &    +3    &    +3    \\
Tile   & $\tileBlack$  &       0      &      +3      &     0    &    +3    &    +3    \\
on     & $\tileH$      &       0      &       0      &    +3    &    +3    &    +3    \\
bottom & $\tileV$      &      +3      &      +3      &    +3    &     0    &     0    \\
       & $\tileC$      &      +3      &      +3      &    +3    &     0    &    +3    \\
\end{tabular}
\caption{The tiling weights for WEIGHTED \tiling\ with periodic boundary conditions.}
\label{table:weightedperiodic}
\end{centering}
\end{table}

In order to avoid any pairing with cost $+3$, we cannot have \tileWhite\ and \tileBlack\ next
to each other.  That is, any region with just these two tile types must be in a checkerboard
pattern.  Of course, when $N$ is odd, we cannot fill the whole grid that way.  Indeed, in each
row and column, there must be at least one \tileH, \tileV, or \tileC\ tile.

If we have a \tileC\ tile anywhere, it must have a \tileH\ tile adjacent to it to the left and
right, and a \tileV\ tile adjacent to it above and below.
Whenever we have a \tileV\ tile, it must form a vertical line with only \tileV\ and \tileC\ tiles
wrapping all the way around the torus.  Similarly, if we have a \tileH\ tile, there must be
a horizontal line with only \tileH\ and \tileC\ tiles.  Furthermore, if there is a single \tileC\ tile
in the horizontal line, the overall cost of the edges within the line is $+2$, and if there is more
than one \tileC\ tile, the cost is larger.

Since \tileH\ tiles cannot be adjacent vertically, and \tileV\ tiles cannot be adjacent horizontally,
nor can either be adjacent to the other in any direction, the only way to avoid a full row or column with
only \tileBlack\ and \tileWhite\ tiles is to have at least one horizontal line and at least one vertical
line.  Whenever a horizontal line and vertical line intersect, there must be a \tileC\ at the crossing.
Each crossing has a cost of $+2$, so to minimize the total cost, we can have only one crossing, and thus
just one horizontal line of {\tileH}'s and one vertical line of {\tileV}'s.
Those horizontal and vertical lines will determine the boundary of an $(N-1) \times (N-1)$
grid, and we can set the boundary conditions at the corners via adjacency to the
\tileH\ and \tileV\ tiles.  The arrangement is much like that of figure~\ref{fig:periodicnoweights}a,
but with \tileBlackD\ replaced by \tileBlack\ and \tileWhiteU\ replaced by \tileWhite.

\end{proof}

\subsection{One-Dimensional Tiling}
\label{sec:oneDclassical}

In one dimension, the tiling problem becomes the following:

\begin{definition}
{\bf $1$-DIM \tiling}

\noindent
{\bf Problem Parameters:} A set of tiles $T = \{t_1,\ldots,t_m\}$.
A set of constraints $H \subseteq T \times T$ such that if
$t_i$ is placed to the left of $t_j$, then it must be the case that
$(t_i,t_j) \in H$. A designated tile $t_1$ that must be placed at the ends of
the line.

\noindent
{\bf Problem Input:} Integer $N$, specified in binary.

\noindent
{\bf Output:} Determine whether there is a valid tiling of a line of length $N$.
\end{definition}

We can also define a WEIGHTED $1$-DIM \tiling\ problem analogously to the WEIGHTED \tiling\ problem in section~\ref{sec:weighted}.

\begin{theorem}
$1$-DIM \tiling\ and WEIGHTED $1$-DIM \tiling\ are in $\Pclass$.
\end{theorem}

\begin{proof}
{\bf Unweighted case:}
Let us create a directed graph with $m$ nodes.  The $i$th node corresponds to the
$i$th tile type $t_i$, and there is an edge between $i$ and $j$ iff $(t_i, t_j) \in H$.
We wish to know whether there is a path of length exactly $N$ starting and ending at $t_1$.
Certainly this can be done in time ${\rm poly}(N)$, but we actually wish to do it
in time ${\rm polylog}(N)$.  We will also generalize slightly to allow the left end of the line
to have tile $t_0$ and the right end to have tile $t_1$.

First, let us make a table of all paths which start and end at the same node
and have  no other repeated nodes.
Such a path is called a {\em simple cycle}.
A path is
a \emph{simple path} if it is a simple cycle or it contains no repeated nodes.
Let $M$ be the
set of all simple cycles, and let $l(p)$ be the length of the path $p$.  Clearly when $p$ is
a simple path,
$l(p) \leq m+1$, which is constant, so constructing the table of simple cycles takes constant time.

Given any path $p$, we can decompose it into a simple path and a multiset $P$ of simple cycles by progressively removing simple cycles from $p$ until
we are left with a simple path, adding each removed cycle to $P$.

\begin{definition}
A multiset $P$ of simple cycles is \emph{allowed} for the simple path $p$ if there exists a path $q$ which can be decomposed as above into exactly the
multiset $P$ plus $p$.
\end{definition}

\begin{claim}
The multiset $P$ is allowed for a path $p$ iff $P$'s underlying set $Q$ (with multiple occurrences removed) is allowed for $p$.
\end{claim}

\begin{proof}[ of claim]
$\Rightarrow:$ If we have any path through the node $t_i$, we can insert a simple cycle for $t_i$ to get another possible path.  If $q$ is a path which can be
decomposed to the set $Q$, then every node on every simple cycle in $P$ appears in $q$, so we can insert as many copies as we need of the simple
cycles in $P$.

$\Leftarrow:$ Conversely, let $q$ be a path for the allowed multiset $P$.  We can express $q$ as a tree, with the root being the simple path $p$ and each node
being a simple cycle.  Cycle
$c'$ is a child of $c$ if they share a node and $c'$ is not the parent of $c$.
The tree can be obtained by iteratively taking two consecutive
occurrences of a node $t_i$ along a cycle such that there are no repetitions
between the occurrences of $t_i$ and identifying the two occurrences of $t_i$.
The cycle is then removed (except for one of the occurrences of $t_i$) and the process
iterated until no cycles remain.
We wish to modify the tree so that the resulting path $q'$ has the same multiset $P$
as its decomposition, but with the property that there is a
subtree rooted at $p$ containing just one of every simple cycle in $Q$.  Then we
can take the path $q$ corresponding to that subtree, and $q$ has
the set $Q$ for its decomposition.

That is, we need that within the tree, for every $r \in Q$, there exists an
occurrence $c$ of $r$ such that the path from $p$ to $c$ in the tree
does not pass through two occurrences of any simple cycle.  Since we can
move any node corresponding to a simple cycle for $t_i$ to be a child of
\emph{any} simple cycle that contains $t_i$, we can put the tree in the
desired format by moving one instance of each $r \in Q$ up the tree until it
lies off of the first cycle which contains $t_i$.
\end{proof}

\begin{claim}
\label{claim:sumpaths}
There is a tiling of a line of length $N$ from $t_0$ to $t_1$ iff $\exists$ simple path $p$, and function $f: M \rightarrow \integer^+ \cup \{0\}$ such that:
\begin{enumerate}
\item $p = (t_{i_0}, t_{i_1}, \ldots, t_{i_L})$, with $i_0 = 0$, $i_L = 1$, $l(p) = L+1$,
\item $L+1 + \sum_{q \in M} f(q) l(q) = N$,
\item $f^{-1} (\integer^+)$ is allowed for $p$.
\end{enumerate}
\end{claim}

Given an allowed multiset $P$ of simple cycles, the function $f(q)$ represents the
number of times the particular simple cycle $q$ appears in the multiset,
which means that $f^{-1}(\integer^+)$ is the underlying set for $P$.
To see that Claim \ref{claim:sumpaths} is true, observe that
the resulting path contains $N$ nodes, as desired.  Conversely,
given $p$ and $f$, we can create an allowed multiset $P$ for $p$ by
including each simple cycle $q$ a number of times equal to $f(q)$.

There are only a constant number of simple paths $p$ and allowed sets $P$ of simple cycles.  By running over $(p, P)$ and applying
claim~\ref{claim:sumpaths}, the $1$-DIM \tiling\ problem reduces to determining if there

exists $f': P \rightarrow \integer^+ \cup \{0\}$  such that $ \sum_{q \in P} f'(q)l(q) = N'$, $N' = N - (L+1) - \sum_{q \in P} l(q)$.  (We take $f'(q)
= f(q)-1$, with $f$ given by the claim.)  With fixed $(p,P)$, we must therefore solve the following problem: Given
a set of $m'$ positive integers $a_k$ (the set of distinct $l(q)$ for $q \in P$), do there exist
non-negative integers $b_k$ (the corresponding $f'(q)$, summed over any cycles with equal lengths)
such that $\sum_k a_k b_k = N'$?  This is a special case of the Unbounded Knapsack Problem
\cite{lueker75} where the set of objects is fixed and only the total cost allowed varies.

Note that, if we allowed $b_k$ to be negative, this would just be answered by determining if
$N'$ is a multiple of $\gcd(a_1, \ldots, a_{m'})$.  Indeed, if $N'$ is greater than $m'g$, where $g$ is
the least common multiple of the $a_k$, this is all we need to determine: If $\gcd(a_1, \ldots, a_{m'}) | r$,
then we can write $r = \sum_k a_k c_k$, with the property that $a_k c_k \geq -g$ for all $k$.
Thus, if $N' = m'g + r$, then $N' = \sum_k a_k (c_k + g/a_k)$, providing our solution.
Conversely, if $\gcd(a_1, \ldots, a_{m'}) \not| N'$, then it is not possible that $\sum_k a_k b_k = N'$.
For $N' < m'g$, matters are more complicated, but since $g \leq (m')!$ is a constant, we can
answer those cases via a look-up table.

\medskip

{\bf Weighted case:}
The $1$-DIM \tiling\ problem remains easy even if we take a weighted variant.
In this case, the directed graph we get is a complete graph, but the edges have a
weight associated with them, and we want to determine the minimal cost path of length
$N$ from $t_0$ to $t_1$.  Now we need to count both the length and the cost of each
simple cycle.  That is, we need to learn if $\exists$ simple path $p$,
$f: M \rightarrow \integer^+ \cup \{0\}$ such that:
\begin{enumerate}
\item $p = (t_{i_0}, t_{i_1}, \ldots, t_{i_L})$, with $i_0 = 0$, $i_L = 1$, $l(p) = L+1$,
\item $L+1 + \sum_{q \in M} f(q) l(q) = N$,
\item $c(p) + \sum_{q \in M} f(q) c(q) \leq c$,
\item $f^{-1} (\integer^+)$ is allowed for $p$,
\end{enumerate}
where $c(q)$ is the cost of path $q$.

Once we fix $(p,P)$, we again get a set $a_k$ of positive integer lengths $l(q)$
for $q \in P$, but we also have $c_k$, the costs $c(q)$ of $q \in P$.
We can assume $f'(q)=0$ for all but
one cycle $q \in P$ for each length --- we should pick the one with the lowest cost $c(q)$.
Now we wish to find non-negative integers $b_k$ that minimize $\sum_k b_k c_k$ subject to
$\sum_k a_k b_k = N'$, which is again a special case of the Unbounded Knapsack Problem.
We again set $g \leq m!$ to be the least common multiple of the $a_k$ (i.e., the lcm of $1, \ldots, m$).
We create a look-up table to give an optimal path for $N < mg$.  When $N$ is a multiple
of $g$, note that we can find an optimal path by simply calculating which simple cycle
has the lowest ratio $c_k/a_k$; then we just use that cycle $g/a_k$ times.

For other values of $N' = dg + r$, with $d \geq m$, we just look at the optimal path for
$(m-1)g + r$ and add $(d-m+1)g/a_k$ copies of the optimal simple cycle.  This will be an
optimal path for $N'$. To show this, note that we can assume without loss of generality that
an optimal solution will not have $b_k \geq g/a_k$ for more than one value of $k$: Suppose there
were two values $k_1$, $k_2$ for which $b_k \geq g/a_k$, with $c_{k_1}/a_{k_1} \leq c_{k_2}/a_{k_2}$.
Then we can shift $b_{k_1} \mapsto b_{k_1} + g/a_{k_1}$ and $b_{k_2} \mapsto b_{k_2} - g/a_{k_2}$,
leaving the total length the same
without increasing the cost.  Furthermore, using the same logic, we know that the one
value $k_1$ for which $b_k \geq g/a_k$ must have a minimal value of $c_k/a_k$.  For $N' > mg$,
there must actually exist such a value of $k_1$, so the optimal path must have been formed in
the way we have described.

\end{proof}

\subsection{Additional Symmetry}

We consider two additional kinds of symmetry.  If we have {\em reflection symmetry},
then  $(t_i, t_j) \in H$ implies $(t_j, t_i) \in H$ as well, and  $(t_i, t_j) \in V$
implies  $(t_i, t_j) \in V$ also.  That is, the tiling constraints to the left and right are
the same, as are the constraints above and below.  However, if we only have reflection
symmetry, there can still be a difference between the horizontal and vertical directions.
If we have {\em rotation symmetry}, we have reflection symmetry and also $(t_i, t_j) \in H$
iff $(t_i, t_j) \in V$.  Now the direction does not matter either.

Note that these types of reflection and rotation symmetries assume that we can reflect
or rotate the tiling rules without simultaneously reflecting or rotating the tiles.
For instance, if a tile $t_i$ has a pattern on it (such as \tileNE) that looks different when it is turned
upside down, then when we reflect vertically, we also could reflect the tile, producing
a new tile $R(t_i)$ (\tileSE\ in this case).  We could define a reflection symmetry for this type of tile too:
$(t_i, t_j) \in H$ iff $(R(t_j), R(t_i)) \in H$, etc., but we just get the same complexity
classes as for the case with no additional reflection symmetry.  This is because we can
add an extra layer of tiles with arrows on them and put on a constraint that any adjacent arrow
tiles must point the same direction.  While either direction will work (or any of the four
directions in the case of rotation symmetry), one direction ends up preferred in any given
potential tiling, so by looking at the arrow tile in a given spot, we can effectively reproduce
rules that have no reflection symmetry. Thus, we address the case where the rules have
reflection or rotation symmetry without simultaneously reflecting or rotating the tiles.

\begin{theorem}
\label{theorem:symmetry}
When the constraints for \tiling\ have reflection symmetry, there exist $N_e, N_o \in \integer^+ \cup
\{\infty\}$ such that for even $N \geq N_e$ or odd $N \geq N_o$, a valid tiling exists, while
for even $N < N_e$ and odd $N < N_o$, there is no tiling (except for $N=1$, when there is a
trivial tiling).
\begin{itemize}
\item When we have open boundary conditions, either $N_e = N_o = \infty$ or $N_e = N_o = 1$.
\item When we have the four corners boundary condition, either $N_o = 3$ or $N_e = N_o = \infty$.
\item When we have periodic boundary conditions, either $N_e = 2$ or $N_e = N_o = \infty$.
\item When we have rotation symmetry, $N_e$ and $N_o$ are computable.
\end{itemize}
\end{theorem}
When we have reflection symmetry but not rotation symmetry, we have been unable to determine
so far whether $N_e$ and $N_o$ are computable for the four corners and periodic boundary conditions,
respectively.

\begin{proof}
To prove the theorem, simply note that given a valid tiling of an $N \times N$ grid, with $N \geq 4$,
we can extend it to a valid tiling of an $(N+2) \times (N+2)$ grid.  The main observation is that we
can repeat existing patterns when we have reflection symmetry, because if $AB$ is a legal configuration,
so is $ABAB$.

In order to extend a tiling, we can do the following: strip off the leftmost column and bottommost row, and replace them with duplicates of the next
two rows and columns.  We can fill in the corner by duplicating the bottom left $2 \times 2$ square once we have stripped off the rows.  Then the
original leftmost column and bottommost row became the new leftmost and bottommost column and row, with some duplication to lengthen them to the right
size.  (See figure~\ref{fig:reflection}.)

\begin{figure}
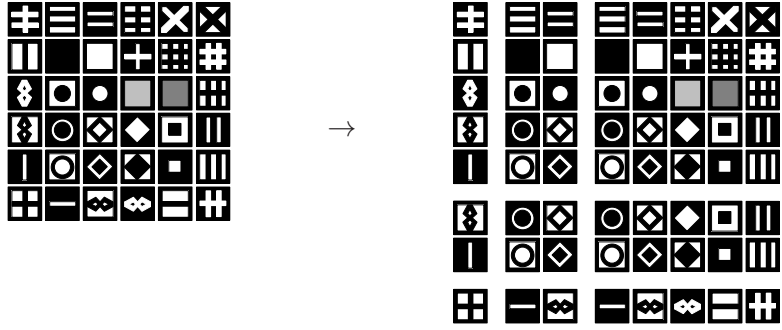

\begin{centering}
\begin{tabular}[t]{c@{\extracolsep{0.1em}}c@{}c@{}c@{}c@{}c}
\tileHHVrev & \tileHH & \tileHHrev & \tileHHV & \tileXrev & \tileX \\
\tileV & \tileBlack & \tileWhite & \tileCrev & \tileCrosshatch & \tileCrosshatchrev \\
\tileVvarrev & \tileCirc & \tileCircrev & \tileLight & \tileDark & \tileHVV \\
\tileVvar & \tileRingrev & \tileDHol & \tileDrev & \tileSquare & \tileVVrev \\
\tileVrev & \tileRing & \tileDHolrev & \tileD & \tileSquarerev & \tileVV \\
\tileC & \tileHrev & \tileHvar & \tileHvarrev & \tileH & \tileHVVrev \\
\end{tabular}
\qquad
\begin{tabular}[t]{c}
\\
\\
\\
$\rightarrow$ \\
\end{tabular}
\qquad
\begin{tabular}[t]{c@{\extracolsep{0.1em}\hspace{0.5em}}c@{}c@{\hspace{0.5em}}c@{}c@{}c@{}c@{}c}
\tileHHVrev & \tileHH & \tileHHrev & \tileHH & \tileHHrev & \tileHHV & \tileXrev & \tileX \\
\tileV & \tileBlack & \tileWhite & \tileBlack & \tileWhite & \tileCrev & \tileCrosshatch & \tileCrosshatchrev \\
\tileVvarrev & \tileCirc & \tileCircrev & \tileCirc & \tileCircrev & \tileLight & \tileDark & \tileHVV \\
\tileVvar &\tileRingrev &  \tileDHol & \tileRingrev & \tileDHol & \tileDrev & \tileSquare & \tileVVrev \\
\tileVrev & \tileRing & \tileDHolrev & \tileRing & \tileDHolrev & \tileD & \tileSquarerev & \tileVV
\vspace{0.5em} \\
\tileVvar & \tileRingrev & \tileDHol & \tileRingrev & \tileDHol & \tileDrev & \tileSquare & \tileVVrev \\
\tileVrev & \tileRing & \tileDHolrev & \tileRing & \tileDHolrev & \tileD & \tileSquarerev & \tileVV
\vspace{0.5em} \\
\tileC & \tileHrev & \tileHvar & \tileHrev & \tileHvar & \tileHvarrev & \tileH & \tileHVVrev \\
\end{tabular}
\caption{Extending a tiling of a $6 \times 6$ grid to $8 \times 8$ when the tiling conditions have a reflection symmetry.  The
spaces on the right mark the repeated rows and columns.}
\label{fig:reflection}
\end{centering}
\end{figure}

This strategy handles even more general boundary conditions than the three main cases we consider.  In particular,
it also works if the tiling rules on the sides of the grid are completely different from the tiling rules in the interior
of the grid.  When the tiling rules are the same on the sides as in the center, except perhaps for the
corners, we can copy the outermost two rows and columns, so this strategy works for $N \geq 2$.

\medskip

{\bf Four corners boundary conditions:} When $N$ is odd and we have the four-corners boundary condition,
we can look to see if there is a tiling of a $2 \times 2$ grid with the corner tile in the upper right corner.  If so,
we can duplicate the right column to add a column on the left, and then duplicate
the top row to add a row on the
bottom.  (See figure~\ref{fig:reflectioncorners}.)  This gives us a tiling
of the $3 \times 3$ grid with all four corners correctly tiled.  Conversely,
if there is no $2 \times 2$ tiling containing one of the corner tiles, then there cannot be a tiling of any $N \times N$
grid with $N>1$, which means that $N_o = N_e  = \infty$.

\begin{figure}
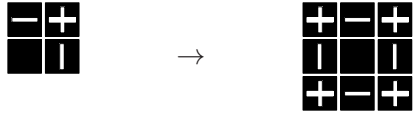

\begin{centering}
\begin{tabular}[t]{c@{\extracolsep{0.1em}}c}
\tileHrev & \tileCrev \\
\tileBlack & \tileVrev \\
\end{tabular}
\qquad
\begin{tabular}[t]{c}
\\
$\rightarrow$ \\
\end{tabular}
\qquad
\begin{tabular}[t]{c@{\extracolsep{0.1em}}c@{}c}
\tileCrev & \tileHrev & \tileCrev \\
\tileVrev & \tileBlack & \tileVrev \\
\tileCrev & \tileHrev & \tileCrev \\
\end{tabular}
\caption{Extending the tiling of a corner to a tiling of a $3 \times 3$ grid with fixed corner boundary conditions.}
\label{fig:reflectioncorners}
\end{centering}
\end{figure}

\medskip

{\bf Open boundary conditions:} In the case of completely open boundary conditions, we need
only check a $2 \times 2$ grid to see if there is a valid tiling.  If not, there cannot be a tiling of
any size $N > 2$ either.  If there is, we can extend it as in figure~\ref{fig:reflectioncorners} to get a tiling of
the $3 \times 3$ grid as well, and extend as in figure~\ref{fig:reflection} to get a tiling of any size grid.

\medskip

{\bf Periodic boundary conditions:} When we have periodic boundary conditions, a valid tiling of a  $2 \times 2$
grid with open boundary conditions gives us a valid periodic tiling of a $2 \times 2$ grid as well, so we
can extend it to all even $N$.  However, if we try to apply the strategy of figure~\ref{fig:reflectioncorners}
to extend it to odd $N$, we potentially ruin the periodicity, so we do not know if $N_o$ is computable
in this case. If there is no tiling of a $2 \times 2$
grid with open boundary conditions, then there can not be a tiling for any $N$ with
periodic boundary conditions and $N_o = N_e = \infty$.

\medskip

{\bf Rotation symmetry:} When we have rotation symmetry (with any boundary conditions), it will suffice to compute $N_e$ and $N_o$ as the minimum even
and odd lengths that allow us to tile a single side of the square.  If a tiling of one side exists, we can use this same tiling on all four sides and
then fill in the center using diagonal stripes of identical tiles, as in figure~\ref{fig:rotation}.  In the case where there are special rules for the
boundary, we will need to tile a side plus the adjacent row/column as preparation, but this presents little additional difficulty.
\begin{figure}
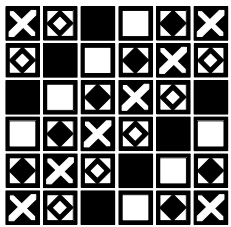

\begin{centering}
\begin{tabular}{c@{\extracolsep{0.1em}}c@{}c@{}c@{}c@{}c}
\tileXrev & \tileDHol & \tileBlack & \tileWhite & \tileD & \tileXrev \\
\tileDHol & \tileBlack & \tileWhite & \tileD & \tileXrev & \tileDHol \\
\tileBlack & \tileWhite & \tileD & \tileXrev & \tileDHol & \tileBlack \\
\tileWhite & \tileD & \tileXrev & \tileDHol & \tileBlack & \tileWhite \\
\tileD & \tileXrev & \tileDHol & \tileBlack & \tileWhite & \tileD \\
\tileXrev & \tileDHol & \tileBlack & \tileWhite & \tileD & \tileXrev \\
\end{tabular}
\caption{A tiling of a $6 \times 6$ square with rotation symmetry.}
\label{fig:rotation}
\end{centering}
\end{figure}

To tile a single side, we can use an approach similar to the previous $1$-dimensional case.  However, now matters are much simpler, since the graph is
now undirected.  There are thus always many size $2$ cycles, so we need only find the minimal even- and odd-length cycles for $t_1$.  That sets an
upper bound on $N_e$ and $N_o$.  It might be one of these can be made smaller, but that is straightforward to check as well.
\end{proof}

The weighted cases with reflection or rotation symmetry are more difficult, so we treat them in separate subsections.

\subsection{Weighted Tiling With Reflection Symmetry}

In this subsection, we consider WEIGHTED \tiling\ with reflection symmetry.

\subsubsection{Open or Four Corners Boundary Conditions}

\begin{theorem}
WEIGHTED \tiling\ is \NEXP-complete with either open boundary conditions or boundary conditions fixed at the corners.
\end{theorem}

It is interesting to note that the total cost $p(N)$ of the satisfying tilings that appear in our proof is linear in $N$.  We have not been able to prove a result for these boundary conditions when the cost function is a constant.  (Recall that we can always shift the costs so that $p(N) = o(N^2)$.)

\begin{proof}

We will describe the costs for open boundary conditions, but they will imply fixed corner tiles, so the four-corners boundary conditions can use the same costs.  When we list costs, we will talk of pairings being ``forbidden.''  Of course, in WEIGHTED \tiling, no pair of adjacent tiles is completely forbidden, but we will assign a large cost (say $+30$) to ``forbidden'' pairings, and will show that a low-cost tiling can never include a forbidden pair.

The tiling will consist of three extra layers beyond those the main layers used to prove the basic \NEXP-completeness of \tiling~(as described in section~\ref{sec:tiling}).  Layer $1$ will break the reflection symmetry in both the horizontal and vertical directions, but the broken symmetry will only be visible at certain locations in the grid.  That is, there is not a unique optimal tiling of Layer $1$.  Any optimal tiling can be globally reflected in the horizontal and/or vertical directions to get another optimal tiling.  We will choose rules such that one of these reflections, applied to an optimal tiling, always produces a distinct tiling.  There will, in fact, be exactly $4$ optimal tilings of Layer $1$, related to each other by reflections. (There is actually another class of optimal tilings of Layer $1$ when it is considered by itself, but that class will be eliminated by the rules for Layers $2$ and $3$.)  Since the reflected tilings are different from each other, there are certain locations we can look at to determine which orientation of the four we actually have.  By choosing one of the four orientations to be a canonical reference point, we can globally define the four directions ``left,'' ``right,'' ``up,'' and ``down.''  We would like to use this information to implement direction-dependent rules for the main layers even though the underlying rules remain reflection-invariant.  To do this, we need to be able to look at an adjacent pair of tiles in Layer $1$ and determine a direction from that.  Since the rules are translation-invariant, we cannot make reference to the location of the tiles, only that they are adjacent in either the horizontal or vertical direction.  Nevertheless, for an optimal tiling of Layer $1$, it is possible in some locations in the grid to determine directions purely from the Layer $1$ tiles.  However, it is not possible everywhere.

Therefore, we use Layer $2$ to extend the broken horizontal reflection symmetry, so there is a locally visible distinction between ``left'' and ``right'' at all locations in the grid.  Layer $3$ similarly extends the broken vertical reflection symmetry, allowing us to define ``up'' and ``down'' at all locations in the grid.  Most of the work, and all of the non-trivial weights, go into Layer $1$.  Layers $2$ and $3$ only have weights $0$ and $+30$ (for forbidden pairings).

We will restrict $N$ to be even and $1 \bmod 3$.  That is, $N \equiv 4 \bmod 6$.  We will set the allowed cost to be $p(N) = 76 -16N$, and assume $N$ is large.

{\bf Layer $1$:}

Layer $1$ will have $9$ types of tile: \tileV, \tileH, \tileC, \tileBlack, \tileWhite, \tileDark, \tileLight, \tileRing, and \tileCirc.  \tileV, \tileH, and \tileC\ will form the outer border of the grid in an optimal tiling.  The interior of the grid will mostly consist of \tileBlack, \tileWhite, \tileDark, and \tileLight\ tiles.  By and large, these four tiles will alternate \tileBlack\ and \tileWhite\ or \tileDark\ and \tileLight\ horizontally, and \tileBlack\ and \tileLight\ or \tileDark\ and \tileWhite\ vertically.  However, there will also be a diagonal stripe of \tileRing\ tiles reaching from corner to corner, with \tileCirc\ tiles at the end.  At the left and right ends of the interior, adjacent to the \tileV\ tiles, we only permit \tileDark\ and \tileWhite\ tiles (alternating vertically) or \tileCirc\ tiles.  Because $N$ is even, the tiles to the left and right of the \tileRing\ and \tileCirc\ tiles will be different in a consistent way, allowing us to distinguish left and right near the central diagonal.  Similarly, the tiles above and below the \tileRing\ and \tileCirc\ tiles are different, allowing us to distinguish up and down.
See figure~\ref{fig:weightedreflection1} for an example of this tiling.

\begin{figure}
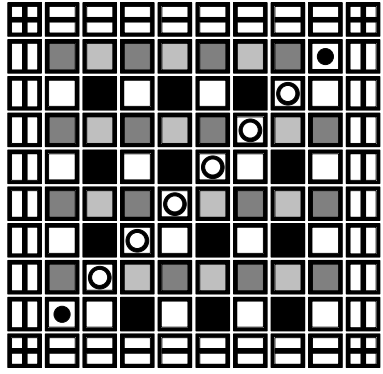

\begin{centering}
\begin{tabular}{c@{\extracolsep{0.1em}}c@{}c@{}c@{}c@{}c@{}c@{}c@{}c@{}c}
\tileC & \tileH     & \tileH     & \tileH     & \tileH     & \tileH     & \tileH     & \tileH     & \tileH     & \tileC \\
\tileV & \tileDark  & \tileLight & \tileDark  & \tileLight & \tileDark  & \tileLight & \tileDark  & \tileCirc  & \tileV \\
\tileV & \tileWhite & \tileBlack & \tileWhite & \tileBlack & \tileWhite & \tileBlack & \tileRing  & \tileWhite & \tileV \\
\tileV & \tileDark  & \tileLight & \tileDark  & \tileLight & \tileDark  & \tileRing  & \tileLight & \tileDark  & \tileV \\
\tileV & \tileWhite & \tileBlack & \tileWhite & \tileBlack & \tileRing  & \tileWhite & \tileBlack & \tileWhite & \tileV \\
\tileV & \tileDark  & \tileLight & \tileDark  & \tileRing  & \tileLight & \tileDark  & \tileLight & \tileDark  & \tileV \\
\tileV & \tileWhite & \tileBlack & \tileRing  & \tileWhite & \tileBlack & \tileWhite & \tileBlack & \tileWhite & \tileV \\
\tileV & \tileDark  & \tileRing  & \tileLight & \tileDark  & \tileLight & \tileDark  & \tileLight & \tileDark  & \tileV \\
\tileV & \tileCirc  & \tileWhite & \tileBlack & \tileWhite & \tileBlack & \tileWhite & \tileBlack & \tileWhite & \tileV \\

\tileC & \tileH     & \tileH     & \tileH     & \tileH     & \tileH     & \tileH     & \tileH     & \tileH     & \tileC
\end{tabular}
\caption{A tiling for Layer $1$ of a $10 \times 10$ grid for the WEIGHTED \tiling\ problem with reflection symmetry.  The reflection symmetry is broken in the immediate vicinity of the central diagonal.}
\label{fig:weightedreflection1}
\end{centering}
\end{figure}

To achieve this, we use the rules given in table~\ref{table:reflectionopen}.
\begin{table}
\begin{centering}
\begin{tabular}{c|rrrrrrrrr}
  \multicolumn{10}{c}{Horizontal tiling rules} \\
              & $\tileV$ & $\tileH$ & $\tileC$ & $\tileBlack$ & $\tileWhite$ & $\tileDark$ & $\tileLight$ & $\tileRing$ & $\tileCirc$ \\
 \hline
$\tileV$      &    30    &    30    &    30    &      30      &       6      &      6      &     30       &     30      &       7     \\
$\tileH$      &    30    &   -11    &     0    &      30      &      30      &     30      &     30       &     30      &      30     \\
$\tileC$      &    30    &     0    &    30    &      30      &      30      &     30      &     30       &     30      &      30     \\
$\tileBlack$  &    30    &    30    &    30    &      30      &       0      &     30      &     30       &      1      &      30     \\
$\tileWhite$  &     6    &    30    &    30    &       0      &      30      &     30      &     30       &      1      &       1     \\
$\tileDark$   &     6    &    30    &    30    &      30      &      30      &     30      &      0       &      1      &       1     \\
$\tileLight$  &    30    &    30    &    30    &      30      &      30      &      0      &     30       &      1      &      30     \\
$\tileRing$   &    30    &    30    &    30    &       1      &       1      &      1      &      1       &     30      &      30     \\
$\tileCirc$   &     7    &    30    &    30    &      30      &       1      &      1      &     30       &     30      &      30     \\
\end{tabular}

\medskip

\begin{tabular}{c|rrrrrrrrr}
  \multicolumn{10}{c}{Vertical tiling rules} \\
              & $\tileV$ & $\tileH$ & $\tileC$ & $\tileBlack$ & $\tileWhite$ & $\tileDark$ & $\tileLight$ & $\tileRing$ & $\tileCirc$ \\
 \hline
$\tileV$      &   -11    &   30     &     0    &      30      &      30      &     30      &      30      &     30      &      30     \\
$\tileH$      &    30    &   30     &    30    &       6      &       6      &      6      &       6      &     30      &       7     \\
$\tileC$      &     0    &   30     &    30    &      30      &      30      &     30      &      30      &     30      &      30     \\
$\tileBlack$  &    30    &    6     &    30    &      30      &      30      &     30      &       0      &      1      &      30     \\
$\tileWhite$  &    30    &    6     &    30    &      30      &      30      &      0      &      30      &      1      &       1     \\
$\tileDark$   &    30    &    6     &    30    &      30      &       0      &     30      &      30      &      1      &       1     \\
$\tileLight$  &    30    &    6     &    30    &       0      &      30      &     30      &      30      &      1      &      30     \\
$\tileRing$   &    30    &   30     &    30    &       1      &       1      &      1      &       1      &     30      &      30     \\
$\tileCirc$   &    30    &    7     &    30    &      30      &       1      &      1      &      30      &     30      &      30     \\
\end{tabular}
\caption{The tiling weights for layer 1 for WEIGHTED \tiling\ with reflection symmetry.  Since there is reflection symmetry, the horizontal and vertical tiling weight matrices are symmetric.}
\label{table:reflectionopen}
\end{centering}
\end{table}
\begin{claim}
\label{claim:reflayer1cost}
With these rules, for large $N$, the cost of an optimal tiling is $76 - 16N$.
\end{claim}

To show that the desired tiling is indeed optimal, let $T$ be some tiling, with $T_{ab}$ the identity of the tile in location $(a,b)$ in the grid.
Let $w(T)$ be the total cost of $T$ and let $w(S_{ab})$ be the total cost of the $2 \times 2$ square in location $(a,b)$: i.e.,
\begin{equation}
w(S_{ab}) = w_H (T_{ab}, T_{(a+1)b}) + w_V (T_{ab}, T_{a(b+1)}) + w_V (T_{(a+1)b}, T_{(a+1)(b+1)}) + w_H (T_{a(b+1)}, T_{(a+1)(b+1)}).
\end{equation}
Finally, let $w(R_b) = \sum_a w_H (T_{ab}, T_{(a+1)b})$ be the internal cost of row $b$ and $w(C_a) = \sum_b w_V (T_{ab}, T_{a(b+1)})$ be the internal cost of column $a$.  Then
\begin{equation}
2w(T) = \sum_{a,b = 1}^{N-1} w(S_{ab}) + w(R_1) + w(R_N) + w(C_1) + w(C_N).
\label{eq:costbyrows}
\end{equation}

The only negative weights are the $\begin{array}{c@{\extracolsep{0.1em}}c} \tileH \tileH \end{array}$ horizontal edges and the vertical edges between two \tileV\ tiles.  There are no possible $2 \times 2$ squares with negative cost, and the only $2 \times 2$ squares with total cost $0$ are
\begin{equation}
\begin{array}{c@{\extracolsep{0.1em}}c}  \tileDark & \tileLight \\ \tileWhite & \tileBlack \\ \end{array}
\end{equation}
and its three reflections.  Therefore, to get a minimal cost tiling, we should have as many $2 \times 2$ squares as possible be of that form, and have preferentially $\tileH$ and $\tileV$ around the edges of the grid.

We can break the whole cost down into pairs of rows.  Let us calculate the minimal value of
\begin{align}
w'(R_b, R_{b+1}) & \equiv w(R_b) + w(R_{b+1}) + 2\sum_{a=1}^N w_V (T_{ab}, T_{a(b+1)}) \\
& = \sum_{a=1}^{N-1} w(S_{ab}) + w_V (T_{1b}, T_{1(b+1)}) + w_V (T_{Nb}, T_{N(b+1)})
\label{eq:tworowscost}
\end{align}
for a particular assignment of tiles to a pair of rows.  We must decide a collection of $N$ vertical pairs of tiles to achieve the minimum cost.
Since
\begin{equation}
2w(T) = \sum_b w'(R_b, R_{b+1}) + w(R_1) + w(R_N),
\end{equation}
the minimal value of $w'$ for pairs of rows is a good guide to the minimal achievable total cost for the whole grid.  By equation~(\ref{eq:tworowscost}), if every square in the pair of rows had cost $0$, and the first and last pairs had the minimal edge cost $-11$, we would have $w'(R_b, R_{b+1}) = -22$, but we cannot quite achieve that:
\begin{lemma}
\label{lemma:refminrowpair}
When $N>4$ is even, the minimum value of $w'(R_b, R_{b+1})$ is $-12$.

Any pair of rows $R_b$, $R_{b+1}$ with $w'(R_b, R_{b+1})=-12$ has the following structure:  Each row starts and ends with \tileV\ tiles, has exactly one \tileCirc\ or \tileRing\ tile, and to the left and right of the \tileCirc\ or \tileRing, alternates either \tileLight\ and \tileDark\ or \tileWhite\ and \tileBlack.  The \tileCirc\ and/or \tileRing\ tiles in the two rows are diagonally adjacent to each other.  Furthermore, if one row has alternating \tileLight\ and \tileDark\ tiles to the right (left) of the \tileCirc\ or \tileRing\ tile, the other row has alternating \tileWhite\ and \tileBlack\ tiles to the right (left) of the \tileCirc\ or \tileRing\ tile, as required by the allowed pairs of adjacent tiles.

Beyond this structure, there are four classes of solutions:
\begin{itemize}
\item[a.] One row contains a \tileCirc\ tile adjacent to one of the \tileV\ tiles.  The other tile in the column with the \tileCirc\ tile is the same (\tileWhite\ or \tileDark) as the tile in the column with the \tileRing\ tile in the other row.

\item[b.] One row contains a \tileCirc\ tile adjacent to one of the \tileV\ tiles.  The other tile in the column with the \tileCirc\ tile is different than the tile in the column with the \tileRing\ tile in the other row.  In this case, one row contains no \tileLight\ or \tileDark\ tiles and the other row contains no \tileWhite\ or \tileBlack\ tiles.

\item[c.] Both rows have \tileRing\ tiles.  The two tiles in the same columns as the \tileRing\ tiles are the same.  In this case, each row contains at least one tile from each pair (\tileLight, \tileDark) and (\tileWhite, \tileBlack).

\item[d.] Both rows have \tileRing\ tiles.  The two tiles in the same columns as the \tileRing\ tiles are different.  In this case, each row contains tiles from only one of the pairs (\tileLight, \tileDark) and (\tileWhite, \tileBlack).
\end{itemize}

If one side (left or right) is forbidden to have a \tileV\ for one or both rows, the minimal value achievable is $-10$.  This can be achieved only through a pair of rows which uses \tileV\ at the other end of both rows, and alternates \tileLight\ and \tileDark\ elsewhere in one row, and \tileBlack\ and \tileWhite\ elsewhere in the other row.

If \tileV\ is forbidden for both ends of one or both rows, the minimal value achievable is $0$.  This can only be achieved by having one row alternating \tileLight\ and \tileDark, and the other row alternating \tileBlack\ and \tileWhite.
\end{lemma}

The four classes of minimal-cost solution given by the lemma depend on two factors.  First, do we have only \tileRing\ tiles, or one \tileCirc\ tile and one \tileRing\ tile?  Classes a and b have \tileCirc, whereas classes c and d do not.  Second, do we have the same alternating pair (\tileBlack, \tileWhite) or (\tileLight, \tileDark) on both sides of the \tileRing\ tile(s), or do we have different alternating pairs on the left and right?  Classes a and c have different alternating pairs, while classes b and d have the same pair on each side of the \tileRing\ tile in a particular row.  In our solution, we are primarily interested in classes b and d because when $N$ is even, they have an asymmetry between left and right in the vicinity of the \tileRing\ tiles.  However, at this stage we cannot rule out classes a and c.

\begin{proof}[ of lemma]
The minimal cost of a pair of rows can be calculated using the efficient $1$-dimensional tiling algorithm described in section~\ref{sec:oneDclassical}. We merge each vertically adjacent pair of tiles into a single node of the graph.  The difference from the algorithm of section~\ref{sec:oneDclassical} is that to minimize $w'(R_b, R_{b+1})$, we must now assign a cost to visiting a node equal to $2 w_V (T_{ab}, T_{a(b+1)})$ in addition to the cost of each edge (which is equal to $w_H (T_{ab}, T_{(a+1)b}) + w_H (T_{a(b+1)}, T_{(a+1)(b+1)})$).  We must adapt the algorithm slightly: When calculating the cost of a simple path, we include the cost of all the nodes, including the beginning and ending node.  When we calculate the cost of a simple cycle, we include the cost of all nodes, but only count the starting/ending node once.  Note that this creates a distinction between the cost accounting for a simple path and a simple cycle, even when the simple path being considered is actually a cycle.  We need to do this because the simple cycles will be inserted into a longer path at the location of the starting node of the cycle, whose cost is already calculated as part of the path.  The ending of the cycle introduces a second copy of that node into the path, so the cost of that node should be added in.

If there is even one forbidden pairing in the pair of rows, the overall cost must be at least $+8 = -22 + 30$.  Since we will be able to achieve a cost of $-12$, it will be sufficient to consider only paths and pairs of tiles involving allowed pairings.

Consider the graph whose nodes are allowed vertical pairs and edges are allowed horizontal $2 \times 2$ squares.  The nodes of the graph are weighted by twice the cost of the horizontal edge in the pair, and the edges of the graph are weighted by the sum of the cost of the two vertical edges in the square.  The graph breaks down into three connected components.  The first component contains the pairs
\begin{equation}
\twotilesvert{\tileV}{\tileV},\
\twotilesvert{\tileRing}{\tileWhite},\ \twotilesvert{\tileRing}{\tileBlack},\
\twotilesvert{\tileRing}{\tileDark},\ \twotilesvert{\tileRing}{\tileLight},\
\twotilesvert{\tileWhite}{\tileRing},\ \twotilesvert{\tileBlack}{\tileRing},\
\twotilesvert{\tileDark}{\tileRing},\ \twotilesvert{\tileLight}{\tileRing},\
\twotilesvert{\tileCirc}{\tileWhite},\ \twotilesvert{\tileCirc}{\tileDark},\
\twotilesvert{\tileWhite}{\tileCirc},\ \twotilesvert{\tileDark}{\tileCirc},\
\twotilesvert{\tileWhite}{\tileDark},\ \twotilesvert{\tileBlack}{\tileLight},\
\twotilesvert{\tileDark}{\tileWhite},\ \twotilesvert{\tileLight}{\tileBlack}.
\end{equation}
The second component contains the pairs
\begin{equation}
\twotilesvert{\tileC}{\tileV},\
\twotilesvert{\tileH}{\tileWhite},\ \twotilesvert{\tileH}{\tileBlack},\
\twotilesvert{\tileH}{\tileDark},\ \twotilesvert{\tileH}{\tileLight},\
\twotilesvert{\tileH}{\tileCirc}.
\end{equation}
The third component is just the vertical mirror image of the second component.  In the second and third components, there are no paths with total cost $w'$ less than or equal to $0$: The vertical pairs which include \tileH\ have cost per node of $+12$ or $+14$, and the minimal cost of an edge is $-11$.  The vertical pair which includes \tileC\ only has node cost $0$, but all the edges leaving it also have positive cost ($+6$ or $+7$).  We will be able to achieve path costs less than $0$ using the first component, so the second and third components are only useful at the top and bottom edges of the $N \times N$ grid.

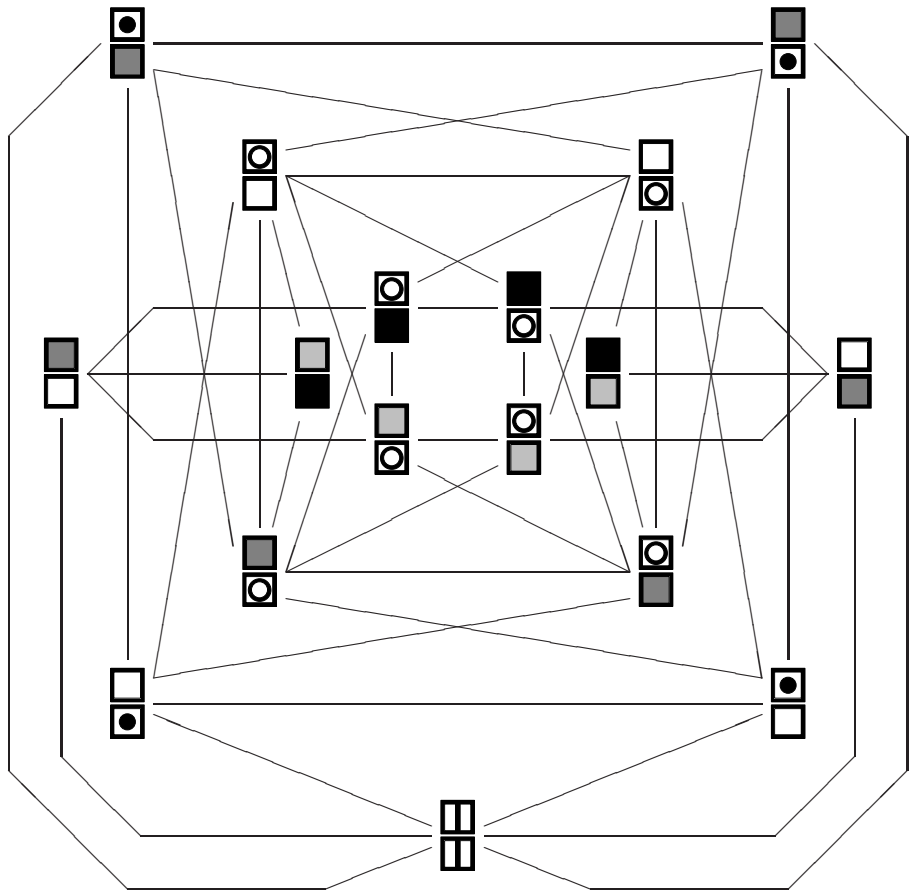
\begin{figure}
\begin{centering}
\begin{picture}(400,400)

\put(200,25){\makebox(0,0){$\twotilesvert{\tileV}{\tileV}$}}

\put(75,75){\makebox(0,0){$\twotilesvert{\tileWhite}{\tileCirc}$}}
\put(325,75){\makebox(0,0){$\twotilesvert{\tileCirc}{\tileWhite}$}}
\put(75,325){\makebox(0,0){$\twotilesvert{\tileCirc}{\tileDark}$}}
\put(325,325){\makebox(0,0){$\twotilesvert{\tileDark}{\tileCirc}$}}

\put(125,125){\makebox(0,0){$\twotilesvert{\tileDark}{\tileRing}$}}
\put(275,125){\makebox(0,0){$\twotilesvert{\tileRing}{\tileDark}$}}
\put(125,275){\makebox(0,0){$\twotilesvert{\tileRing}{\tileWhite}$}}
\put(275,275){\makebox(0,0){$\twotilesvert{\tileWhite}{\tileRing}$}}

\put(175,175){\makebox(0,0){$\twotilesvert{\tileLight}{\tileRing}$}}
\put(225,175){\makebox(0,0){$\twotilesvert{\tileRing}{\tileLight}$}}
\put(225,225){\makebox(0,0){$\twotilesvert{\tileBlack}{\tileRing}$}}
\put(175,225){\makebox(0,0){$\twotilesvert{\tileRing}{\tileBlack}$}}

\put(50,200){\makebox(0,0){$\twotilesvert{\tileDark}{\tileWhite}$}}
\put(145,200){\makebox(0,0){$\twotilesvert{\tileLight}{\tileBlack}$}}
\put(255,200){\makebox(0,0){$\twotilesvert{\tileBlack}{\tileLight}$}}
\put(350,200){\makebox(0,0){$\twotilesvert{\tileWhite}{\tileDark}$}}

\put(185,175){\line(1,0){30}}
\put(185,225){\line(1,0){30}}
\put(175,192){\line(0,1){16}}
\put(225,192){\line(0,1){16}}

\put(185,235){\line(2,1){80}}
\put(235,185){\line(1,3){30}}
\put(165,215){\line(-1,-3){30}}
\put(215,165){\line(-2,-1){80}}

\put(185,165){\line(2,-1){80}}
\put(235,215){\line(1,-3){30}}
\put(165,185){\line(-1,3){30}}
\put(215,235){\line(-2,1){80}}

\put(85,315){\line(1,-6){30}}
\put(85,315){\line(6,-1){180}}
\put(135,115){\line(6,-1){180}}
\put(285,265){\line(1,-6){30}}

\put(85,85){\line(1,6){30}}
\put(85,85){\line(6,1){180}}
\put(315,315){\line(-1,-6){30}}
\put(315,315){\line(-6,-1){180}}

\put(140,218){\line(-1,4){10}}
\put(140,182){\line(-1,-4){10}}
\put(135,200){\line(-1,0){75}}

\put(260,218){\line(1,4){10}}
\put(260,182){\line(1,-4){10}}
\put(265,200){\line(1,0){75}}

\put(60,200){\line(1,1){25}}
\put(85,225){\line(1,0){80}}
\put(60,200){\line(1,-1){25}}
\put(85,175){\line(1,0){80}}

\put(340,200){\line(-1,1){25}}
\put(315,225){\line(-1,0){80}}
\put(340,200){\line(-1,-1){25}}
\put(315,175){\line(-1,0){80}}

\put(85,325){\line(1,0){230}}
\put(75,308){\line(0,-1){216}}
\put(85,75){\line(1,0){230}}
\put(325,308){\line(0,-1){216}}

\put(135,125){\line(1,0){130}}
\put(135,275){\line(1,0){130}}
\put(125,142){\line(0,1){116}}
\put(275,142){\line(0,1){116}}

\put(190,29){\line(-5,2){105}}
\put(210,29){\line(5,2){105}}

\put(190,25){\line(-1,0){110}}
\put(80,25){\line(-1,1){30}}
\put(50,55){\line(0,1){128}}
\put(210,25){\line(1,0){110}}
\put(320,25){\line(1,1){30}}
\put(350,55){\line(0,1){128}}

\put(190,21){\line(-5,-2){40}}
\put(150,5){\line(-1,0){75}}
\put(75,5){\line(-1,1){45}}
\put(30,50){\line(0,1){240}}
\put(30,290){\line(1,1){35}}

\put(210,21){\line(5,-2){40}}
\put(250,5){\line(1,0){75}}
\put(325,5){\line(1,1){45}}
\put(370,50){\line(0,1){240}}
\put(370,290){\line(-1,1){35}}

\end{picture}
\caption{The first component of the graph of vertical pairs of tiles.}
\label{fig:tworowsmain}
\end{centering}
\end{figure}

Let us consider in more detail the first component, pictured in figure~\ref{fig:tworowsmain}.  First, the costs of the nodes: The pair with two \tileV\ tiles has node cost $-22$.  The $12$ nodes which involve a \tileRing\ or \tileCirc\ have node cost $+2$, and the remaining $4$ nodes have cost $0$.  However, the edges from the node with two \tileV\ tiles have large positive cost: $+12$ to connect to the nodes involving both \tileWhite\ and \tileDark, and $+13$ to connect to the $4$ nodes involving \tileCirc.  (It is not adjacent to the other nodes.)  The pairs involving \tileCirc\ or \tileRing\ have edge costs $+1$ (to connect to a node which does not involve \tileCirc, \tileRing, or \tileV), $+2$ (to connect to another node with a \tileCirc\ or \tileRing), or $+13$ (to connect to the two-\tileV\ node, which is only possible for the four \tileCirc\ nodes).  The remaining four nodes involve either the pair (\tileWhite, \tileDark) or the pair (\tileBlack, \tileLight).  Besides the connections listed above, they also connect to one of the other four nodes in this class with an edge of cost $0$.

Let us first consider the cost of simple cycles.  We are particularly interested in minimum-cost simple cycles of various lengths.  For a simple cycle, we count the node cost for the beginning/ending node only once.  It is thus easy to see that we cannot achieve a negative cost simple cycle --- the benefit (negative cost) of having a (\tileV, \tileV) node is outweighed by the high positive cost of the edges in and out of that node.  Therefore, the minimum cost of a cycle is $0$, and this is easy to achieve for a cycle of length 2, using a (\tileWhite, \tileDark) node and the matching (\tileBlack, \tileLight) node.

Thus, if our path contains any (\tileWhite, \tileDark) or (\tileBlack, \tileLight) node, we can extend it to a path whose length has the same cost and the same parity (even or odd) by inserting copies of an appropriate cost $0$ cycle.  In order to cover paths whose length might have a different parity than the simple path we started with, we will also need to insert a simple cycle of odd length.  There are no other cost $0$ simple cycles, so the minimum-cost odd length cycle will have positive cost, and we will only want to use at most one of them.  It is not possible to build an odd cycle of allowed transitions using just (\tileV, \tileV), (\tileBlack, \tileLight), and (\tileWhite, \tileDark) nodes, so we will need to include at least one node with a \tileCirc\ or \tileRing\ tile.  Since there is a cost associated to the nodes involving \tileCirc\ or \tileRing\ tiles, we would like to minimize the number of such nodes we use.  It is straightforward to see there are no odd cycles with only one \tileCirc\ or \tileRing\ node, so we will want cycles which use two of these nodes.

Let us first consider odd simple cycles that contain a (\tileBlack, \tileLight) or (\tileWhite, \tileDark) node.  In this case, we can achieve a cycle with total cost $+8$ and length $3$ by starting with a (\tileBlack, \tileLight) or (\tileWhite, \tileDark) node, followed by two nodes involving \tileRing, and then back to the original node.  There are a number of allowed paths of this form, such as
\begin{equation}
\twotilesvert{\tileWhite}{\tileDark} \ \twotilesvert{\tileBlack}{\tileRing} \
\twotilesvert{\tileRing}{\tileLight} \ \twotilesvert{\tileWhite}{\tileDark}.
\label{path:middlecorrect}
\end{equation}
There are no other allowed length $3$ cycles starting and ending on a (\tileBlack, \tileLight) or (\tileWhite, \tileDark) node.  There are longer simple cycles, but they also have larger cost.

The other minimal cost odd simple cycle starts on a (\tileV, \tileV) node, connecting from there to any of the four \tileCirc\ nodes, and from there to another \tileCirc\ node (with the \tileCirc\ in the opposite position), then back to the (\tileV, \tileV) node.  The total cost is $+10$.  This type of length 3 cycle costs more than the previously discussed class of length 3 cycles, so would only be useful if we wanted to extend a path that contained no (\tileBlack, \tileLight) or (\tileWhite, \tileDark) nodes.  This will not be needed, so we can ignore this simple cycle.

The next step is to consider minimal cost simple paths.  Our goal is to
achieve an overall negative cost, and it is clear the only way to achieve
that is through the use of (\tileV, \tileV) nodes.  In particular, the only
possible locations for the
(\tileV, \tileV) nodes are the beginning and the ending of our
path, since those are the only locations where the cost of edges in and out
 of the node does not overcome the negative cost of the node itself.
We will assume for now that the (\tileV, \tileV) nodes occur at both ends
of the path and address  the other cases later.

The minimum-cost way to leave a (\tileV, \tileV) node is by connecting to one of the two (\tileWhite, \tileDark) nodes.  Then we can connect directly back to (\tileV, \tileV), for a total cost of $-20$.  Extending this simple path using a cost $0$ length 2 simple cycle, we can create arbitrarily long paths of odd length with cost $-20$.  However, $N$ is even.  By using one of the cost $+8$ length 3 simple cycles, we can create paths of any even length ($N \geq 6$) with total cost $-12$.  This produces the paths of class d in the statement of the lemma.

Next, let us consider even-length simple paths that start and end on (\tileV, \tileV) nodes.  To achieve an even length, we again need to use at least two \tileCirc\ or \tileRing\ nodes, and to minimize the cost, we will want to use exactly two.  Since we have the beginning and ending (\tileV, \tileV) nodes, exactly two \tileCirc\ or \tileRing\ nodes, and cannot repeat any (\tileBlack, \tileLight) or (\tileWhite, \tileDark) node (because we want a \emph{simple} path), we will get a simple path of length $4$, $6$, or $8$.

For length 4, we have the simple path version of the simple cycle of length $3$ starting and ending with (\tileV, \tileV).  As a path, this has total cost $-12$.  However, it contains no (\tileBlack, \tileLight) or (\tileWhite, \tileDark) node, so cannot be extended to a longer path without increasing the cost.

Next, we can consider paths that include a \tileCirc\ tile.  We can indeed achieve a total cost of $-12$ using the paths
\begin{equation}
\twotilesvert{\tileV}{\tileV} \  \twotilesvert{\tileDark}{\tileCirc} \
\twotilesvert{\tileRing}{\tileWhite} \ \twotilesvert{\tileLight}{\tileBlack} \
\twotilesvert{\tileDark}{\tileWhite} \ \twotilesvert{\tileV}{\tileV}
\label{path:endcorrect}
\end{equation}
and
\begin{equation}
\twotilesvert{\tileV}{\tileV} \ \twotilesvert{\tileDark}{\tileCirc} \
\twotilesvert{\tileRing}{\tileDark} \ \twotilesvert{\tileBlack}{\tileLight} \
\twotilesvert{\tileWhite}{\tileDark} \ \twotilesvert{\tileV}{\tileV}
\label{path:endflipped}
\end{equation}
or variations.  These paths start with a \tileV\ node, continue to a \tileCirc\ node, then a \tileRing\ node, then a (\tileBlack, \tileLight) node, and finally a (\tileWhite, \tileDark) node, before returning to \tileV.  (Or we could have the same progression from right to left instead.)  The difference between the two classes is that in the class represented by (\ref{path:endflipped}) (class a in the statement of the lemma), one row (the top one in this case) contains both \tileDark\ and \tileBlack\ tiles, whereas in (\ref{path:endcorrect}) (class b in the statement of the lemma), both rows consist of either all \tileWhite\ and \tileBlack\ tiles or all \tileLight\ and \tileDark\ tiles (plus \tileV, \tileRing, and/or \tileCirc).
These paths all include (\tileBlack, \tileLight) and (\tileWhite, \tileDark) nodes, so can be extended to any even $N$ while maintaining the total cost $-12$.

There are two more types of simple paths of total cost $-12$.  In these, we use two \tileRing\ nodes, and the (\tileV, \tileV) nodes connect to (\tileWhite, \tileDark) nodes.  Either we have just these six nodes, for instance:
\begin{equation}
\twotilesvert{\tileV}{\tileV} \ \twotilesvert{\tileWhite}{\tileDark} \
\twotilesvert{\tileBlack}{\tileRing} \ \twotilesvert{\tileRing}{\tileBlack} \
\twotilesvert{\tileDark}{\tileWhite} \ \twotilesvert{\tileV}{\tileV},
\label{path:middleflipped}
\end{equation}
or we also include the two (\tileBlack, \tileLight) nodes, as in this path:
\begin{equation}
\twotilesvert{\tileV}{\tileV} \ \twotilesvert{\tileWhite}{\tileDark} \ \twotilesvert{\tileBlack}{\tileLight} \
\twotilesvert{\tileWhite}{\tileRing} \ \twotilesvert{\tileRing}{\tileWhite} \
\twotilesvert{\tileLight}{\tileBlack} \ \twotilesvert{\tileDark}{\tileWhite} \ \twotilesvert{\tileV}{\tileV}.
\label{path:middleflipped2}
\end{equation}
Both of these types of simple path give class c in the statement of the lemma.

We also need to consider the case where one or both sides are
forbidden to have \tileV.  When one side is forbidden to have \tileV,
our best strategy is to use the simple path (\tileV, \tileV) followed
by one of the (\tileWhite, \tileDark) nodes, and then extend using the
cost $0$ length $2$ simple cycle.  This path will have a cost of
$-10$. If both sides are forbidden \tileV,
then we should just use the cost $0$ length $2$ simple cycle
everywhere, for an overall cost of $0$.
\end{proof}

Now let us consider the top and bottom rows.  Again, we will consider a pair of rows, but with a new formula for costs to take into account the effect of the edge.  For the first and second rows, we should consider the following formula:
\begin{equation}
w''(R_1, R_2) = 2w(R_1) + w(R_2) + 2\sum_{a=1}^N w_V (T_{a1},T_{a2}).
\label{eq:toprow}
\end{equation}
We can define $w''(R_{N-1},R_N)$ similarly, counting $w(R_N)$ twice.  This formula reflects the appearance of $w(R_1)$ in equation (\ref{eq:costbyrows}).  Indeed, we have that
\begin{equation}
2 w(T) = w''(R_1, R_2) + \sum_{b=2}^{N-2} w'(R_b, R_{b+1}) + w''(R_{N-1}, R_N).
\label{eq:totalrowpairs}
\end{equation}

\begin{lemma}
\label{lemma:reftoprow}
The minimum value of $w''(R_1, R_2)$ is $24-10N$, using a path that has \tileH\ everywhere in the top row, and an alternating pattern of either (\tileLight, \tileDark) or (\tileWhite, \tileBlack) in the second row.

If one corner is required to have a \tileC\ tile, the minimum value of $w''(R_1, R_2)$ is $40-10N$.  This can be achieved by having \tileH\ elsewhere in the top row.  In the second row, under the \tileC\ tile there is a \tileV\ tile.  Elsewhere the second row alternates between either the pair (\tileLight, \tileDark) or the pair (\tileWhite, \tileBlack).

If the top row begins and ends with \tileC, the minimum value of $w''$ is $58-10N$, and this can be achieved only if the top row is \tileH\ in all other locations.  The other row must start and end with \tileV, and in between alternate either \tileBlack\ and \tileWhite\ tiles or \tileLight\ and \tileDark, with a \tileCirc\ tile at one end, just before or after the \tileV\ tile.
\end{lemma}

\begin{proof}[ of lemma]
Returning to the graph of allowed vertical pairs, the second component now becomes the most favorable: The component with \tileH\ at the top can produce a $-\Theta(N)$ total cost for the top pair of rows.  Similarly, the component with \tileH\ at the bottom can produce a $-\Theta(N)$ total cost for the bottom pair of rows.  The component of the graph with \tileH\ at the top is pictured in figure~\ref{fig:tworowsH}.
\begin{figure}
\begin{centering}
\begin{picture}(220,130)
\put(25,60){\makebox(0,0){$\twotilesvert{\tileH}{\tileLight}$}}
\put(70,60){\makebox(0,0){$\twotilesvert{\tileH}{\tileDark}$}}
\put(35,60){\line(1,0){25}}

\put(110,90){\makebox(0,0){$\twotilesvert{\tileC}{\tileV}$}}
\put(80,60){\line(2,3){20}}
\put(120,90){\line(2,-3){20}}

\put(110,30){\makebox(0,0){$\twotilesvert{\tileH}{\tileCirc}$}}
\put(80,60){\line(2,-3){20}}
\put(120,30){\line(2,3){20}}
\put(110,47){\line(0,1){36}}

\put(150,60){\makebox(0,0){$\twotilesvert{\tileH}{\tileWhite}$}}
\put(195,60){\makebox(0,0){$\twotilesvert{\tileH}{\tileBlack}$}}
\put(160,60){\line(1,0){25}}

\put(5,60){\makebox(0,0){$+12$}}
\put(65,82){\makebox(0,0){$+12$}}
\put(110,8){\makebox(0,0){$+14$}}
\put(110,112){\makebox(0,0){$0$}}
\put(155,82){\makebox(0,0){$+12$}}
\put(215,60){\makebox(0,0){$+12$}}

\put(47,65){\makebox(0,0){$-22$}}
\put(83,35){\makebox(0,0){$-21$}}
\put(85,87){\makebox(0,0){$+6$}}
\put(135,87){\makebox(0,0){$+6$}}
\put(120,60){\makebox(0,0){$+7$}}
\put(135,35){\makebox(0,0){$-21$}}
\put(172,65){\makebox(0,0){$-22$}}
\end{picture}
\caption{A component of the graph of vertical pairs of tiles.  The edge costs are calculated for the $w''$ formula appropriate for the top two rows.}
\label{fig:tworowsH}
\end{centering}
\end{figure}
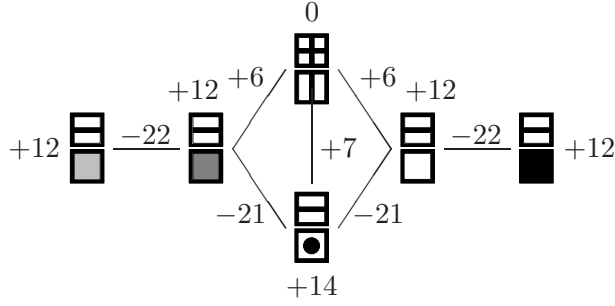
The node and edge weights in figure~\ref{fig:tworowsH} are calculated according to equation~(\ref{eq:toprow}).

There are two minimal-cost simple cycles, both of length 2 and cost $-20$: One contains (\tileH, \tileLight) and (\tileH, \tileDark), the other (\tileH, \tileWhite) and (\tileH, \tileBlack).  There are two minimal-cost odd-length cycles, both of length 3 and cost $+18$.  They involve (\tileC, \tileV), (\tileH, \tileCirc), and either (\tileH, \tileDark) or (\tileH, \tileWhite).

The minimal-cost simple paths have length $2$ and cost $+2$, e.g., (\tileH, \tileLight), (\tileH, \tileDark).  This can be extended to even $N$ with a cost $-20$ length 2 simple cycle to get a total cost of $22-10N$.

If we insist that the simple path start at (\tileC, \tileV), then the minimal-cost length $2$ simple path is (\tileC, \tileV) followed by (\tileH, \tileDark) or (\tileH, \tileWhite), which has cost $+18$.  The minimal-cost length $3$ simple path adds on (\tileH, \tileLight) or (\tileH, \tileDark) on the end, giving cost $+8$.  We can than get a minimal-cost path of length $N$ (for $N$ even) by taking the length-2 simple path and extending with a length-2 simple cycle, for total cost $38-10N$.

If we insist that our starting simple path both start and end at (\tileC, \tileV), our best strategy is to use the length 3 simple cycle as our simple path.  As a simple path, it has length 4 and also has cost $+18$.  We then get an even length path of total cost $58 - 10N$.
\end{proof}

As we can see, there is a mismatch between the optimal tilings of the top two rows and the middle pairs of rows.  The optimal tilings for a middle row have \tileV\ on both ends, but that would produce forbidden transitions (and a corresponding suboptimal pair of middle rows) if we tried also to put in the optimal tiling for the top two rows.

We wish to apply equation (\ref{eq:totalrowpairs}) to minimize the total cost.  We consider various different combinations of tiling the top (and bottom) pairs of rows and middle pairs of rows.

There are three combinations which involve no forbidden pairs.
\begin{enumerate}
\item The top and bottom pairs of rows use optimal tilings, and the middle pairs of rows use tilings which do not include \tileV\ on either end.  In that case, we find the total cost for Layer $1$ is $[2 (22-10N) + (N-3)(0)]/2 = 22 - 10N$.
\item The top and bottom rows each have a \tileC\ in one corner but not in the other.  The middle rows have \tileV\ on one end (between the \tileC\ tiles), but not in the other.  The optimal tiling subject to these constraints has a cost of $[2(38-10N) + (N-3)(-10)]/2 = 53-15N$.
\item We have \tileC\ tiles in all four corners and the middle pairs of rows use optimal tilings.  The total cost is then $[2 (58-10N) + (N-3)(-12)]/2 = 76 - 16N$.
\end{enumerate}
For sufficiently large $N$, the last choice is optimal.

If we were to use forbidden pairs, we could combine, for instance, an optimal tiling of the top pair of rows with an optimal tiling of most of the middle pairs of rows.  However, for that particular combination, we would need to use at least two forbidden pairs, and the cost of doing that is greater than the cost of switching the top pair of rows to a tiling with \tileC\ in both corners.  Similarly for other combinations that involve forbidden pairs.

Now let us investigate whether it is actually possible to achieve the cost $76-16N$ using the third combination (\tileC\ in all four corners).  By lemma~\ref{lemma:refminrowpair}, we have vertical lines of \tileV\ tiles on the left and right sides of the grid, and with a path of \tileCirc\ or \tileRing\ tiles somewhere in between them.  The path is diagonal everywhere, but could potentially move back and forth to the left and right.  If the dividing line is adjacent at some point to an end \tileV\ tile, it is a \tileCirc\ tile; otherwise, it is a \tileRing\ tile.  The dividing line splits the remaining locations into two or more connected components, each of which consists of rows of alternating \tileLight\ and \tileDark\ tiles and rows of alternating \tileWhite\ and \tileBlack\ tiles, with the two types of rows alternating as well.

In summary, we have \tileC\ at all four corners, \tileH\ everywhere else in the top and bottom rows, and \tileV\ everywhere on the left and right sides, except at the corners.  There is exactly one \tileCirc\ or \tileRing\ tile in each row.  In the second and $(N-1)$th rows, it must be a \tileCirc\ tile and must be located in either the second or $(N-1)$th column. Elsewhere, the \tileCirc\ or \tileRing\ tiles from two adjacent rows must be located diagonally from each other.  Since $N$ is even, it is not possible to do all of this if the \tileCirc\ tiles from the top and bottom rows are in the same column.  Therefore, the \tileCirc\ and \tileRing\ tiles must form a diagonal line reaching from one corner (e.g., the $(2,2)$ location) to the opposite corner (e.g., the $(N-1, N-1)$ location) of the interior.  The ends of the line are \tileCirc\ tiles, and the rest of the diagonal line is composed of \tileRing\ tiles.  On each side of the line, the interior of the grid is tiled by rows alternating between \tileBlack\ and \tileWhite\ tiles and \tileLight\ and \tileDark\ tiles.  Columns $2$ and $N-1$ contain only \tileWhite\ and \tileDark\ tiles, as well as \tileCirc\ and \tileH.

We thus get $8$ possible minimal cost solutions for Layer $1$.  There is the
arrangement of figure~\ref{fig:weightedreflection1}, and three rotations of it.
These correspond to cases b and d in Lemma \ref{lemma:refminrowpair}.
Alternatively, we can take the upper left corner or lower right corner of
figure~\ref{fig:weightedreflection1} and rotate it $180^\circ$ around the
center of the grid to fill in the other corner.  There is one rotation for
each of those solutions as well, with the diagonal line going from the upper
left to the lower right instead. These correspond to cases a and c in
Lemma \ref{lemma:refminrowpair}.
 The first four solutions, rotations of
figure~\ref{fig:weightedreflection1}, have a local breaking of reflection
symmetry in the vicinity of the diagonal line: In the orientation of
figure~\ref{fig:weightedreflection1}, immediately to the left of and
above the diagonal line, we have \tileBlack\ and \tileDark\ tiles,
whereas immediately to the right of and below the diagonal line, we
have \tileWhite\ and \tileLight\ tiles.  In contrast, while the other
four solutions also break vertical and horizontal reflection symmetry,
there is no local, translationally-invariant rule anywhere in the grid
that can allow us to distinguish the directions.

{\bf Layer 2:}

Layer $2$ will have only $3$ types of tile: \tileHH, \tileHHV, \tileCrosshatch.  The tiling rules will only specify permitted or forbidden pairs of adjacent tiles; forbidden pairs have a cost of $+30$, as before, and permitted pairs have a cost of $0$.  When the Layer $1$ tile is \tileV\ or \tileC, the Layer $2$ tile must be \tileHH.  For any other Layer $1$ tile, any of the three Layer $2$ tiles is possible.

For the vertical tiling rules, we only allow each type of tile to be adjacent to itself.  Horizontally, each type of tile is forbidden to be adjacent to itself.  Thus, to avoid any forbidden pairs of tiles, Layer $2$ must consist of vertical lines, all of one kind of tile, and no two adjacent lines can be the same.  Our goal is to have the three types of tile cycle: \tileHH\ followed by \tileHHV\ followed by \tileCrosshatch, and then back to \tileHH.  Then by looking at Layer $2$ for any horizontally adjacent pair of locations in the grid, we can distinguish left from right.

To achieve this, we determine whether two different Layer $2$ tiles can be adjacent horizontally depending on the Layer $1$ tiles underlying them.  If neither of the Layer $1$ tiles is a \tileCirc\ or \tileRing, any pair of Layer $2$ tiles is permitted.  If one of the Layer $1$ tiles is a \tileV, \tileC, or \tileH\ tile (although the latter two are forbidden in Layer $1$ to be horizontally adjacent to \tileCirc\ or \tileRing), then any pair of Layer $2$ tiles is permitted.  It is forbidden in Layer $1$ to have two \tileCirc\ or \tileRing\ tiles adjacent, but if it happens, we allow any pair of Layer $2$ tiles.  Otherwise, the following combinations are allowed.  We use brackets to indicate the pair of Layer $1$ and Layer $2$ tile, so $[\tileCirc \ \tileHH ]$ indicates a \tileCirc\ on Layer $1$ and \tileHH\ on Layer $2$ in a given location.  In all of the following pairs, we use \tileRing\ for one of the Layer $1$ tiles, but we have precisely the same rules if \tileCirc\ is substituted for \tileRing.
\begin{equation}
\begin{tabular}{cc|cc}
\multicolumn{2}{c|}{\bf{Allowed}} &  \multicolumn{2}{c}{\bf{Forbidden}} \\
$[\tileRing \ \tileHH ] \ [\tileWhite \ \tileHHV ]$ & $[\tileBlack \ \tileHH ] \ [\tileRing \ \tileHHV ]$ &
        $[\tileRing \ \tileHH ] \ [\tileBlack \ \tileHHV ]$ & $[\tileWhite \ \tileHH ] \ [\tileRing \ \tileHHV ]$ \\
$[\tileRing \ \tileHH ] \ [\tileLight \ \tileHHV ]$ & $[\tileDark \ \tileHH ] \ [\tileRing \ \tileHHV ]$ &
        $[\tileRing \ \tileHH ] \ [\tileDark \ \tileHHV ]$ & $[\tileLight \ \tileHH ] \ [\tileRing \ \tileHHV ]$ \\
$[\tileRing \ \tileHHV ] \ [\tileWhite \ \tileCrosshatch ]$ & $[\tileBlack \ \tileHHV ] \ [\tileRing \ \tileCrosshatch ]$ &
        $[\tileRing \ \tileHHV ] \ [\tileBlack \ \tileCrosshatch ]$ & $[\tileWhite \ \tileHHV ] \ [\tileRing \ \tileCrosshatch ]$ \\
$[\tileRing \ \tileHHV ] \ [\tileLight \ \tileCrosshatch ]$ & $[\tileDark \ \tileHHV ] \ [\tileRing \ \tileCrosshatch ]$ &
        $[\tileRing \ \tileHHV ] \ [\tileDark \ \tileCrosshatch ]$ & $[\tileLight \ \tileHHV ] \ [\tileRing \ \tileCrosshatch ]$ \\
$[\tileRing \ \tileCrosshatch ] \ [\tileWhite \ \tileHH ]$ & $[\tileBlack \ \tileCrosshatch ] \ [\tileRing \ \tileHH ]$ &
        $[\tileRing \ \tileCrosshatch ] \ [\tileBlack \ \tileHH ]$ & $[\tileWhite \ \tileCrosshatch ] \ [\tileRing \ \tileHH ]$ \\
$[\tileRing \ \tileCrosshatch ] \ [\tileLight \ \tileHH ]$ & $[\tileDark \ \tileCrosshatch ] \ [\tileRing \ \tileHH ]$ &
        $[\tileRing \ \tileCrosshatch ] \ [\tileDark \ \tileHH ]$ & $[\tileLight \ \tileCrosshatch ] \ [\tileRing \ \tileHH ]$ \\
\end{tabular}
\label{eq:layer2rules}
\end{equation}

If Layer $1$ is the tiling of figure~\ref{fig:weightedreflection1}, then the middle $N-2$ columns of Layer $2$ will, according the rules of (\ref{eq:layer2rules}), cycle between columns of \tileHH\ tiles, \tileHHV\ tiles, and \tileCrosshatch\ tiles, in that order from left to right.  Columns $1$ and $N$, since Layer $1$ contains \tileV\ tiles there, must be \tileHH\ tiles on Layer $2$.  Since adjacent columns must have different Layer $2$ tiles, column $2$ cannot be \tileHH\ tiles.  If column $2$ contained \tileCrosshatch\ tiles, since $N \equiv 1 \bmod 3$, then column $N-1$ would contain \tileHH\ tiles, which is forbidden, since column $N$ also contains \tileHH\ tiles.  Thus, column $2$ contains \tileHHV\ tiles and column $N-1$ contains \tileCrosshatch\ tiles, so the whole grid cycles between the three types of tiles for Layer $2$, as pictured in figure~\ref{fig:weightedreflection2}.  Similarly if Layer $1$ is a rotation of figure~\ref{fig:weightedreflection1}.
\begin{figure}
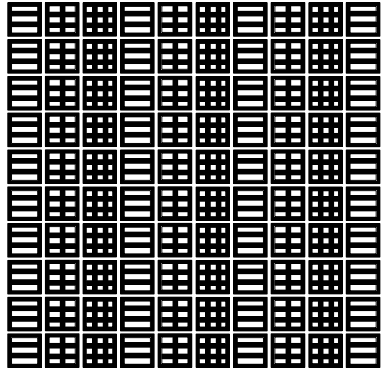

\begin{centering}
\begin{tabular}{c@{\extracolsep{0.1em}}c@{}c@{}c@{}c@{}c@{}c@{}c@{}c@{}c}
\tileHH & \tileHHV & \tileCrosshatch & \tileHH & \tileHHV & \tileCrosshatch & \tileHH & \tileHHV & \tileCrosshatch & \tileHH \\
\tileHH & \tileHHV & \tileCrosshatch & \tileHH & \tileHHV & \tileCrosshatch & \tileHH & \tileHHV & \tileCrosshatch & \tileHH \\
\tileHH & \tileHHV & \tileCrosshatch & \tileHH & \tileHHV & \tileCrosshatch & \tileHH & \tileHHV & \tileCrosshatch & \tileHH \\
\tileHH & \tileHHV & \tileCrosshatch & \tileHH & \tileHHV & \tileCrosshatch & \tileHH & \tileHHV & \tileCrosshatch & \tileHH \\
\tileHH & \tileHHV & \tileCrosshatch & \tileHH & \tileHHV & \tileCrosshatch & \tileHH & \tileHHV & \tileCrosshatch & \tileHH \\
\tileHH & \tileHHV & \tileCrosshatch & \tileHH & \tileHHV & \tileCrosshatch & \tileHH & \tileHHV & \tileCrosshatch & \tileHH \\
\tileHH & \tileHHV & \tileCrosshatch & \tileHH & \tileHHV & \tileCrosshatch & \tileHH & \tileHHV & \tileCrosshatch & \tileHH \\
\tileHH & \tileHHV & \tileCrosshatch & \tileHH & \tileHHV & \tileCrosshatch & \tileHH & \tileHHV & \tileCrosshatch & \tileHH \\
\tileHH & \tileHHV & \tileCrosshatch & \tileHH & \tileHHV & \tileCrosshatch & \tileHH & \tileHHV & \tileCrosshatch & \tileHH \\
\tileHH & \tileHHV & \tileCrosshatch & \tileHH & \tileHHV & \tileCrosshatch & \tileHH & \tileHHV & \tileCrosshatch & \tileHH \\
\end{tabular}
\caption{The permitted tiling of Layer $2$ for a $10 \times 10$ grid when Layer $1$ is the optimal tiling of figure~\ref{fig:weightedreflection1}.  Notice that at every location, it is possible to distinguish left and right by looking at the adjacent tiles.}
\label{fig:weightedreflection2}
\end{centering}
\end{figure}

On the other hand, if Layer $1$ is arranged according to one of the other optimal tilings, there will be no allowed tiling of Layer $2$.  This is because in two adjacent rows, we will have a configuration such as
\begin{equation}
\begin{array}{c@{\extracolsep{0.1em}}c}
\tileBlack & \tileRing \\
\tileRing & \tileBlack
\end{array}
\end{equation}
There is no way to tile Layer $2$ for these four locations consistent with both the rules of (\ref{eq:layer2rules}) and the constraint that only identical tiles can be vertically adjacent in Layer $2$.  Thus, in order to achieve a cost of $76 - 16N$, Layer $1$ must be in a tiling corresponding to one of the four rotations of figure~\ref{fig:weightedreflection1}.

{\bf Layer 3:}

Layer $3$ is very similar to Layer $2$.  There are again $3$ kinds of tiles: \tileVVrev, \tileHVVrev, and \tileCrosshatchrev.  When the Layer $1$ tile is \tileH, the Layer $3$ tile must be \tileVVrev; otherwise, any Layer $3$ tile is allowed.  The adjacency rules are effectively the same as Layer $2$, but with horizontal and vertical exchanged, and with \tileVVrev, \tileHVVrev, and \tileCrosshatchrev\ substituted for \tileHH, \tileHHV, and \tileCrosshatch, respectively.  Thus, the only permitted tiling of Layer $3$ has rows consisting of a single type of tile, with the three kinds of rows cycling.

{\bf Main Layers:}

To tile the main layers, we just use the rules from section~\ref{sec:tiling}.  Those rules distinguish between up and down and between left and right, so to achieve that, we look at the tiles in Layers $2$ and $3$.  If Layer $2$ is a \tileHH\ for one location and \tileHHV\ for the second location, we consider the first location to be to the left of the second location.  Similarly, we consider a \tileHHV\ Layer $2$ tile to be to the left of a \tileCrosshatch\ Layer $2$ tile, and a \tileCrosshatch\ Layer $2$ tile to be to the left of a \tileHH\ Layer $2$ tile.  We consider a \tileVVrev\ Layer $3$ tile to be above a \tileHVVrev\ Layer $3$ tile, a \tileHVVrev\ Layer $3$ tile to be above a \tileCrosshatchrev\ Layer $3$ tile, and a \tileCrosshatchrev\ Layer $3$ tile to be above a \tileVVrev\ Layer $3$ tile.  We can set the corner boundary conditions for the main layers by looking at Layer $1$: When the Layer $1$ tile is a \tileC\ tile, the main layer tile must be a corner tile.

If there is a valid tiling of the main layer, which occurs when the instance is a ``yes'' instance, then we achieve an overall cost $76 - 16N$.  If the instance is a ``no'' instance, then if Layer $1$ has an optimal arrangement, and Layers $2$ and $3$ are in the allowed tiling consistent with Layer $1$, then there cannot be a valid tiling of the main layer.  Thus, for a ``no'' instance, the cost must be greater than $76 - 16N$.

\end{proof}

\subsubsection{Periodic Boundary Conditions}

Now we turn to WEIGHTED \tiling\ with reflection symmetry and
periodic boundary conditions.  In this case, we are able to prove results for both constant and linear cost functions:

\begin{theorem}
For WEIGHTED \tiling\ with reflection symmetry and
periodic boundary conditions:
\begin{itemize}
\item When the cost function $p(N)$ is a constant, independent of $N$, the problem is in \Pclass.  In particular, there exist $N_e, N_o \in \integer^+ \cup \{\infty\}$ such that for even $N \geq N_e$ or odd $N \geq N_o$, a valid tiling exists, while for even $N$, $2c < N < N_e$ and odd $N$, $2c < N < N_o$, there is no tiling.  $N_e$ is computable.

\item When the cost function $p(N)$ is linear in $N$, the problem is \NEXP-complete.

\item When $N$ is even, the problem is in \Pclass\ for any cost function $p(N)$.

\end{itemize}

\end{theorem}

\begin{proof}

{\bf Constant cost function:}
$p(N) = c$, some constant independent of $N$.  Let us suppose we have a tiling of the $N \times N$ grid with
periodic boundary conditions with total cost $c' \leq c$.
Let the rows of this tiling be $R_1, \ldots, R_N$, and
because of the periodic boundary conditions, let $R_{N+1} = R_1$.
Let $w(R_a) = \sum_b w_H (T_{ab}, T_{a(b+1)})$, where $T_{ab}$ is the tile in location $(a,b)$, and let $w(R_a, R_{a+1}) = \sum_b w_V(T_{ab}, t_{(a+1)b})$.  Then $c' = \sum_a (w(R_a) +
w(R_a, R_{a+1}))$.  Notice that if we duplicate two adjacent rows, say $R_a$ and $R_{a+1}$, then the total cost becomes $c' + w(R_a) + w(R_{a+1}) + 2
w(R_a, R_{a+1})$ (because of the reflection symmetry).  Now,
\begin{equation}
\sum_{a=1}^N \left[  w(R_a) + w(R_{a+1}) + 2 w(R_a, R_{a+1}) \right] = 2 \sum_{a=1}^N \left[ w(R_a) + w(R_a, R_{a+1}) \right] = 2c'.
\end{equation}
Thus, there must exist some $a$ for which
\begin{equation}
w(R_a) + w(R_{a+1}) + 2 w(R_a, R_{a+1}) \leq 2c'/N.
\end{equation}
When $c'$ is a constant, for $N > 2c'$, that means there is some $a$
for which
\begin{equation}
w(R_a) + w(R_{a+1}) + 2 w(R_a, R_{a+1}) \leq 0,
\end{equation}
since the
weights are integers.%
\footnote{If we generalize the definition of WEIGHTED \tiling\ to allow rational weights, we can simply rescale to get integer weights.  If we allow irrational weights, this argument fails, but it is possible to prove the same result using the ideas presented in the even $N$ case below. }
Then we can duplicate rows $R_a$ and $R_{a+1}$
without increasing the cost.  Similarly, there are two columns which
we can duplicate without increasing the cost.  Thus, we also have a
tiling for the $(N+2) \times (N+2)$ square grid with cost $c^{\prime\prime} \leq c' \leq c$.
The only difference from theorem~\ref{theorem:symmetry} is that there might be
some small exceptions with $N \leq 2c$ which have valid tilings but cannot be extended.
$N_e$ then is the smallest even value for $N$ such that $N_e \ge 2c$ and there is
a valid tiling of the $N_e \times N_e$ grid of weight $c' \le c$.
Similarly, $N_o$ is the smallest value for $N$ such that $N_o \ge 2c$ and there is
a valid tiling of the $N_o \times N_o$ grid of weight $c' \le c$.

Note that this argument requires that the cost $p(N)$ be a constant.  It fails if $p(N)$
grows at least linearly with $N$.

\medskip

{\bf Even $N$:}
When $N$ is even, we can easily compute the exact minimal cost achievable.  Consider all
possible $2 \times 2$ squares of tiles, and compute the cost of each by adding the costs of
the four adjacent pairs in the square.  Given any possible tiling of the periodic $N \times N$
grid of total cost $c'$, let $S_{a,b}$ be the $2 \times 2$ square whose top left corner is in
location $(a,b)$, and let $w(S_{a,b})$ be the cost of $S_{a,b}$.  Then $\sum_{a,b} w(S_{a,b}) = 2c'$.
Therefore, the minimal possible cost achievable for the $N \times N$ grid is $N^2 w/2$, where $w$
is the minimal cost of any possible $2 \times 2$ square.  When $N$ is even, this is actually achievable
by repeating a minimal cost square $N/2$ times in each direction.  The squares with their top left
corner in location $(a,b)$, with $a,b$ even, are exactly the square that we chose to repeat, and the
squares in other locations are reflections of that square, which have the same cost.

\medskip

{\bf Linear cost function:}

The construction uses similar ideas to the case with open boundary conditions, but differs in some details.

The argument now uses $4$ layers of tiles, plus a main layer implementing the tiling rules of section~\ref{sec:tiling}.  The first three layers work in much the same way as the three layers discussed in the previous proof for WEIGHTED \tiling\ with open boundary conditions.  The fourth layer defines borders that can be used to set the boundary conditions in the main layer.  We will consider odd $N$, with $N$ divisible by $3$ (so $N \equiv 3 \bmod 6$).  The allowed cost is $6N-2$.

{\bf Layer $1$:}

Layer $1$ uses tiles \tileWhite, \tileBlack, \tileLight, \tileDark, \tileH, \tileHrev, and \tileRing.  The tiles \tileWhite, \tileBlack, \tileLight, \tileDark, and \tileRing\ serve much the same purpose as in the proof for open boundary conditions, and indeed have the same adjacency rules between them.  Only the adjacency rules for \tileH\ and \tileHrev\ are new.

The desired optimal tiling for Layer $1$ has, for all but one row, a single \tileRing\ tile interrupting an alternating line of either \tileWhite\ and \tileBlack\ or \tileLight\ and \tileDark\ (with the two types of row alternating).  The remaining row consists of alternating \tileH\ and \tileHrev\ tiles, also interrupted by one \tileRing\ tile.  The \tileRing\ tiles form a diagonal line circling the torus.  An example of the optimal tiling is given in figure~\ref{fig:periodicweightedreflection}.

\begin{figure}
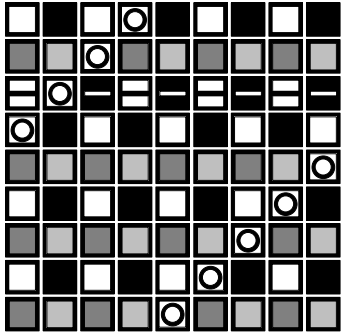

\begin{centering}
\begin{tabular}{c@{\extracolsep{0.1em}}c@{}c@{}c@{}c@{}c@{}c@{}c@{}c}
\tileWhite & \tileBlack & \tileWhite & \tileRing  & \tileBlack & \tileWhite & \tileBlack & \tileWhite & \tileBlack \\
\tileDark  & \tileLight & \tileRing  & \tileDark  & \tileLight & \tileDark  & \tileLight & \tileDark  & \tileLight \\
\tileH     & \tileRing  & \tileHrev  & \tileH     & \tileHrev  & \tileH     & \tileHrev  & \tileH     & \tileHrev  \\
\tileRing  & \tileBlack & \tileWhite & \tileBlack & \tileWhite & \tileBlack & \tileWhite & \tileBlack & \tileWhite \\
\tileDark  & \tileLight & \tileDark  & \tileLight & \tileDark  & \tileLight & \tileDark  & \tileLight & \tileRing  \\
\tileWhite & \tileBlack & \tileWhite & \tileBlack & \tileWhite & \tileBlack & \tileWhite & \tileRing  & \tileBlack \\
\tileDark  & \tileLight & \tileDark  & \tileLight & \tileDark  & \tileLight & \tileRing  & \tileDark  & \tileLight \\
\tileWhite & \tileBlack & \tileWhite & \tileBlack & \tileWhite & \tileRing  & \tileBlack & \tileWhite & \tileBlack \\
\tileDark  & \tileLight & \tileDark  & \tileLight & \tileRing  & \tileDark  & \tileLight & \tileDark  & \tileLight
\end{tabular}
\caption{An optimal tiling for the periodic boundary condition case with $N = 9$.}
\label{fig:periodicweightedreflection}
\end{centering}
\end{figure}

To achieve this, we use the adjacency rules given in table~\ref{table:reflectionperiodic}.

\begin{table}
\begin{centering}
\begin{tabular}{c|rrrrrrr}
  \multicolumn{8}{c}{Horizontal tiling rules} \\
              & $\tileH$ & $\tileHrev$ & $\tileBlack$ & $\tileWhite$ & $\tileDark$ & $\tileLight$ & $\tileRing$ \\
 \hline
$\tileH$      &    30    &      0      &      30      &      30      &     30      &     30       &      1      \\
$\tileHrev$   &     0    &     30      &      30      &      30      &     30      &     30       &      1      \\
$\tileBlack$  &    30    &     30      &      30      &       0      &     30      &     30       &      1      \\
$\tileWhite$  &    30    &     30      &       0      &      30      &     30      &     30       &      1      \\
$\tileDark$   &    30    &     30      &      30      &      30      &     30      &      0       &      1      \\
$\tileLight$  &    30    &     30      &      30      &      30      &      0      &     30       &      1      \\
$\tileRing$   &     1    &      1      &       1      &       1      &      1      &      1       &     30      \\
\end{tabular}

\medskip

\begin{tabular}{c|rrrrrrr}
  \multicolumn{8}{c}{Vertical tiling rules} \\
              & $\tileH$ & $\tileHrev$ & $\tileBlack$ & $\tileWhite$ & $\tileDark$ & $\tileLight$ & $\tileRing$ \\
 \hline
$\tileH$      &    30    &     30      &       1      &      30      &      1      &      30      &      2      \\
$\tileHrev$   &    30    &     30      &      30      &       1      &     30      &       1      &      2      \\
$\tileBlack$  &     1    &     30      &      30      &      30      &     30      &       0      &      1      \\
$\tileWhite$  &    30    &      1      &      30      &      30      &      0      &      30      &      1      \\
$\tileDark$   &     1    &     30      &      30      &       0      &     30      &      30      &      1      \\
$\tileLight$  &    30    &      1      &       0      &      30      &     30      &      30      &      1      \\
$\tileRing$   &     2    &      2      &       1      &       1      &      1      &       1      &     30      \\
\end{tabular}
\caption{The tiling weights for layer 1 for WEIGHTED \tiling\ with reflection symmetry and periodic boundary conditions.  Since there is reflection symmetry, the horizontal and vertical tiling weight matrices are symmetric.}
\label{table:reflectionperiodic}
\end{centering}
\end{table}


The only $2 \times 2$ squares with total cost $0$ involve one each of \tileWhite, \tileBlack, \tileLight, and \tileDark.  Since they must alternate colors both vertically and horizontally, but $N$ is odd, it is not possible to fill the whole torus with such squares.  Indeed, if we wish to avoid forbidden pairings, any region containing only \tileWhite, \tileBlack, \tileLight, and \tileDark\ tiles cannot include any path which is topologically non-trivial on the torus.  Such a path would circle the torus either vertically, horizontally, or both, and would therefore produce inconsistent constraints on the tile types.  Any tiling must therefore contain enough adjacent pairs of tiles in the grid with nonzero cost to block the non-trivial cycles.  Forbidden pairings have a high cost, so the cheapest way to block the non-trivial cycles is with \tileRing, \tileH, and \tileHrev\ tiles.  To block all kinds of non-trivial cycles, we need at least $2N-1$ \tileRing, \tileH, or \tileHrev\ tiles in the tiling.

Each \tileRing\ included in the tiling has a net cost of $4$, since for any tile (other than another \tileRing), a \tileRing\ placed adjacent to it in any direction costs exactly $1$ more than any other allowed tile in the same position.  We can account each \tileH\ or \tileHrev\ tile at an overall cost of $2$, since that is the minimum cost to place tiles both above and below it.  Therefore, \tileH\ and \tileHrev\ tiles are cheaper than \tileRing\ tiles.  However, \tileH\ and \tileHrev\ tiles are only allowed to be horizontally adjacent to each other or a \tileRing\ tile.  Therefore, the \tileH\ and \tileHrev\ tiles must form horizontal lines.  One such line can act as one of the two topologically non-trivial cycles, but a second such line is not helpful.  Furthermore, \tileH\ and \tileHrev\ must alternate horizontally, which means we cannot go all the way around a horizontal cycle with just them.  Therefore, the horizontal line must contain at least one \tileRing; the cost will be minimal if there is exactly one.

Thus, to minimize the cost, we must have a configuration which contains a large region of \tileWhite, \tileBlack, \tileLight, and \tileDark\ tiles, arranged to avoid forbidden pairings, plus a horizontal line containing $N-1$ alternating \tileH\ and \tileHrev\ tiles, and $N$ additional \tileRing\ tiles, for a total cost of $2(N-1) + 4N = 6N - 2$.  Furthermore, there must be exactly one \tileRing\ tile in each row (since there are only $N$ of them, and each row must have at least one in order to get the parity right).  The \tileRing\ tiles cannot be adjacent vertically, and they must form a continuous path (in order to form a second topologically non-trivial cycle), so they are adjacent diagonally.  Because $N$ is odd, the only way that the \tileRing\ tiles can form a closed path is therefore for them to form a single diagonal line.  In short, to achieve a cost $6N-2$, we must have a configuration much like that of figure~\ref{fig:periodicweightedreflection}.

{\bf Layers $2$ and $3$:}

Layers $2$ and $3$ work almost the same way as for the open boundary condition case.  For the purpose of the Layer $2$ and Layer $3$ rules, a Layer $1$ \tileHrev\ tile is treated the same way as a Layer $1$ \tileBlack\ tile, and a Layer $1$ \tileH\ tile is treated the same way as a Layer $1$ \tileWhite\ tile.  As before, Layer $2$ will cycle between columns of \tileHH, \tileHHV, and \tileCrosshatch, and Layer $3$ will cycle between rows of \tileVVrev, \tileHVVrev, and \tileCrosshatchrev\ tiles.  Since $N$ is divisible by $3$, it is possible to do this consistently, and there are no other configurations consistent with an optimal Layer $1$ tiling.  However, unlike the open boundary conditions case, there is nothing to fix which columns have \tileHH\ and which rows have \tileVVrev, so there are three possible Layer $2$ tilings and three possible Layer $3$ tilings which are consistent with any optimal Layer $1$ tiling.

{\bf Layer $4$:}

Layer $4$ is used to set the boundaries of the square $(N-1) \times (N-1)$ region used in the main layer to implement the construction of section~\ref{sec:tiling}.  It uses the tiles \tileH, \tileV, \tileC, \tileN, \tileE, \tileS, \tileW, \tileNW, \tileNE, \tileSE, \tileSW, and \tileWhite.  We consult Layer $1$ to determine the locations of boundaries, and Layers $2$ and $3$ to determine the orientation left/right and up/down. The tiling weights are given in table~\ref{table:reflectionperiodic4}.  We also have the rules in table~\ref{table:reflectionperiodic4compat} which specify which pairs of layer 4 tiles and layer 1 tiles can be in the same location.

\begin{table}
\begin{centering}
\begin{tabular}{lc|rrrrrrrrrrrr}
     &               &                \multicolumn{12}{c}{Layer 4 tile} \\
     &               & $\tileN$ & $\tileS$ & $\tileE$ & $\tileW$ & $\tileNW$ & $\tileNE$ & $\tileSW$ & $\tileSE$ & $\tileWhite$ & $\tileH$ & $\tileV$
& $\tileC$ \\
 \hline
     & $\tileH$      &    N     &     N    &     N    &     N    &     N     &     N     &     N     &     N     &      N       &    Y    &    N
&    N     \\
     & $\tileHrev$   &    N     &     N    &     N    &     N    &     N     &     N     &     N     &     N     &      N       &    Y    &    N
&    N     \\
Layer 1 & $\tileBlack$ &  Y     &     Y    &     Y    &     Y    &     Y     &     Y     &     Y     &     Y     &      Y       &    N    &    Y
&    N     \\
tile & $\tileWhite$  &    Y     &     Y    &     Y    &     Y    &     Y     &     Y     &     Y     &     Y     &      Y       &    N    &    Y
&    N     \\
     & $\tileDark$   &    Y     &     Y    &     Y    &     Y    &     Y     &     Y     &     Y     &     Y     &      Y       &    N    &    Y
&    N     \\
     & $\tileLight$  &    Y     &     Y    &     Y    &     Y    &     Y     &     Y     &     Y     &     Y     &      Y       &    N    &    Y
&    N     \\
     & $\tileRing$   &    Y     &     Y    &     Y    &     Y    &     Y     &     Y     &     Y     &     Y     &      Y       &    N    &    Y
&    Y     \\
\end{tabular}
\caption{The compatibility rules between layer 1 tiles and layer 4 tiles for WEIGHTED \tiling\ with reflection symmetry and periodic boundary conditions.}
\label{table:reflectionperiodic4compat}
\end{centering}
\end{table}

\begin{table}
\begin{centering}
\begin{tabular}{lc|rrrrrrrrrrrr}
     &               &                \multicolumn{12}{c}{Tile on right} \\
     &               & $\tileN$ & $\tileS$ & $\tileE$ & $\tileW$ & $\tileNW$ & $\tileNE$ & $\tileSW$ & $\tileSE$ & $\tileWhite$ & $\tileH$ & $\tileV$
& $\tileC$ \\
 \hline
     & $\tileN$      &     0    &    30    &    30    &    30    &    30     &     0     &    30     &    30     &     30       &   30    &    30
&   30     \\
     & $\tileS$      &    30    &     0    &    30    &    30    &    30     &    30     &    30     &     0     &     30       &   30    &    30
&   30     \\
     & $\tileE$      &    30    &    30    &    30    &    30    &    30     &    30     &    30     &    30     &     30       &   30    &     0
&   30     \\
     & $\tileW$      &    30    &    30    &    30    &    30    &    30     &    30     &    30     &    30     &      0       &   30    &    30
&   30     \\
Tile & $\tileNW$     &     0    &    30    &    30    &    30    &    30     &     0     &    30     &    30     &     30       &   30    &    30
&   30     \\
on   & $\tileNE$     &    30    &    30    &    30    &    30    &    30     &    30     &    30     &    30     &     30       &   30    &     0
&   30     \\
left & $\tileSW$     &    30    &     0    &    30    &    30    &    30     &    30     &    30     &     0     &     30       &   30    &    30
&   30     \\
     & $\tileSE$     &    30    &    30    &    30    &    30    &    30     &    30     &    30     &    30     &     30       &   30    &     0
&   30     \\
     & $\tileWhite$  &    30    &    30    &     0    &    30    &    30     &    30     &    30     &    30     &      0       &   30    &    30
&   30     \\
     & $\tileH$      &    30    &    30    &    30    &    30    &    30     &    30     &    30     &    30     &     30       &    0    &    30
&    0     \\
     & $\tileV$      &    30    &    30    &    30    &     0    &     0     &    30     &     0     &    30     &     30       &   30    &    30
&   30     \\
     & $\tileC$      &    30    &    30    &    30    &    30    &    30     &    30     &    30     &    30     &     30       &    0    &    30
&   30     \\
\end{tabular}

\medskip

\begin{tabular}{lc|rrrrrrrrrrrr}
       &               &                \multicolumn{12}{c}{Tile on top} \\
      &               & $\tileN$ & $\tileS$ & $\tileE$ & $\tileW$ & $\tileNW$ & $\tileNE$ & $\tileSW$ & $\tileSE$ & $\tileWhite$ & $\tileH$ & $\tileV$
& $\tileC$ \\
 \hline
       & $\tileN$      &    30    &     30   &    30    &    30    &    30     &     30    &    30     &    30     &     30       &    0    &    30
&   30     \\
       & $\tileS$      &    30    &     30   &    30    &    30    &    30     &     30    &    30     &    30     &      0       &   30    &    30
&   30     \\
       & $\tileE$      &    30    &     30   &     0    &    30    &    30     &      0    &    30     &    30     &     30       &   30    &    30
&   30     \\
       & $\tileW$      &    30    &     30   &    30    &     0    &     0     &     30    &    30     &    30     &     30       &   30    &    30
&   30     \\
Tile   & $\tileNW$     &    30    &     30   &    30    &    30    &    30     &     30    &    30     &    30     &     30       &    0    &    30
&   30     \\
on     & $\tileNE$     &    30    &     30   &    30    &    30    &    30     &     30    &    30     &    30     &     30       &    0    &    30
&   30     \\
bottom & $\tileSW$     &    30    &     30   &    30    &     0    &     0     &     30    &    30     &    30     &     30       &   30    &    30
&   30     \\
       & $\tileSE$     &    30    &     30   &     0    &    30    &    30     &      0    &    30     &    30     &     30       &   30    &    30
&   30     \\
       & $\tileWhite$  &     0    &     30   &    30    &    30    &    30     &     30    &    30     &    30     &      0       &   30    &    30
&   30     \\
       & $\tileH$      &    30    &      0   &    30    &    30    &    30     &     30    &     0     &     0     &     30       &   30    &    30
&   30     \\
       & $\tileV$      &    30    &     30   &    30    &    30    &    30     &     30    &    30     &    30     &     30       &   30    &     0
&    0     \\
       & $\tileC$      &    30    &     30   &    30    &    30    &    30     &     30    &    30     &    30     &     30       &   30    &     0
&   30     \\
\end{tabular}
\caption{The tiling weights for layer 4 for WEIGHTED \tiling\ with reflection symmetry and periodic boundary conditions.  Even though the underlying rules have reflection symmetry, we have presented them in a way without reflection symmetry.  The distinction between ``left'' and ``right'' in a horizontally adjacent pair is provided by the corresponding pair of layer 2 tiles, and the distinction between ``top'' and ``bottom'' for a vertically adjacent pair is provided by layer 3.}
\label{table:reflectionperiodic4}
\end{centering}
\end{table}



Suppose we have an optimal tiling of Layer $1$.  Then there is a row which contains \tileH\ and \tileHrev, with one \tileRing.  The locations with \tileH\ or \tileHrev\ on Layer $1$ must have \tileH\ on Layer $4$, but the \tileRing\ location cannot, so the only remaining possibility for that location is a \tileC.  Since \tileC\ must have \tileH\ adjacent to it horizontally, and there can be no other \tileH\ tiles on Layer $4$, this is the only location on Layer $4$ that contains a \tileC\ tile.  The rest of the column containing the \tileC\ must thus be all \tileV\ tiles.  That is, there is a vertical line of \tileV\ tiles intersecting a horizontal line of \tileH\ tiles at a \tileC\ location.  There can be no other \tileH, \tileV, or \tileC\ elsewhere in Layer $4$.


We consider the remaining tiles to be the ``interior'' of the tiling, defining an $(N-1) \times (N-1)$ grid.  The tiles \tileN, \tileNW, and \tileNE\ will mark the upper boundary of the interior, \tileE, \tileNE, and \tileSE\ mark the right edge of the interior, \tileS, \tileSW, and \tileSE\ mark the bottom row, and \tileW, \tileNW, and \tileSW\ indicate the leftmost column of the interior.  The corners of the interior are indicated by \tileNE, \tileSE, \tileSW, and \tileNW.  The remaining interior tiles are \tileWhite.  The resulting tiling pattern is given in figure~\ref{fig:periodicweightedreflection4}.

\begin{figure}
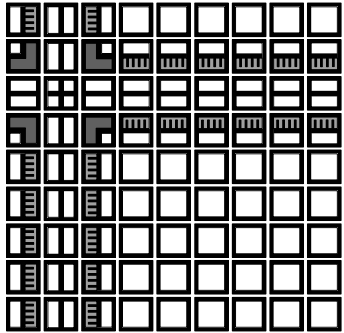

\begin{centering}
\begin{tabular}{c@{\extracolsep{0.1em}}c@{}c@{}c@{}c@{}c@{}c@{}c@{}c}
\tileE  & \tileV & \tileW  & \tileWhite & \tileWhite & \tileWhite & \tileWhite & \tileWhite & \tileWhite \\
\tileSE & \tileV & \tileSW & \tileS     & \tileS     & \tileS     & \tileS     & \tileS     & \tileS     \\
\tileH  & \tileC & \tileH  & \tileH     & \tileH     & \tileH     & \tileH     & \tileH     & \tileH     \\
\tileNE & \tileV & \tileNW & \tileN     & \tileN     & \tileN     & \tileN     & \tileN     & \tileN     \\
\tileE  & \tileV & \tileW  & \tileWhite & \tileWhite & \tileWhite & \tileWhite & \tileWhite & \tileWhite \\
\tileE  & \tileV & \tileW  & \tileWhite & \tileWhite & \tileWhite & \tileWhite & \tileWhite & \tileWhite \\
\tileE  & \tileV & \tileW  & \tileWhite & \tileWhite & \tileWhite & \tileWhite & \tileWhite & \tileWhite \\
\tileE  & \tileV & \tileW  & \tileWhite & \tileWhite & \tileWhite & \tileWhite & \tileWhite & \tileWhite \\
\tileE  & \tileV & \tileW  & \tileWhite & \tileWhite & \tileWhite & \tileWhite & \tileWhite & \tileWhite
\end{tabular}
\caption{The arrangement of layer $4$ when the Layer $1$ tiling is given by figure~\ref{fig:periodicweightedreflection}.}
\label{fig:periodicweightedreflection4}
\end{centering}
\end{figure}

{\bf Main Layers:}

In the main layers, we again define directions, breaking the reflection symmetry by looking at Layers $2$ and $3$.  We implement the protocol of section~\ref{sec:tiling} on an $(N-1) \times (N-1)$ grid defined by the interior locations determined by Layer $4$.  If the Layer $4$ tile is a \tileH, \tileV, or \tileC, the main layer tile must be an extra tile not used in the main set.  The extra tile type can only be in a location that has one of those three tile types on Layer $4$, and it can be adjacent to any other main layer tile.  Locations on Layer $4$ that have a \tileNE, \tileSE, \tileSW, or \tileNW\ tile must have a corner tile in the main layer.  The other Layer $4$ tiles can correspond to any main layer tile.  Thus, when Layer $1$ has an optimal tiling, and Layers $2$ -- $4$ have no forbidden pairings, the main layer only has a tiling without forbidden pairings if the universal Turing machine $M$ accepts on input $f_{BC}(N-1)$.  That is, there is an overall tiling of cost $6N-2$ iff the problem is a ``yes'' instance; otherwise, the cost is higher.

\end{proof}

\subsection{Weighted Tiling with Rotation Symmetry}

When we have full rotation symmetry, even the WEIGHTED \tiling\ problem is easy:

\begin{theorem}
For WEIGHTED \tiling\ with rotation symmetry, we have the following results:
\begin{itemize}

\item With open boundary conditions: Either there exists computable $N_0 \in \integer^+$ such that tiling is possible for $N
    \geq N_0$ and impossible for $N < N_0$ or there exists computable $N_0 \in \integer^+$ such that tiling is possible for $N < N_0$ and
    impossible for $N \geq N_0$.  ($N_0$ depends on the weights and maximum allowed cost $p(N)$.)

\item With periodic boundary conditions: One of the following three cases holds: There exists computable $N_0 \in \integer^+$
    such that tiling is possible for all $N \geq N_0$, there exists computable $N_0 \in \integer^+$ such that tiling is impossible for all $N \geq
    N_0$, or tiling is possible for all even $N$ and there exists computable $N_o \in \integer^+ \cup \{\infty\}$ such that for odd $N \geq N_o$,
    tiling is possible.

\item With the four-corners boundary condition: There exist computable $N_e, N_o \in \integer^+ \cup
\{\infty\}$ such that either there exists a tiling for any even $N \geq N_e$ or for no even $N \geq N_e$, and either
there exists a tiling for any odd $N \geq N_o$ or for no odd $N \geq N_o$.
\end{itemize}
\end{theorem}

\begin{proof}

{\bf Rotation symmetry and open boundary conditions:} This case is essentially
trivial. We find the minimal weight $w$ between any pair of tiles, and simply tile
the square in a checkerboard pattern with those two tiles. That is certainly
the minimum cost achievable.  The resulting cost is $2N(N-1)w$.

\medskip

{\bf Rotation symmetry and periodic boundary conditions:} Find the lowest
weight of a cycle of length $N$ for the one-dimensional WEIGHTED \tiling\
problem. Call that $w$. The minimum achievable total cost for $2$-D WEIGHTED
\tiling\ with rotation symmetry and periodic boundary conditions is then $2Nw$
($N$ rows and $N$ columns each of cost $w$). This can actually always be done
using the same tiling as the unweighted rotation symmetry case
(figure~\ref{fig:rotation}).  Thus, the problem reduces to the one-dimensional
case, which we have argued is in $\Pclass$. However, with the additional reflection
symmetry, the argument can be simplified.
If all the weights are positive, then $w \geq N$
and tiling is only possible for small $N$ since we are assuming $p(N) = o(N)$.  If there exists a pair of tiles with negative
weight, then for sufficiently large $N$, the minimal cost path will use many
of that negative pair, and tiling is always possible for sufficiently large $N$.
If the smallest weight is $0$, then any sufficiently long one-dimensional cycle
with constant weight $w$ must use a $0$-cost pairing, and it is possible to
duplicate the pair of tiles used in that $0$-cost pairing.  Thus, the minimum
weight cycle of length $N+2$ is also at most $w$.  Therefore, we need only
determine the minimum-weight cycle of odd length containing a zero-cost pairing.  (When $N$ is
even and there is a zero-cost pairing, a zero-cost cycle is always possible
by alternating between the two tiles involved in the pairing.)

\medskip

{\bf Rotation symmetry and 4-corners boundary condition:} This case is more difficult, but still computationally easy (in $\Pclass$). For large $N$,
the idea is that we will again tile most of the square using a pair of tiles $t_i$, $t_j$ with minimal cost to be adjacent.  That determines the
asymptotic scaling of the total cost for large $N$. However, the details are more complicated because of the effect of the corners.

First, note that if the minimum weight for every pair of tiles is positive,
then it is certainly not possible to tile an $N \times N$ grid when $2N (N-1) > p(N)$.
Conversely, if there is a pair of tiles $(t_i, t_j)$ such that the cost for
having them adjacent is negative, then we can always tile an $N \times N$
grid for sufficiently large $N$ by taking a checkerboard of $t_i$ and $t_j$,
with only the corners different (if the designated corner tile $t_1$ is not
already $t_i$ or $t_j$).  Suppose that $w (t_1, t_i), w (t_1, t_j) \leq w$.
Then ``sufficiently large'' $N$ means $N$ such that $2N(N-1) - 8 \geq 8w-p(N)$.

The only potentially difficult case is thus when there are pairs of tiles $(t_i, t_j)$ with $w(t_i, t_j) = 0$,
but no pairs with a negative cost.  We can assume $p(N) = c$,
a constant, as we can always achieve a constant cost $8w$ using the strategy above ($0$ cost pairs
everywhere but the corners).  We will show that if there exists a tiling of constant cost, the vicinity
of the four corners must satisfy certain conditions.  Then by running through the possible configurations
near the corners, we can find a minimum-weight tiling.

Consider the graph whose nodes are tile types $t_i$ and edges
correspond to pairings of tiles with $0$ cost.  Let $Z_\alpha$ be a set of tiles comprising a connected
component of this graph with more than one tile, and let $Z = \cup_\alpha Z_\alpha$.  We will
show that any tiling of total cost at most $c$ must consist mostly of tiles from one particular
$Z_\alpha$, with only a constant number of different tiles, and that indeed it is sufficient to
place all the non-$Z_\alpha$ tiles near the corners of the grid.  Then by considering all
possible low-cost combinations of tiles near the corners, we will be able to determine the
minimum achievable cost for the grid.

We will define the {\em $k$-square for the upper left corner} to be the $k \times k$ square of tiles in the upper
left corner of the grid. The {\em upper-left $k$-border} consists of the rightmost column and bottom
row of the $k$-square for the upper left corner.
We make similar definitions for the other corners: the upper-right $k$-border is the
leftmost column and bottom row of the $k$-square for the upper right corner, the lower left $k$-border is the
rightmost column and top row of the $k$-square for the lower left corner
and the lower right $k$-border is the
leftmost column and top row of the $k$-square for the lower right corner.
An example is given in Figure \ref{fig:squares}.
We say that a square is {\em valid} if its border contains only zero-cost  adjacent pairs and
the corner of the square that is also the corner of the larger grid has the correct corner tile.
Since a $k$-border must be connected that means that the edges within the border are all contained
within a particular connected component $Z_{\alpha}$.

\begin{figure}
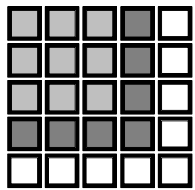

\begin{centering}
\begin{tabular}[t]{c@{\extracolsep{0.1em}}c@{}c@{}c@{}c}
\tileLight & \tileLight & \tileLight & \tileDark & \tileWhite \\
\tileLight & \tileLight & \tileLight & \tileDark & \tileWhite \\

\tileLight & \tileLight & \tileLight & \tileDark & \tileWhite \\

\tileDark & \tileDark & \tileDark & \tileDark & \tileWhite \\

\tileWhite & \tileWhite & \tileWhite & \tileWhite & \tileWhite \\

\end{tabular}
\caption{The shaded tiles are an upper-left $4$-square. The tiles in the darker shade
are the upper-left $4$-border.}
\label{fig:squares}
\end{centering}
\end{figure}

We will determine if there is an $N \times N$ tiling of the grid by considering all valid tilings
for the squares in each corner of the grid where we allow the sizes of the squares to go up
to $2c$.
For each such combination, we require that the total cost of the tiling be some
$c' \le c$ and that the tiles in each border come from the same connected component $Z_{\alpha}$.
We then
must determine whether it is possible to fill the remainder of the grid with
zero-cost adjacent pairs, possibly interrupted by a few additional adjacent pairs with
total cost at most $c-c'$ (although it will turn out that this is never necessary).
Suppose we have such a tiling, and consider tile $t_i \in Z_\alpha$ on the border of
one square, and tile $t_j \in Z_\alpha$ on the border of a different square.  Then, as shown below,
there must exist a path from $t_i$ to $t_j$ that only includes zero-cost adjacent pairs of tiles.  Note that
the parity of the length of the path does not depend on which path we choose.

More generally, given two $k$-squares $I_1$ and $I_2$ for different corners of the grid, we are
interested in whether it is \emph{possible} to create a path of zero-cost adjacent pairs to get
from $t_i \in I_1$ to $t_j \in I_2$, without necessarily being concerned about whether we
can tile the whole grid:

\begin{definition}
Suppose we have two $k$-squares $I_1$ and $I_2$ that are located
at different corners of the grid, with borders from $Z_\alpha$.  We say they are
\emph{compatible for $N$ and $Z_\alpha$} if, given tile $t_i$ on the
border of $I_1$ and tile $t_j$ on the border of $I_2$, there is a path
of zero-cost pairings from $t_i$ to $t_j$ of the correct
length for $I_1$ and $I_2$ to be placed together on an $N \times N$ grid.
We say a set of four squares $I_1$, $I_2$, $I_3$, and $I_4$ which fit
into the
four corners of the grid is \emph{compatible for $N$ and $Z_\alpha$} if
they are compatible in pairs.
\end{definition}


Note that it does not matter which tile we pick from the border to define compatibility.
This is because if we use $t_i'$ on the border of $I_1$ instead of $t_i$,
we can take a path from $t_i'$ to $t_i$ using only the zero-cost pairings
in the border of $I_1$, and concatenate it with the path from $t_i$ to $t_j$.
Similarly, if $I_1$, $I_2$, and $I_3$ are valid squares located at three
different corners, then if $I_1$ is compatible with $I_2$ and $I_2$ is
compatible with $I_3$, it follows that $I_1$ is compatible with $I_3$.
Also, if $I_1$ and $I_2$ are compatible for $N$, they are also compatible
for $N+2$: Because of the reflection symmetry, we can repeat a pair of adjacent tiles
in a path to lengthen it by $2$.

For large $N$, if there is a tiling of the $N \times N$ grid
with total cost at most $c$, then there must exist, for some $\alpha$,
valid squares in the four corners
$I_1$, $I_2$, $I_3$, and $I_4$ which are compatible
for $N$ and $Z_\alpha$ such that the total cost of $I_1$, $I_2$, $I_3$,
and $I_4$ is at most $c$ and the sizes of the squares are between $c$ and $2c$.
To see why this is true, note that there can be only $c$ non-zero-cost adjacent pairs
in the entire tiling. Therefore, as $k$ ranges from $c$ to $2c$, it must be the
case that for each corner at least one $k$-border contains only zero-cost adjacent
pairs. Furthermore, for each such pair of borders, there are at least $c+1$ disjoint paths
from a tile in one border to a tile in the other border. This follows from the fact
that the borders themselves have more than $c$ tiles.
At least one of these paths
must contain only zero-weight pairings.
Thus we know that for sufficiently large grids,
if there is a tiling of total weight at most $c$,
then there must be four valid squares for the four corners which
are compatible for $N$ and $Z_\alpha$.
The converse is also true:
\begin{lemma}
Let $I_1$, $I_2$, $I_3$, and $I_4$ be four squares which
are compatible for $N$ and $Z_\alpha$, $N$ sufficiently large.
Then there is a tiling of the $N \times N$ grid using
$I_1$, $I_2$, $I_3$, and $I_4$ in the four corners, with only tiles
from $Z_{\alpha}$ used in the rest of the grid.
\end{lemma}

\begin{proof}[ of lemma]

Assume without loss of generality that $I_1$ is in the upper left corner,
$I_2$ is in the upper right corner, $I_3$ is in the lower left corner, and
$I_4$ is in the lower right corner of the grid.

We will extend each  square to
create an isosceles right triangle whose outer diagonal border is composed
of just one kind of tile.  This can be done by progressively duplicating
the outer sides of the square, each time shifting by one space towards
the edge of the grid, as in figure~\ref{fig:shelftriangular}.  Finally,
we can add additional diagonal layers to make sure the outer diagonal border
is any particular tile from $Z_\alpha$, and to increase the size of the
triangle as much as desired.

\begin{figure}
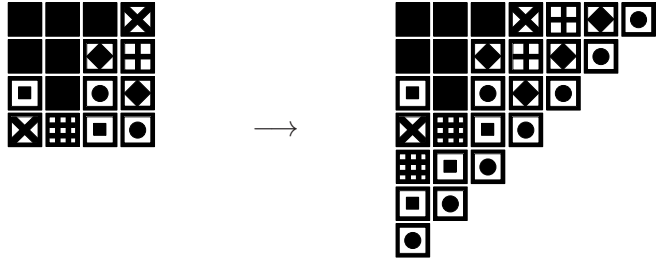

\begin{centering}
\begin{tabular}[t]{c@{\extracolsep{0.1em}}c@{}c@{}c}
\tileBlack & \tileBlack & \tileBlack & \tileX \\
\tileBlack & \tileBlack & \tileD     & \tileC \\
\tileSquare & \tileBlack & \tileCirc  & \tileD \\
\tileX & \tileCrosshatch & \tileSquare & \tileCirc \\
\end{tabular}
\qquad
\begin{tabular}[t]{c}
\\
\\
\\
$\longrightarrow$ \\
\end{tabular}
\qquad
\begin{tabular}[t]{c@{\extracolsep{0.1em}}c@{}c@{}c@{}c@{}c@{}c}
\tileBlack & \tileBlack & \tileBlack & \tileX & \tileC    & \tileD    & \tileCirc \\
\tileBlack & \tileBlack & \tileD     & \tileC & \tileD    & \tileCirc & \\
\tileSquare & \tileBlack & \tileCirc  & \tileD & \tileCirc &           & \\
\tileX & \tileCrosshatch & \tileSquare & \tileCirc &  &           & \\
\tileCrosshatch & \tileSquare & \tileCirc &        &           &           & \\
\tileSquare & \tileCirc  &            &        &           &           & \\

\tileCirc  &            &        &           &           &            &\\
\end{tabular}
\caption{Extending a square to make a corner triangle.  The black tiles have nonzero cost to be adjacent to any other tile.  All other pairings that
occur in the figure are zero cost.}
\label{fig:shelftriangular}
\end{centering}
\end{figure}

We can define compatibility of the newly created corner triangles
 in the same way as for the squares.
When squares $I$ and $J$ for two different corners
are compatible for $N$ and $Z_\alpha$,
then their corresponding triangles
remain compatible when extended as described above, at least when $N$ is large. In particular, there is a path of zero-cost
tile pairs  that will let us tile the top row of the grid between
$I_1$ and $I_2$.
Then we can copy the tiling of the interval between
triangles diagonally down and to the left, much as in figure~\ref{fig:rotation}
describing the unweighted case.  This fills in the upper left corner of the
grid, up to the diagonal line defined by the leftmost tile of the triangle for $I_2$
and the topmost tile of the triangle for $I_3$.
We can do the same thing for the lower right corner of the grid,
tiling the right edge between $I_2$ and $I_4$, and copying this
tiling diagonally down and to the left.  This leaves only a diagonal
strip between $I_2$ and $I_3$, and we can tile that using a checkerboard
pattern consisting of the outermost tile $t_i$ for the triangles and any
other tile from $Z_\alpha$ which has a zero-cost pairing with $t_i$.
See figure~\ref{fig:fourcornersrotation} for an example of a complete grid tiled in this way.

\begin{figure}
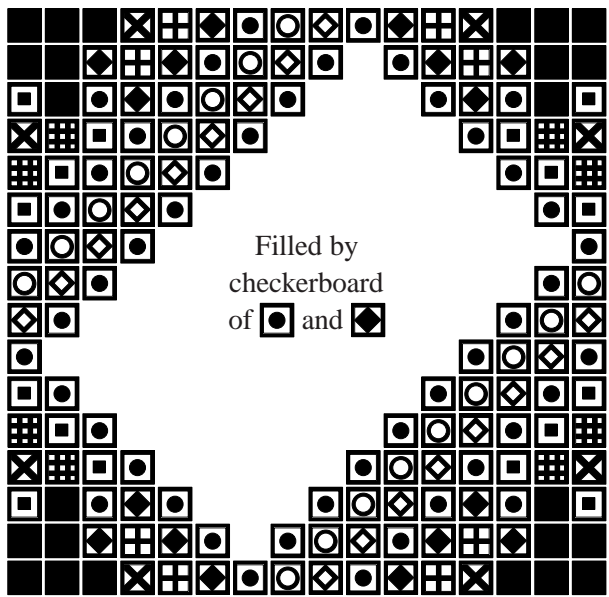

\begin{centering}
\begin{tabular}{c@{\extracolsep{0.1em}}c@{}c@{}c@{}c@{}c@{}c@{}c@{}c@{}c@{}c@{}c@{}c@{}c@{}c@{}c}
\tileBlack & \tileBlack & \tileBlack & \tileX & \tileC    & \tileD    & \tileCirc & \tileRing & \tileDHol & \tileCirc & \tileD & \tileC & \tileX &
\tileBlack & \tileBlack & \tileBlack \\
\tileBlack & \tileBlack & \tileD     & \tileC & \tileD    & \tileCirc & \tileRing & \tileDHol & \tileCirc &        & \tileCirc & \tileD & \tileC &
\tileD     & \tileBlack & \tileBlack \\
\tileSquare & \tileBlack & \tileCirc & \tileD & \tileCirc & \tileRing & \tileDHol & \tileCirc &        &           &        & \tileCirc & \tileD &
\tileCirc & \tileBlack & \tileSquare \\
\tileX & \tileCrosshatch & \tileSquare & \tileCirc & \tileRing & \tileDHol & \tileCirc &  &           &        &           &        & \tileCirc &
\tileSquare & \tileCrosshatch & \tileX \\
\tileCrosshatch & \tileSquare & \tileCirc & \tileRing & \tileDHol & \tileCirc &        &           &        &           &        &           &        &
\tileCirc & \tileSquare & \tileCrosshatch \\
\tileSquare & \tileCirc  & \tileRing & \tileDHol & \tileCirc &\multicolumn{6}{c}{ }                                         &           &
&        & \tileCirc & \tileSquare \\
\tileCirc  & \tileRing   & \tileDHol & \tileCirc &           &\multicolumn{6}{c}{Filled by}                                 &           &
&        &         & \tileCirc \\
\tileRing  & \tileDHol   & \tileCirc &         &             &\multicolumn{6}{c}{checkerboard}                              &          &
&        & \tileCirc & \tileRing \\
\tileDHol & \tileCirc   &        &           &               &\multicolumn{6}{c}{of $\tileCirc$ and $\tileD$}               &              &
& \tileCirc & \tileRing & \tileDHol \\
\tileCirc &             &        &           &           &         &           &        &           &        &          &               & \tileCirc
& \tileRing & \tileDHol & \tileCirc \\
\tileSquare & \tileCirc &        &           &        &           &        &           &            &        &          & \tileCirc & \tileRing
& \tileDHol & \tileCirc & \tileSquare \\
\tileCrosshatch & \tileSquare & \tileCirc &  &           &        &           &        &            &        & \tileCirc & \tileRing & \tileDHol &
\tileCirc & \tileSquare & \tileCrosshatch \\
\tileX & \tileCrosshatch & \tileSquare & \tileCirc &     &        &           &        &         & \tileCirc & \tileRing & \tileDHol & \tileCirc &
\tileSquare & \tileCrosshatch & \tileX \\
\tileSquare & \tileBlack & \tileCirc  & \tileD & \tileCirc &        &           &        & \tileCirc & \tileRing & \tileDHol & \tileCirc & \tileD &
\tileCirc & \tileBlack & \tileSquare \\
\tileBlack & \tileBlack & \tileD     & \tileC & \tileD    & \tileCirc &        & \tileCirc & \tileRing & \tileDHol & \tileCirc & \tileD & \tileC &
\tileD & \tileBlack & \tileBlack \\
\tileBlack & \tileBlack & \tileBlack & \tileX & \tileC    & \tileD    & \tileCirc & \tileRing & \tileDHol & \tileCirc & \tileD & \tileC & \tileX &
\tileBlack & \tileBlack & \tileBlack \\
\end{tabular}
\caption{A $16 \times 16$ grid tiled by extending triangles on the four corners.}
\label{fig:fourcornersrotation}
\end{centering}
\end{figure}
\end{proof}

Thus, we can use the following algorithm to determine what the lowest
achievable cost is for large even and odd grid sizes, when the lowest-cost
pairings are cost $0$: List the low-cost valid squares for each corner
up to size $2c$ and determine which
sets of four are compatible for large even $N$ and for large odd $N$.
The lowest total cost for each case tells us the minimum achievable cost.

\end{proof}

\section{The Quantum Case}
\label{sec:quantum}

\subsection{Preliminaries}

As in the 2-dimensional classical tiling problem, we make use of
a binary counting Turing machine $M_{BC}$. Because we are working
with quantum systems, we will require that $M_{BC}$ be reversible.
Bernstein and Vazirani \cite{bernsteinVazirani} have
shown that any deterministic Turing machine can
be made reversible, meaning that given a configuration of the Turing machine,
it has a unique predecessor in the computation.
There may be some additional overhead but it is not
significant.
We can still assume that  there is a
function $f: \mathbb{Z} \rightarrow \{0,1\}^*$ such that for some
constant $N_0$ and
every $N \ge N_0$, if $M_{BC}$ runs for $N$ steps, then the string $f_{BC}(N)$
will be written on the tape with the rest of the tape blank.
Moreover there are constants $c_1$ and $c_2$ such that
if $n$ is the length of the string $f_{BC}(N)$ and
$N \ge N_0$, then
$2^{c_1 n} \le N \le 2^{c_2 n}$.
We will also assume that for any binary string $x$, we can compute $N$
such that $f_{BC}(N) = x$ in time that is polynomial in the length of $x$.

We can reduce any language in $\QMAEXP$ to a language $L$ that
is accepted by a verifier who uses a witness of size $2^{c_1 n}$
and whose computation lasts for $2^{c_1 n}$ steps, where $n$ is the
length of the input.
This is the same reduction used in the classical case, in which the
input is padded to length $|x|^k/c_1$.
We can use standard boosting techniques to assume that the
probability of acceptance is at least $1 - \epsilon$ for a ``yes'' instance or at most
$\epsilon$ for a ``no'' instance, with $\epsilon = 1/{\mathrm{poly}}(N)$ \cite{Kitaev:book}.
Suppose we are given
an arbitrary verifier quantum Turing machine $V$ which takes as input
a classical/quantum pair $(x,\ket{\psi})$ such that $\ket{\psi}$
has $2^{c_1 n}$ qubits and halts in $2^{c_1 n}$ steps.
Based on $V$, we will produce a Hamiltonian term $H$ which acts on a pair of
finite-dimensional particles. We will also produce two polynomials $p$ and $q$.
The reduction for Theorem~\ref{th:qmaexp} will then take input string $x$
and output an integer $N$ such that $f_{BC}(N-3)=x$.
The Hamiltonian will have the property that for any $x$, if there exists
a $\ket{\psi}$ that causes $V$ to accept with probability
at least $1 - \epsilon$, then when $H$ is applied to every neighboring pair in
a chain of length $N$, the resulting system has a unique ground state whose energy
is at most $p(N)$. If for every $\ket{\psi}$, $V$ accepts with
probability at most $\epsilon$, then the ground state energy of the system is at least
$p(N) + 1/q(N)$.

Quantum Turing machines were first introduced in \cite{deutsch} and further
developed in \cite{bernsteinVazirani}. The latter paper showed that we can
make a number of simplifying assumptions about the form of a quantum Turing machine
and not restrict its power in a complexity-theoretic sense.
In particular, we can assume without loss of generality that the Turing machine $V$ has a one-way
infinite tape and that the head starts in designated start state $q_0$ at the left-most
end of the tape. We will also assume that on input $x$, after $2^{c_1 |x|}$ steps, the
Turing machine is in an accepting or rejecting
state and the head is again at the left-most end
of the tape. We will also assume
that the witness will be stored in a parallel track with the left-most
qubit in the left-most position of the tape.

We now describe the set of states for the particles.
A standard basis state for the whole system will be denoted by
the state for each particle. States \leftend\ and \rightend\ are
special bracket states that occur at the ends of the chain.

\begin{definition}
A standard basis state is {\em bracketed} if the left-most particle
is in state \leftend, the right-most particle is in state \rightend,
and no other particle in the chain is in state \leftend\ or \rightend.
$\calS_{br}$ is the space spanned by all bracketed states.
\end{definition}

We will restrict our attention for now to $\calS_{br}$ and add a term
later that gives an energy penalty to any state outside $\calS_{br}$.
The rest of the states will be divided into six tracks, so the state of a particle
is an ordered $6$-tuple with each entry specifying the state for a particular track.
The set of allowable states will not necessarily be the full cross product of the
states for each track.

Two of the tracks will implement a clock, with one track working as sort of a second hand
and another track as a minute hand. The other four tracks will be used to implement
two Turing machines which share a work tape.
Track 3 holds the work tape. Track 4 holds the state and head location for the first
Turing Machine (which is $M_{BC}$) and Track 5 holds the state and head location for
the second Turing Machine (which is $V$).
The sixth track will hold the quantum witness for $V$.
Since there is limited interaction between the tracks,
it will be simpler to describe the Hamiltonian as it acts on each track separately
and then describe how they interact. The figure below gives a picture of the
start state for the system. Each column represents the state of an individual
particle.

\begin{center}
\begin{small}
\begin{tabular}{|@{}c@{}|@{}c@{}|@{}c@{}|}
\hline
$\leftend$ &
\begin{tabular}{@{}c@{}|@{}c@{}|clc|@{}c@{}|@{}c@{}}
\arrRzero & \blankR & $\cdots$ & Track 1: Clock second hand & $\cdots$ & \blankR &  \blankR \\
\hline
\zeroB & \zero & $\cdots$ & Track 2: Clock minute hand & $\cdots$ & \zero &  \zero \\
\hline
\# & \# & $\cdots$  & Track 3: Turing machine work tape & $\cdots$ & \# &  \# \\
\hline
$q_0$ & \blankR  & $\cdots$ & Track 4: Tape head and state for TM $M_{BC}$ & $\cdots$ & \blankR   & \blankR   \\
\hline
$q_0$ & \blankR   & $\cdots$ & Track 5: Tape head and state for TM $V$ & $\cdots$ & \blankR   &  \blankR   \\
\hline
0/1 & 0/1& $\cdots$ & Track 6: Quantum witness for $V$ & $\cdots$ & 0/1&  0/1\\
\end{tabular}
& $\rightend$ \\
\hline
\end{tabular}
\end{small}
\end{center}

As is typical in hardness results for finding ground state energies,
the Hamiltonian applied to each pair
will consist of a sum of terms of which there are two types.
Type I terms will have the form
$\ket{ab}\bra{ab}$ where $a$ and $b$ are possible states.
This has the effect of adding an energy penalty to any
state which has a particle
in state $a$ to the immediate left of a particle in state $b$.
We will say a configuration is {\em legal} if it does not violate any
Type I constraints.
Type II terms will have the form:
$\frac{1}{2}
(\ketbra{ab}{ab} + \ketbra{cd}{cd} -
\ketbra{ab}{cd} - \ketbra{cd}{ab})$.
These terms enforce that for any eigenstate with zero
energy, if there is a configuration $A$ with two neighboring
particles in states $a$ and $b$, there must be a configuration $B$ with
equal amplitude that is the same as $A$ except that $a$ and $b$ are
replaced by $c$ and $d$.
Even though a Type II term is symmetric, we associate
a direction with it by denoting it with $ab \rightarrow cd$.
Type II terms are also referred to as
{\em transition rules}. We will say that configuration
$A$  transitions into configuration $B$ by rule $ab \rightarrow cd$
if $B$ can be
obtained from $A$ by replacing an occurrence of $ab$ with an occurrence of
$cd$. We say that the transition rule
applies to $A$ in the forward direction
and applies to $B$ in the backwards direction.
We will choose the terms so that for any legal configuration, at most one
transition rule applies to it in the forward direction and at most one
rule applies in the backwards direction. Thus, a state satisfying
all Type I and Type II constraints must consist of
an equal superposition of
legal configurations such that there is exactly one transition rule
that carries each configuration to the next.
The illegal pairs are chosen so that any state which satisfies the Type I and Type II
constraints corresponds to a process we would like to simulate or encode in the
ground state. In our case, the process is the execution of two Turing Machines
each for $N-3$ steps, where $N$ is the length of the chain.

We will make use of the following simple lemma throughout the construction
in limiting the set of  standard basis states in the support of the
ground state.

\begin{lemma}
For any regular expression over the set of particle states in which each
state appears at most once, we can use illegal pairs to ensure that
any legal standard basis state for the system is a substring of
a string in the regular set.
\end{lemma}

\begin{proof}
The alphabet for the regular expression is the set of particle states.
Since each character appears once in the regular expression, the set of
characters $b$ which can follow a particular character $a$ is well defined
and does not depend on where the character appears in the string.
Therefore, we can add an illegal pair $ab$ if $b$ is not one of the characters
which can follow $a$. Any substring of the regular expression has no illegal pairs.
\end{proof}

\subsection{Outline of the Construction}

We give now a brief outline of the construction and provide the full
details in the subsections that follow.
Illegal pairs are used to enforce that the state of Track 1 is always of the
form $\leftend \blankL^* ( \arrR + \arrL) \blankR^* \rightend$.  (The $+$ denotes the
regular expression {\sc OR} and not a quantum superposition.)
There is one arrow symbol on Track 1 that shuttles back and forth between the left end and the right end
and operates as a second hand for our clock. We call one round trip of the arrow
on Track 1 an {\em iteration}.
Every iteration has $2(N-2)$
distinct states and $2(N-2)$ transitions. Each iteration causes one change
in the configuration on Track 2 which acts then as a minute hand for the clock.
The Track 2  states are partitioned into two phases. The first phase is called the
{\em Counting Phase} and consists of all $N-2$ of the states
of the form $\leftend \one^* \zeroB \zero^*  \rightend$. The second phase is the
{\em Computation Phase} and consists of
all $N-2$ of the states of the form
$\leftend \one^* \oneB \two^*\rightend$ states.
The $\leftend \one^* \zeroB \zero^*  \rightend$ states are ordered according to the number of
particles in state $\one$ and
the $\leftend \one^* \oneB \two^*  \rightend$ states are ordered according to the number of
particles in state $\two$.
The state immediately after $\leftend \one^* \zeroB  \rightend$  in the ordering
is $\leftend \one^* \oneB  \rightend$.
The target ground state for the clock is the uniform superposition of all the clock
states, entangled appropriately with states of the other $4$ tracks.
We need to have illegal pairs that cause all other states to have an energy
cost. As is the case in other such proofs, it is not possible to disallow all states
directly with illegal pairs. Instead, we need to show that some states are
unfavorable because they evolve via forward or backwards transitions to
high energy states.

Each of the arrow states for Track 1 will come in three varieties: $\arrRzero$ and $\arrLzero$ will be
used during the initial minute of the clock when it is in state $\leftend \zeroB \zero^* \rightend$
and will be used to check initial conditions on the other tracks.
$\arrRone$ and $\arrLone$ will be used during the counting phase and
$\arrRtwo$ and $\arrLtwo$ will be used during the computation phase.
$\arrRone$ and $\arrRtwo$ will be used to trigger different actions on the other tapes.
Every time the $\arrRone$ sweeps from the left end of the chain
to the right end of the chain, it causes $M_{BC}$ to execute one
more step. Thus, $M_{BC}$ is run for exactly $N-2$ steps.
The $\arrRtwo$ symbol is what causes the Turing machine
$V$ to execute a step.
We then add a term that penalizes any state which is in the
final clock state and does not have an accepting Turing machine state.
Thus, only accepting computations will have low energy.

Finally we use an additional term to enforce the boundary conditions.
This is achieved by weighting the Hamiltonian terms for the illegal pairs
and transition rules by a factor of three.
Then an additional term is applied to each particle, which gives a benefit
to any particle that is in the \leftend\ or \rightend\ state.
Only the left-most and right-most particles can obtain this
energy benefit without incurring the higher cost of having an
endpoint state in the middle of the chain.

\begin{table}
\begin{centering}
\begin{tabular}{llccc}
Phase                & Function             & Track $1$ state  & Track $2$ state  & Number of steps  \\
\hline
Initialization Phase & Check start state & $\leftend \blankL^* ( \arrRzero + \arrLzero) \blankR^* \rightend$ &
    $\leftend \zeroB \zero^*  \rightend$        & $2(N-2)$ \\
Counting Phase       & Run $M_{BC}$         & $\leftend \blankL^* ( \arrRone + \arrLone) \blankR^* \rightend$   &
    $\leftend \one^* \zeroB \zero^*  \rightend$ & $(N-2)(2N-5)$ \\
Computation Phase    & Run $V$              & $\leftend \blankL^* ( \arrRtwo + \arrLtwo) \blankR^* \rightend$   &
    $\leftend \one^* \oneB \two^*\rightend$     & $(N-2)(2N-5)$
\end{tabular}
\end{centering}
\caption{The different clock phases in the quantum construction on a line. }
\label{table:qphases}
\end{table}

\subsection{The Clock Tracks}
\label{sec:clock}

There will be four types of states in Track 1:
\arrR, \arrL, \blankL, and \blankR.  (However, the \arrR\ and \arrL\ states each come in multiple varieties --- see below for the details.)
We will have illegal pairs that will enforce that any standard basis
state with no illegal pairs must be a substring of a string from the regular  expression:
$$\leftend \blankL^* ( \arrR + \arrL) \blankR^* \rightend.$$
Furthermore, any bracketed standard basis state must be exactly a string corresponding to this
regular expression.
The transition rules then move the arrows in the direction in which they are pointing:

\begin{enumerate}
\item $\arrR \blankR \rightarrow \blankL \arrR$, $\blankL \arrL \rightarrow \arrL \blankR$: in the forward
direction, arrows move in the direction they are pointing. In the reverse direction, they move in the
opposite direction.
\item $\arrR \rightend \rightarrow \arrL \rightend$, $\leftend \arrL \rightarrow \leftend \arrR$: arrows change
direction when they hit an endpoint.
\end{enumerate}

Note that the $\arrR$ symbol always interacts with the particle on its right in the forward direction and
the particle on its left in the reverse direction.
Similarly, the $\arrL$ symbol interacts with the particle on its left in the forward direction and the
particle on the right in the reverse direction.
Therefore, we know that every legal configuration for Track 1 has exactly one transition rule that applies
in the forward direction and one that applies in the reverse direction.
The arrow symbol on Track 1 shuttles back and forth between the left end and the right end as
shown below and operates as a second hand for our clock. We call one round trip of the arrow
on Track 1 an {\em iteration}.
Every iteration has $2(N-2)$
distinct states and $2(N-2)$ transitions. Each iteration causes one change
in the configuration on Track 2 which acts then as a minute hand for the clock.

\begin{center}
\begin{tabular}{cc}
\leftend \arrR \blankR \blankR \blankR \rightend &  \leftend   \blankL \blankL \arrL \blankR \rightend \\
\leftend  \blankL \arrR \blankR \blankR \rightend  & \leftend   \blankL \arrL \blankR \blankR \rightend \\
\leftend  \blankL \blankL \arrR \blankR \rightend  &  \leftend  \arrL \blankR \blankR \blankR \rightend \\
\leftend  \blankL \blankL \blankL \arrR \rightend  & \leftend  \arrR \blankR \blankR \blankR \rightend  \\
\leftend  \blankL \blankL \blankL \arrL \rightend &  \\
\end{tabular}
\end{center}

There are five  states for Track 2: \zero, \one, \two, and \zeroB\ and \oneB.
We will have illegal pairs that will enforce that any legal configuration
on Track 2 must have the form $\leftend \one^* (\zeroB \zero^* +\oneB \two^*) \rightend$.
We will impose an ordering on the set of all possible such states and select transitions
so that in one Track 1 iteration, the state on Track 2 will advance from
one configuration to the next in the ordering.
Each state for Track 2 corresponds to a minute in our clock and
the states are partitioned into two phases. The first phase is called the
{\em Counting Phase} and consists of all the states
of the form $\leftend \one^* \zeroB \zero^*  \rightend$. The second phase is the
{\em Computation Phase} and consists of
all of the states of the form
$\leftend \one^* \oneB \two^*\rightend$.
The $\leftend \one^* \zeroB \zero^*  \rightend$ states are ordered according to the number of
particles in state $\one$ and
the $\leftend \one^* \oneB \two^*  \rightend$ states are ordered according to the number of
particles in state $\two$.
The state immediately after $\leftend \one^* \zeroB  \rightend$  in the ordering
is $\leftend \one^* \oneB  \rightend$. The ordering of the states for Track 2 is shown
below for a six particle system.

\begin{center}
\begin{tabular}{cc}
\leftend \zeroB \zero \zero \zero \rightend &  \leftend   \one \one \oneB \two \rightend \\
\leftend  \one \zeroB \zero \zero \rightend  & \leftend   \one \oneB \two \two \rightend \\
\leftend  \one \one \zeroB \zero \rightend  &  \leftend  \oneB \two \two \two \rightend \\
\leftend  \one \one \one \zeroB \rightend  &  \\
\leftend  \one \one \one \oneB \rightend &  \\
\end{tabular}
\end{center}

The arrow particle on Track 1 will cause a transition from one Track 2 state to  the next.
Each of the arrow states for Track 1 will come in three varieties: $\arrRzero$ and $\arrLzero$ will be
used during the initial minute of the clock when it is in state $\leftend \zeroB \zero^* \rightend$
and will be used to check initial conditions on the other tracks.
$\arrRone$ and $\arrLone$ will be used during the counting phase and
$\arrRtwo$ and $\arrLtwo$ will be used during the computation phase.
We need to describe how these different Track 1 symbols trigger transitions on
Track 2. The transitions for Track 1 will remain as described in that
the transition $\arrR \blankR \rightarrow \blankL \arrL$ will always cause a transition
from one of the $\arrR$ symbols to one of the other $\arrR$ symbols, but we need to
specify which $\arrR$ symbols are used in order to ensure that forward and backwards
transitions remain unique for each standard basis state.

We will use ordered pairs to denote a combined Track 1 and Track 2 state for
a particle as in $[ \blankR, \one ]$. The states for the left-most and right-most particles
are not divided into tracks and are simply $\leftend$ or $\rightend$. As space
permits we will denote the states for the different tracks vertically aligned. Examples
for neighboring particle states are given below:
$$\threecellsL{\arrR}{\two} \hspace{1in}
\threecellsR{\arrR}{\two} \hspace{1in}
\fourcells{\arrR}{\one}{\blankR}{\oneB}.$$
We will use the symbol \variable\ to denote a variable state. Thus $[\arrR, \variable]$ is
used to denote any state in which the Track 1 state is \arrR.

{\bf Clock States for the Initialization:}
When Track 1 has a $\arrRzero$ or $\arrLzero$ symbol, we want to ensure that
Track 2 is in state $\leftend \zeroB \zero^* \rightend$.
Therefore, we disallow  $[{\arrRzero},{\variable}]$ and
$[{\arrLzero},{\variable}]$ for any \variable\ that is not \zero\ or \zeroB.
We have the usual transitions that advance the Track 1 arrow:
$$\fourcells{\arrRzero}{\variable}{ \blankR}{\variable} \rightarrow \fourcells{\blankL } {\variable}{\arrRzero}{\variable} \hspace{1in}
\fourcells{\blankL }{\variable}{\arrLzero}{\variable} \rightarrow \fourcells{\arrLzero}{\variable}{ \blankR}{\variable} \hspace{1in}
\threecellsR{\arrRzero }{\variable} \rightarrow \threecellsR{\arrLzero }{\variable}$$
These happen regardless of the values on the other tracks and do not change
the values on the other Tracks.
The state $\leftend [{\arrRzero},{\zeroB}]$ does not have a transition in the
reverse direction since this only occurs in the initial state.
Finally, we have the transition $$\threecellsL {\arrLzero}{\zeroB} \rightarrow
\threecellsL{\arrRone}{\zeroB}$$
The presence of the $\zeroB$ on Track 2 specifies that $\arrRone$ should transition
to $\arrLzero$ in the reverse direction instead of \arrLone. The initial iteration of the second hand
is demonstrated on a chain of length six in figure~\ref{fig:initialization}.

\begin{figure}[htbp]
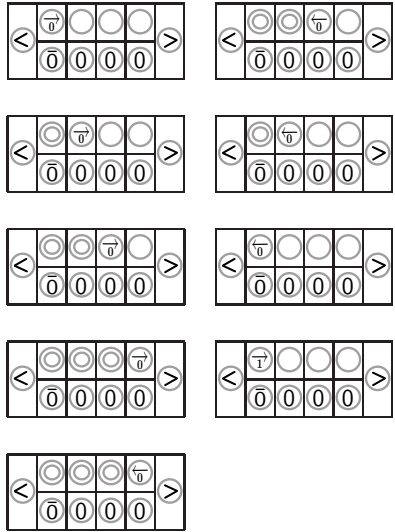

\begin{center}
\begin{tabular}{cc}
\eightcells{\arrRzero}{ \blankR}{ \blankR}{ \blankR}{\zeroB}{\zero}{\zero}{\zero}
&  \eightcells{\blankL}{ \blankL }{\arrLzero}{ \blankR}{\zeroB}{\zero}{\zero}{\zero} \\
& \\
\eightcells{  \blankL }{\arrRzero}{ \blankR}{ \blankR}{\zeroB}{\zero}{\zero}{\zero}
 & \eightcells{   \blankL }{\arrLzero}{ \blankR}{ \blankR }{\zeroB}{\zero}{\zero}{\zero} \\
& \\
\eightcells{  \blankL }{\blankL}{ \arrRzero }{\blankR }{\zeroB}{\zero}{\zero}{\zero}
&  \eightcells{  \arrLzero}{ \blankR}{ \blankR}{ \blankR }{\zeroB}{\zero}{\zero}{\zero} \\
& \\
\eightcells{  \blankL}{ \blankL}{ \blankL}{ \arrRzero}{\zeroB}{\zero}{\zero}{\zero}
  & \eightcells{  \arrRone}{ \blankR}{ \blankR}{ \blankR}{\zeroB}{\zero}{\zero}{\zero} \\
& \\
\eightcells{  \blankL}{ \blankL }{\blankL}{ \arrLzero }{\zeroB}{\zero}{\zero}{\zero}&  ~~\\
\end{tabular}
\caption{The initial iteration of the second hand on a chain of length six}
\label{fig:initialization}
\end{center}
\end{figure}

{\bf Clock States for the Counting Phase:}
During the counting phase, Track 2 will always be in some state of the form
$\leftend \one^* \zeroB \zero ^* \rightend$, so we will forbid the states
$[{\arrRone},{\variable}]$ and $[{\arrLone},{\variable}]$ when \variable\ is \two\ or \oneB.
The right-moving transitions will remain unchanged: $\arrRone \blankR \rightarrow \blankL \arrRone$.
These occur regardless of the contents of Track 2 and do not effect any change on Track 2.
The change in direction at the right end will depend on the contents of Track 2:
$$\threecellsR{\arrRone}{\zero}  \rightarrow \threecellsR{\arrLone}{\zero} \hspace{1in}
\threecellsR{\arrRone}{\zeroB}  \rightarrow \threecellsR{\arrLtwo}{\oneB} $$
The latter transition triggers the transition to the computation phase.
In the left-moving direction, the arrow will sweep past pairs of \zero\ particles
and pairs of \one\ particles:
$$\fourcells{\blankL}{\zero}{\arrLone}{\zero} \rightarrow \fourcells{\arrLone}{\zero}{\blankR}{\zero} \hspace{1in}
\fourcells{\blankL}{\one}{\arrLone}{\one} \rightarrow \fourcells{\arrLone}{\one}{\blankR}{\one}.$$
When the arrow on Track 1 sweeps left and meets the $\zeroB$ particle on Track 2,
it triggers an advance of the minute hand. This does not change the
transition on Track 1:
$$\fourcells{\blankL}{\zeroB}{\arrLone}{\zero} \rightarrow \fourcells{\arrLone}{\one}{\blankR}{\zeroB}.$$
Note that the $\arrLone$ never coincides with the $\zeroB$, so we will disallow the state $[\arrLone, \zeroB]$.
Finally, we have that \arrLone\ must turn at the left end:
$$ \threecellsL{\arrLone}{\one} \rightarrow  \threecellsL{\arrRone}{\one}.$$
We illustrate in figure~\ref{fig:counting} the first iteration of the counting phase and then the last iteration in
the counting phase. The very last transition
illustrates the transition to the computation phase.

\begin{figure}[htbp]
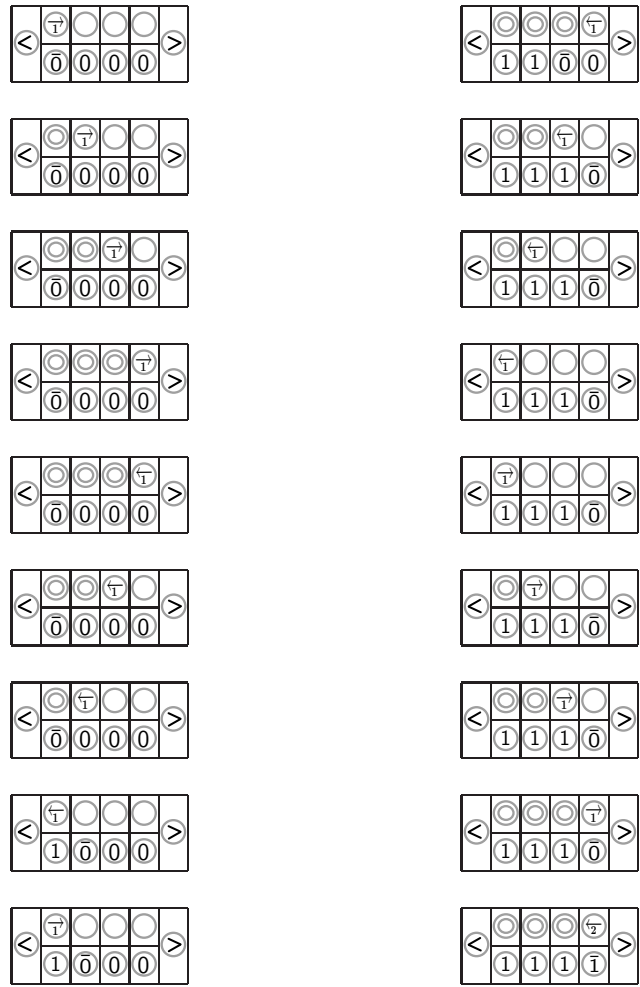

\begin{center}
\begin{tabular}{cc}

\begin{tabular}{c}
First iteration of the counting phase\\
\\
\eightcells{\arrRone}{ \blankR}{ \blankR}{ \blankR}{\zeroB}{\zero}{\zero}{\zero} \\  \\
\eightcells{\blankL}{ \arrRone}{ \blankR}{ \blankR}{\zeroB}{\zero}{\zero}{\zero} \\  \\
\eightcells{\blankL}{ \blankL}{ \arrRone}{ \blankR}{\zeroB}{\zero}{\zero}{\zero} \\  \\
\eightcells{\blankL}{ \blankL}{ \blankL}{ \arrRone}{\zeroB}{\zero}{\zero}{\zero} \\  \\
\eightcells{\blankL}{ \blankL}{ \blankL}{ \arrLone}{\zeroB}{\zero}{\zero}{\zero} \\  \\
\eightcells{\blankL}{ \blankL}{ \arrLone}{ \blankR}{\zeroB}{\zero}{\zero}{\zero} \\  \\
\eightcells{\blankL}{ \arrLone}{ \blankR}{ \blankR}{\zeroB}{\zero}{\zero}{\zero} \\  \\
\eightcells{\arrLone}{ \blankR}{ \blankR}{ \blankR}{\one}{\zeroB}{\zero}{\zero} \\  \\
\eightcells{\arrRone}{ \blankR}{ \blankR}{ \blankR}{\one}{\zeroB}{\zero}{\zero} \\  \\
\end{tabular}

\begin{tabular}{c}
Last iteration of the counting phase\\
\\
\eightcells{\blankL}{ \blankL}{ \blankL}{ \arrLone}{\one}{\one}{\zeroB}{\zero} \\  \\
\eightcells{\blankL}{ \blankL}{ \arrLone}{ \blankR}{\one}{\one}{\one}{\zeroB} \\  \\
\eightcells{\blankL}{ \arrLone}{ \blankR}{ \blankR}{\one}{\one}{\one}{\zeroB} \\  \\
\eightcells{\arrLone}{ \blankR}{ \blankR}{ \blankR}{\one}{\one}{\one}{\zeroB} \\  \\
\eightcells{\arrRone}{ \blankR}{ \blankR}{ \blankR}{\one}{\one}{\one}{\zeroB} \\  \\
\eightcells{\blankL}{ \arrRone}{ \blankR}{ \blankR}{\one}{\one}{\one}{\zeroB} \\  \\
\eightcells{\blankL}{ \blankL}{ \arrRone}{ \blankR}{\one}{\one}{\one}{\zeroB} \\  \\
\eightcells{\blankL}{ \blankL}{ \blankL}{ \arrRone}{\one}{\one}{\one}{\zeroB} \\  \\
\eightcells{\blankL}{ \blankL}{ \blankL}{ \arrLtwo}{\one}{\one}{\one}{\oneB} \\  \\
\end{tabular}

\end{tabular}
\caption{The first iteration and last iterations of the counting phase. The very last transition
illustrates the transition to the computation phase.}
\label{fig:counting}
\end{center}
\end{figure}

{\bf Clock States for the Computation Phase:}
During this time, Track 2 will always be in some state of the form
$\leftend \one^* \oneB \two ^* \rightend$, so we will forbid the states
$[{\arrRtwo},{\variable}]$ and $[{\arrLtwo},{\variable}]$ when \variable\
is \zero\ or \zeroB. Since $\arrRtwo$ should never be in the same configuration as
$\zeroB$, we will also disallow the pair $[\arrRtwo, \one][\blankR, \zeroB]$.
The left-moving transitions will remain unchanged: $\blankL \arrLtwo \rightarrow \arrLtwo \blankR $.
These occur regardless of the contents of Track 2 and do not effect any change on Track 2.
The change in direction at the left end will happen as long as there is a \one\ left in Track 2:
$$\threecellsL {\arrLtwo}{\one}   \rightarrow \threecellsL {\arrRtwo}{\one}.$$
When the state on Track 2 becomes $\leftend \oneB \two^* \rightend$ and the
left-moving arrow on Track 1 reaches the left end of the chain, we want the clock
to stop since this is the very last state. Thus,
there is no forward transition out of $\leftend [{\arrLtwo},{\oneB}]$.
In the right-moving direction, the arrow will sweep past pairs of \one\ particles
and pairs of \two\ particles:
$$\fourcells{\arrRtwo}{\one}{\blankR}{\one} \rightarrow \fourcells{\blankL}{\one}{\arrRtwo}{\one}
\hspace{1in}
\fourcells{\arrRtwo}{\two} {\blankR}{\two} \rightarrow \fourcells{\blankL}{\two}{\arrRtwo}{\two}$$
When the arrow on Track 1 sweeps right and meets the $\oneB$ particle on Track 2,
it triggers an advance of the minute hand:
$$\fourcells{\arrRtwo}{\one}{\blankR}{\oneB} \rightarrow
\fourcells{\blankL}{\oneB}{\arrRtwo}{\two}$$
Note that since the \arrRtwo\ never coincides with the \oneB,
we can make the state $[\arrRtwo,\oneB]$ illegal.
Finally, we have the turning at the right end:
$$  \threecellsR{\arrRtwo}{\two} \rightarrow   \threecellsR{\arrLtwo}{\two}$$
We illustrate in figure~\ref{fig:computation} the first iteration in the computation phase and then the last iteration. The very last state
shown is the final state for the clock.

\begin{figure}[htbp]
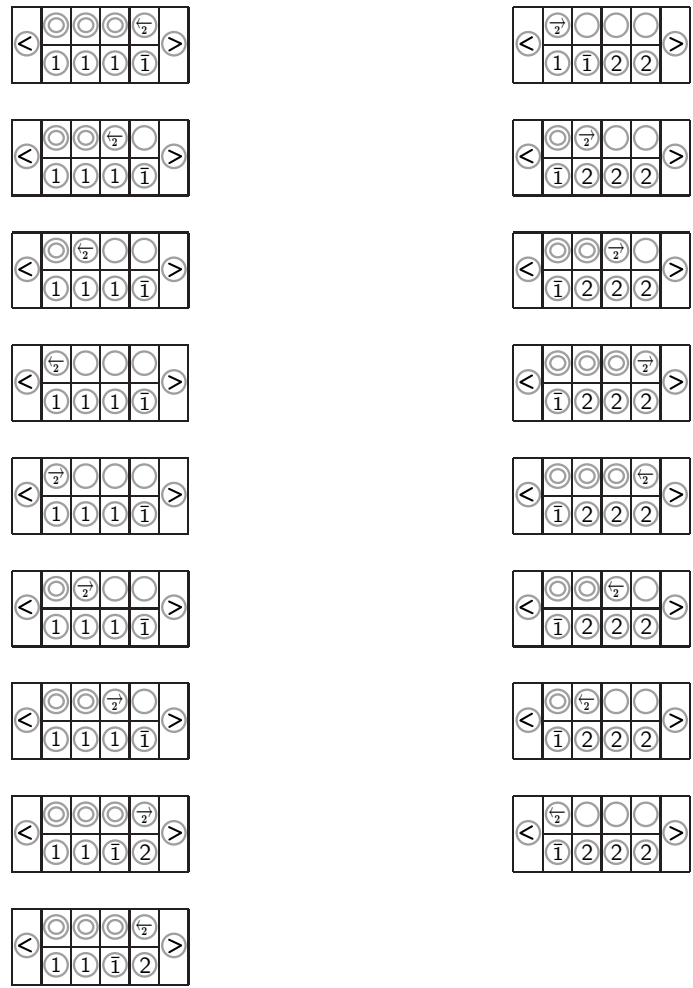

\begin{center}
\begin{tabular}{cc}

\begin{tabular}[t]{c}
First iteration of the Computation Phase\\
\\
\eightcells{\blankL}{ \blankL}{ \blankL}{ \arrLtwo}{\one}{\one}{\one}{\oneB} \\  \\
\eightcells{\blankL}{ \blankL}{ \arrLtwo}{ \blankR}{\one}{\one}{\one}{\oneB} \\  \\
\eightcells{\blankL}{ \arrLtwo}{ \blankR}{ \blankR}{\one}{\one}{\one}{\oneB} \\  \\
\eightcells{\arrLtwo}{ \blankR}{ \blankR}{ \blankR}{\one}{\one}{\one}{\oneB} \\  \\
\eightcells{\arrRtwo}{ \blankR}{ \blankR}{ \blankR}{\one}{\one}{\one}{\oneB} \\  \\
\eightcells{\blankL}{ \arrRtwo}{ \blankR}{ \blankR}{\one}{\one}{\one}{\oneB} \\  \\
\eightcells{\blankL}{ \blankL}{ \arrRtwo}{ \blankR}{\one}{\one}{\one}{\oneB} \\  \\
\eightcells{\blankL}{ \blankL}{ \blankL}{ \arrRtwo}{\one}{\one}{\oneB}{\two} \\  \\
\eightcells{\blankL}{ \blankL}{ \blankL}{ \arrLtwo}{\one}{\one}{\oneB}{\two} \\  \\
\end{tabular}

\begin{tabular}[t]{c}
Last iteration of the Computation Phase\\
\\
\eightcells{\arrRtwo}{ \blankR}{ \blankR}{ \blankR}{\one}{\oneB}{\two}{\two} \\  \\
\eightcells{\blankL}{ \arrRtwo}{ \blankR}{ \blankR}{\oneB}{\two}{\two}{\two} \\  \\
\eightcells{\blankL}{ \blankL}{ \arrRtwo}{ \blankR}{\oneB}{\two}{\two}{\two} \\  \\
\eightcells{\blankL}{ \blankL}{ \blankL}{ \arrRtwo}{\oneB}{\two}{\two}{\two} \\  \\
\eightcells{\blankL}{ \blankL}{ \blankL}{ \arrLtwo}{\oneB}{\two}{\two}{\two} \\  \\
\eightcells{\blankL}{ \blankL}{ \arrLtwo}{ \blankR}{\oneB}{\two}{\two}{\two} \\  \\
\eightcells{\blankL}{ \arrLtwo}{ \blankR}{ \blankR}{\oneB}{\two}{\two}{\two} \\  \\
\eightcells{\arrLtwo}{ \blankR}{ \blankR}{ \blankR}{\oneB}{\two}{\two}{\two} \\  \\
\end{tabular}

\end{tabular}
\caption{The first and last iterations of the computation phase.  The very last state
shown is the final state for the clock.}
\label{fig:computation}
\end{center}
\end{figure}

Before describing the states and transitions for the other  tracks, we will
pause to consider the Hamiltonian created so far.
A single two-particle term will be the sum of the terms from the transition rules and
the illegal pairs described so far.
Suppose that a state of a particle is described only by its state on Tracks 1 and 2.
Let $H_N$ be the system resulting from applying this term to every neighboring
pair in a chain of $N$ particles.
What is the dimension and structure of the ground state of $H_N$ restricted to
$\calS_{br}$?

\begin{definition}
We will say that a standard basis state is {\em well-formed} if the Track 1 state
corresponds to a string in the regular language $\leftend \blankL^* ( \arrR + \arrL) \blankR^* \rightend$
and if  the Track 2 state corresponds to a string in the regular language
$\leftend \one^* (\zeroB \zero^* +\oneB \two^*) \rightend$.
\end{definition}

Note that if a standard basis state is well-formed it must be in $\calS_{br}$.

\begin{lemma}
\label{lem:well-formed}
Consider a well-formed standard basis state.
There is at most one transition rule that applies to the state in the
forward direction and at most one that applies to the state in the reverse
direction. Furthermore, the set of
well-formed states is closed under the transition rules.
\end{lemma}

\begin{proof} The rules are summarized in table~\ref{table:transitionrules12}.
If only the Track 1 states are specified then the rule holds for any states on Track 2 and
does not alter the Track 2 states.
No pair appears twice
on the left-hand side of a transition rule
and no pair appears twice in the right-hand side of a transition rule.
It can be verified that the application of a transition rule maintains the condition of
being well-formed.

\begin{table}[htbp]
\[
\begin{array}{ccc}

\begin{array}[t]{|c|}
\hline
\mbox{Right-moving Control} \\
\hline
\arrRzero \rightend \rightarrow \arrLzero \rightend \\
\hline
\arrRzero \blankR \rightarrow \blankL \arrRzero \\
\hline
\\
\threecellsR{\arrRone}{\zero} \rightarrow \threecellsR{\arrLone}{\zero} \\
\\
\hline
\\
\threecellsR{\arrRone}{\zeroB} \rightarrow \threecellsR{\arrLtwo}{\oneB} \\
\\
\hline
\arrRone \blankR \rightarrow \blankL \arrRone \\
\hline
\arrRtwo \rightend \rightarrow \arrLtwo \rightend \\
\hline
\\
\fourcells{\arrRtwo}{\one}{\blankR}{\one} \rightarrow \fourcells {\blankL}{\one}{\arrRtwo}{\one} \\
\\
\hline
\\
\fourcells{\arrRtwo}{\two}{\blankR}{\two} \rightarrow \fourcells {\blankL}{\two}{\arrRtwo}{\two} \\
\\
\hline
\\
\fourcells{\arrRtwo}{\one}{\blankR}{\oneB} \rightarrow \fourcells {\blankL}{\oneB}{\arrRtwo}{\two} \\
\\
\hline
\\
\threecellsR{\arrRtwo}{\two} \rightarrow \threecellsR{\arrLtwo}{\two} \\
\\
\hline

\end{array}

\hspace{1in}

\begin{array}[t]{|c|}
\hline
\mbox{Left-moving Control} \\
\hline
\\
 \threecellsL {\arrLzero}{ \zeroB} \rightarrow  \threecellsL{ \arrRone}{\zeroB} \\
\\
\hline
\blankL \arrLzero  \rightarrow \arrLzero \blankR \\
\hline
\\
\threecellsL {\arrLone}{ \one} \rightarrow  \threecellsL{ \arrRone}{\one} \\
\\
\hline
\\
\fourcells{\blankL}{\zero}{\arrLone}{\zero} \rightarrow \fourcells {\arrLone}{\zero}{\blankR}{\zero} \\
\\
\hline
\\
\fourcells{\blankL}{\one}{\arrLone}{\one} \rightarrow \fourcells {\arrLone}{\one}{\blankR}{\one} \\
\\
\hline
\\
\fourcells{\blankL}{\zeroB}{\arrLone}{\zero} \rightarrow \fourcells {\arrLone}{\one}{\blankR}{\zeroB} \\
\\
\hline
\blankL \arrLtwo \rightarrow \arrLtwo \blankR \\
\hline
\\
\threecellsL{\arrLtwo}{\one} \rightarrow \threecellsL{\arrRtwo}{\one}\\
\\
\hline
\end{array}

\end{array}
\]
\caption{The transition rules for Tracks $1$ and $2$.}
\label{table:transitionrules12}
\end{table}

\end{proof}

Consider the set of states that correspond to the valid clock states beginning with
$\leftend \arrRzero \blankR^* \rightend$ on Track 1 and
$\leftend \zeroB \zero^* \rightend$ on Track 2 and ending with $\leftend \arrLtwo \blankR^* \rightend$
on Track 1 and $\leftend \oneB \two^* \rightend$ on Track 2.
There are exactly $4(N-2)^2$ such states in this sequence.
Let $\ket{\phi_{cl}}$ be  the uniform superposition of these states.

\begin{lemma}
\label{lem:clairvoyance}
Consider a well-formed standard basis state $s$ that is not in the
support of $\ket{\phi_{cl}}$. Then for some $i \le 2N$, $s$ will reach a state with an
illegal pair after $i$ applications of the transition rules in either
the forward or reverse direction.
\end{lemma}

\begin{proof}
We will argue that for each well-formed standard basis state, it is either in the support of
$\ket{\phi_{cl}}$, has an illegal pair or is within $2N$ transitions of a state which
contains an illegal pair. In some cases, we have added penalties for a particular
particle state (e.g., $[\arrRzero,\two]$).
This can easily be handled with illegal pairs by making any pair illegal which contains
that particular state.

Every Track 1 configuration of the form $\leftend \blankL^* ( \arrRzero + \arrLzero) \blankR^* \rightend$
occurs in $\ket{\phi_{cl}}$. Moreover, for every standard basis state in the support of $\ket{\phi_{cl}}$,
if the control state is $\arrRzero$ or $\arrLzero$, then Track 2 is in state
$\leftend \zeroB \zero^* \rightend$. Now consider a state in which
 Track 2 has a $\one$, $\oneB$ or $\two$ particle and
the control particle on Track 1 is $\arrRzero$ or $\arrLzero$. The control particle will transition
(in either the forward or reverse direction) towards the $\one$, $\oneB$ or $\two$ particle
on Track 2 and eventually they will coincide. This happens in fewer than $N$ moves and
will result in an illegal state.

Now consider the standard basis states whose control particle is $\arrRone$ or $\arrLone$.
Every possible combination of $\leftend \blankL^* ( \arrRone + \arrLone) \blankR^* \rightend$
occurs in $\ket{\phi_{cl}}$ with every possible combination of $\leftend \one^* \zeroB \zero^* \rightend$
except those where the control particle is $\arrLone$ and coincides with the
$\zeroB$. The state $[\arrLone, \zeroB]$ is an illegal particle state.
So we now need to take care of the case where the
control particle on Track 1 is $\arrRone$ or $\arrLone$ and  Track 2 is a $\leftend \one^* \oneB \two^* \rightend$
state. Since it is illegal for a $\arrRone$ or $\arrLone$ to coincide with a $\oneB$ or $\two$
on Track 2, we can assume that the control  particle on Track 1 is
over a $\one$ particle on Track 2.
In this case, it will transition in the forward direction in fewer than $2N$ steps (possibly
turning at the left end of the chain) until
it reaches the $\oneB$. This will result in a particle in state $[\arrRone, \oneB]$ which is illegal.

Finally, consider the case where the control particle is $\arrRtwo$ or $\arrLtwo$.
Every possible combination of the expression $\leftend \blankL^* ( \arrRtwo + \arrLtwo) \blankR^* \rightend$
occurs in $\ket{\phi_{cl}}$ with every possible combination of $\leftend \one^* \oneB \two^* \rightend$
except those where the control particle is $\arrRtwo$ and coincides with the
$\oneB$. The state $[\arrRtwo,\oneB]$ is illegal. So we now need to take care of the case where the
control particle on Track 1 is $\arrRtwo$ or $\arrLtwo$ and  Track 2 is a $\leftend \one^* \zeroB \zero^* \rightend$
state. Since it is illegal for a $\arrRtwo$ or $\arrLtwo$ to coincide with a $\zeroB$ or $\zero$
on Track 2, we can assume that the control  particle on Track 1 is
over a $\one$ particle on Track 2.
In this case, it will transition in the forward direction in fewer than $2N$ steps (possibly
turning at the left end of the chain) until
it is just before  the $\zeroB$ on Track 2. This will result in a pair $[\arrRtwo, \one][\blankR, \zeroB]$ which is illegal.

\end{proof}

\begin{lemma}
\label{lem:gap}
$\ket{\phi_{cl}}$ is the unique ground state of $H_N |_{\calS_{br}}$.
All other eigenstates have an energy that is at least $\Omega(1/N^3)$.
\end{lemma}

\begin{proof}
The proof of this lemma follows from the
standard techniques used for showing $\QMA$-completeness
of $1$-dimensional Hamiltonians \cite{QMA1D},
so we only give a brief sketch here.
The idea is that any standard basis state inside $\calS_{br}$ which is not
well-formed will have energy at least one from the illegal pairs, so we
can restrict our attention to well-formed states.
Now we create a graph over well-formed standard basis states. There is an edge from
state $a$ to state $b$ if $b$ can be reached from $a$ by the application
of one transition rule in the forward direction. Lemma \ref{lem:well-formed} implies that
this graph is composed of disjoint directed paths. Call $H_{trans}$ the Hamiltonian
resulting from the sum of all the terms from transition rules and $H_{legal}$ the
Hamiltonian from illegal pairs. The subspace spanned by the states in a single
path is closed under $H_{trans}$ and $H_{legal}$. Furthermore, the matrix for
$H_{trans}$ restricted to the states in a path looks like
\[
\left(
\begin{array}{rrrrrrr}
\smfrac{1}{2} & \mns \smfrac{1}{2} &0 & & \cdots& & 0 \\ \mns \smfrac{1}{2} & 1 & \mns \smfrac{1}{2} & 0 &
\ddots & & \vdots\\ 0 & \mns \smfrac{1}{2} & 1 & \mns \smfrac{1}{2} & 0 & \ddots & \vdots\\ & \ddots & \ddots
& \ddots & \ddots & \ddots & \\ \vdots& & 0 & \mns \smfrac{1}{2} &1 & \mns \smfrac{1}{2}& 0 \\ & & & 0 & \mns
\smfrac{1}{2} &1 & \mns \smfrac{1}{2} \\ 0& & \cdots& & 0&
\mns \smfrac{1}{2} & \smfrac{1}{2} \\
\end{array}
\right)
\]
where the dimension of the matrix is the length of the path.
The matrix for $H_{legal}$ is diagonal in the standard basis.
If the path contains no states with illegal pairs, then $H_{legal}$ is zero.
The unique ground state has zero energy and
is the uniform superposition of all standard basis states
in the path. Moreover, the next highest energy state has energy at least $1/L$,
where $L$ is the length of the path. In our case, Lemma \ref{lem:clairvoyance}
indicates there is exactly one path
with no illegal states which corresponds to $\ket{\phi_{cl}}$. The length of this path
is $O(N^2)$, so the next highest eigenvalue in this space is at least $\Omega(1/N^2)$.

Now consider a path with some illegal states and suppose that the ratio of illegal
states in the path to the length of the path is $1/s$. Again, it is known \cite{QMA1D}
that any state in
the subspace spanned by the states in this path will have energy be at least $1/s^3$.
By Lemma \ref{lem:clairvoyance}, we know that for any path which does not correspond
to $\ket{\phi_{cl}}$ the ratio of illegal states to the total number of
states at least $1/2N$, which means that any state
in this subspace will have energy at least $\Omega(1/N^3)$.
\end{proof}

The remaining subsections describe how the arrow on Track 1 interacts with
the other four tracks. The only role of Track 2 is to record the time.
It causes the control state on Track 1 to transition from the counting
phase to the computation phase and finally to stop iterating at the
end of the computation phase.

\subsection{Initialization Phase}

Track 3 will be used as the tape for both Turing machines. Therefore the set
of track 3 states for a particle consists of the union of the two alphabets
for the two Turing machines.
We want to ensure that this track is initialized to all blank symbols,
so
we disallow any state in which Track 1 is in \arrRzero\ and Track 3 has anything
but $\blankR$.

Track 4 will store the location and state for $M_{BC}$. Therefore, we will add
illegal pairs to ensure that the Track 4 state for the system always has
the form $\leftend \blankL^* \variable \blankR^* \rightend$, where $\variable$
is a state for $M_{BC}$.
We assume that $M_{BC}$ has a designated start state $\stateS$, so we want
to ensure that Track 4 starts out in the configuration
$\leftend \stateS \blankR^* \rightend$.
To do this, we make any state illegal in which the Track 1 state is
$\arrRzero$ and the Track 4 state is anything except $\stateS$ or $\blankR$.
Furthermore, for any pair in which the left particle is in state \leftend,
and the Track 1 state is $\arrRzero$, the Track 4 state must be \stateS.
For any pair in which the left particle is not in state \leftend,
and the Track 1 state is $\arrRzero$, the Track 4 state must be \blankR.
The initial conditions for Track 5 are similar, except that the starting
state for $V$ is used.
Since Track 6 holds the quantum witness for $V$, it can be any
state on $n$ qubits. A state of a particle is now specified by a $6$-tuple.

\subsection{Counting Phase}

Every time the $\arrRone$ sweeps from the left end of the chain
to the right end of the chain, it causes $M_{BC}$ to execute one
more step. Thus, $M_{BC}$ is run for exactly $N-2$ steps.
Since the classical, reversible Turing machine $M_{BC}$ is a special
case of a quantum Turing machine, we describe the details for the
more general case. The $\arrRtwo$ symbol is what causes the Turing machine
$V$ to execute a step.

\subsection{Computation Phase}

We will examine a particular quantum rule  and explain the desired behavior of our machine.

Consider a pair $(q,a)$, where $q \in Q$ and
$a \in \Sigma$. $a$ will encompass the state on the work tape as well
as the state on the witness tape. Since the \arrRtwo\ particle triggers
the execution of a step, we will consider particles as triplet states
of the form $[\arrRtwo, q, a]$, $[\blankR , q, a]$ and $[\blankL , q, a]$,
for $a \in \Sigma$ and $q \in Q \cup \{ \blankR \} \cup \{ \blankL \}$.
$q$ specifies the Track 5 state, and $a$ tells us the state of Track 3 and Track 6.
It will be convenient to refer to $q$ as a generic Track 5 state of a particle
which could be a state from $Q$ as well as \blankR\ or \blankL.
We will also be interested in the {\em computation state} of a particle
which just consists of a pair $[q,a]$, where
$q \in Q \cup \{ \blankR \} \cup \{ \blankL \}$ and $a \in \Sigma$.

If the TM is in state $q$ and in a location with an $a$
on the tape, the QTM defines a superposition of possible next moves.
Let $\delta(q,a,p,b,D)$ denote the amplitude that the next state
will be a configuration in which the state is $p$, the
tape symbol is overwritten by $b$ and the head moves in direction $D$.
We will use a fact established in \cite{bernsteinVazirani} that the states of
a quantum Turing machine $Q$ can be partitioned into two sets $Q_L$ and $Q_R$ such that
states in $Q_L$ can be reached only by moves in which the head moves left
and $Q_R$ can be reached only by moves in which the head moves right.
We will use $q_L$ to designate a generic element of $Q_L$ and
$q_R$  to designate a generic element of $Q_R$.

We will need to execute the moves in which the head moves left separately
from the moves in which the head moves right. In order to do this, we introduce a new state
$q_R'$ for every state $q_R \in Q_R$. We will call this set
$Q'_R$. It will be forbidden to have the Turing machine head
in a state from $Q'_R$  without also having the \arrRtwo\ state on Track 1
in the same location. It will also be forbidden to have the  \arrRtwo\ in the
same location as $q$ if $q \in Q_L$.
We will be interested in two particular subspaces defined on
the computation state for a pair of neighboring particles
$[q,a][p,b]$, where $a,b \in \Sigma$
and $q,p \in Q \cup \{ \blankR \} \cup \{ \blankL \}$.
$\calS_A$ is the space spanned by all such legal basis states such that
$q \not\in Q_L$, $p \not\in Q'_R$.  Note that if the Track 1 state for a pair of
particles is $\ket{ \arrRtwo \blankR}$, then
the only possible computation states for the pair are in $\calS_A$.
We also define $\calS_B$ to be the space spanned by all legal basis states
of the form $[q,a][p,b]$ such that
$p \not\in Q_L$, $q \not\in Q'_R$.
Similarly, if the Track 1 state for a pair of particles is $\ket{ \blankL \arrRtwo }$, then
the only possible legal computation states for the pair are in $\calS_B$.

We will describe a transformation $T$ that maps $\calS_A$ to $\calS_B$.
The extended map $T \otimes \ketbra{\blankL \arrRtwo}{\arrRtwo \blankR}$
as it is applied to each pair of particles in turn will implement a
step of the Turing machine. We will call this extended map a {\em transition  rule} since it
carries one state to another in the same way that the transition rules for the clock state did.
The difference now is that the states will in general be quantum states.

The transformation on $\calS_A$ works in two parts.
At the moment that the Track 1 arrow moves from the position to the left
of the head, we execute the move for that location. For every $(q_L,b,L)$,
with amplitude $\delta(q,a,q_L,b,L)$,
we write a $b$ in the old head location and move the head left into state $q_L$.
For every $(q_R,b,R)$, with amplitude $\delta(q,a,q_R,b,R)$, we write a $b$
in the old head location, transition to state $q_R'$ and leave the head
in the same  location.
We need to defer the action of moving the head right until we have access
to the new location. In the next step of the clock, when the \arrRtwo\
is aligned with the state $q_R'$, we move it right and convert it to $q_R$.

That is, we want a sequence of two transitions.
\begin{equation}
\label{eq:T1}
\left| \ \sixcells{\arrRtwo}{\blankR}{\blankL}{q}{\variable}{a} \ \right\rangle \rightarrow
\sum_{a,b,q_L,q}  \delta(q,a,q_L,b,L) \left| \ \sixcells{\blankL}{\arrRtwo}{q_L}{\blankR}{\variable}{b} \ \right\rangle
+ \sum_{a,b,q_R,q}  \delta(q,a,q_R,b,R) \left| \ \sixcells{\blankL}{\arrRtwo}{\blankL}{q_R'}{\variable}{b} \ \right\rangle.
\end{equation}
After this step, the configuration is a superposition of states in which the step has
been performed on the configurations with the head in the same location as the
\arrRtwo\ state, except that moving the head to the right has been deferred.
If the TM state is primed and is aligned with the location of
the \arrRtwo, that triggers the execution of the right-moving step.
\begin{equation}
\label{eq:T2}
\left| \ \sixcells{\arrRtwo}{\blankR}{q_R'}{\blankR}{\variable}{\variable} \ \right\rangle \rightarrow
 \left| \ \sixcells{\blankL}{\arrRtwo}{\blankL}{q_R}{\variable}{\variable} \ \right\rangle.
\end{equation}
Otherwise, the \arrRtwo\ state just sweeps to the right, leaving the other tracks
unchanged:
\begin{equation}
\label{eq:T3}
\left| \ \sixcells{\arrRtwo}{\blankR}{q_R}{\blankR}{\variable}{\variable} \ \right\rangle \rightarrow
 \left| \ \sixcells{\blankL}{\arrRtwo}{q_R}{\blankR}{\variable}{\variable} \ \right\rangle
\hspace{1in}
\left| \ \sixcells{\arrRtwo}{\blankR}{\blankR}{\blankR}{\variable}{\variable} \ \right\rangle \rightarrow
 \left| \ \sixcells{\blankL}{\arrRtwo}{\blankR}{\blankR}{\variable}{\variable} \ \right\rangle
\hspace{1in}
\left| \ \sixcells{\arrRtwo}{\blankR}{\blankL}{\blankL}{\variable}{\variable} \ \right\rangle \rightarrow
 \left| \ \sixcells{\blankL}{\arrRtwo}{\blankL}{\blankL}{\variable}{\variable} \ \right\rangle
\end{equation}
Note that transformations (\ref{eq:T1}), (\ref{eq:T2}) and
(\ref{eq:T3}) define $T$ on every state in $\calS_A$. Furthermore, the image of
$T|_{\calS_A}$ is in $\calS_B$.
A critical fact is that after $T$ is applied to a pair, the computation state of
that pair is in $\calS_B$, which implies that the next pair over is now in subspace
$\calS_A$ and $T$ can be applied to this new pair.
 We need to now show that $T$ has the
necessary properties to be expressed as a Type II Hamiltonian term with the $\ket{cd}$ final
state a superposition.
\begin{claim}
$T$ is a partial isometry (meaning it preserves inner products) when restricted to
$\calS_A$.
\end{claim}
Since $\calS_A$ and $\calS_B$ have the same dimension, the claim implies that
$T^{\dagger}$ is well-defined and a partial isometry on $\calS_B$. We can extend $T$
to be a unitary map on the full Hilbert space.
Then we select the Hamiltonian term to be:
$$I_{\calS_A} \otimes \ketbra{\arrRtwo  \blankR}{\arrRtwo  \blankR}
+ I_{\calS_B} \otimes \ketbra{\blankL \arrRtwo}{\blankL \arrRtwo}
- T \otimes \ketbra{\blankL \arrRtwo}{\arrRtwo  \blankR}
- T^{\dagger} \otimes \ketbra{\arrRtwo  \blankR}{\blankL \arrRtwo} .$$

\begin{proof}[ of claim]
Expressions (\ref{eq:T1}), (\ref{eq:T2}) and (\ref{eq:T3})
each define $T$ on a different orthogonal subspace of $\calS_A$.
$\calS_{A1}$ is the space of the form $[\blankL,\variable][q,a]$,
$\calS_{A2}$ is the space of states of the form $[q'_R,\variable][\blankR,\variable]$
and $\calS_{A3}$ is the space spanned by states of the form
$[q_R,\variable][\blankR,\variable]$, $[\blankR,\variable][\blankR,\variable]$
and $[\blankL,\variable][\blankL,\variable]$.
We will argue that $T$ is a partial isometry when restricted to each of
$\calS_{A1}$, $\calS_{A2}$ and $\calS_{A3}$.
Furthermore the images of $T_{\calS_{A1}}$, $T_{\calS_{A2}}$ and $T_{\calS_{A3}}$
are mutually orthogonal.
This implies that $T$ is a partial isometry on $\calS_A$.

$T$ is clearly a partial isometry on each of $\calS_{A2}$ and $\calS_{A3}$.
We only need to examine the transformation given in (\ref{eq:T1}).
We know that the application of a move of the QTM is unitary. In particular,
it is unitary when restricted to the space of all configurations with
the head in a particular location. This means that the transformation (\ref{eq:T1})
($T$ restricted to states of the form $[{\blankL}, \variable][q,a]$) is a partial isometry.
\end{proof}

By construction, a single iteration will implement the correct transition rule for the
TM head in a single specific position.  We should also verify that it acts correctly when
the head is in a superposition of locations.
In particular, when the TM
head moves right from location $i$, it can end up in the same location as when it moves
left from location $i+2$.  In our implementation, when each step is implemented, these
steps are orthogonal because the Track 1 state differs.  They are also orthogonal in
the original TM rule, since after the head moves right, it will be in a state from $Q_R$,
whereas after it moves left, it will be in a state from $Q_L$, which is orthogonal to $Q_R$.
Thus, our implementation is globally performing the correct TM rule.

We have argued that the extended
transformations (or transition rules) which implement the Turing Machine behave as we
would like them to as long as the states to which they are applied obey certain
conditions. The following definition formalizes those conditions.

\begin{definition}
We say that a state is {\em invalid} if one of the following conditions holds:
\begin{enumerate}
\item The control particle on Track 1 is on site $i$ and is in state \arrRone. Furthermore,
for some particle other than $i$, its Track 4 state is $q'_R$ for some $q'_R \in Q'_R$.
\item The control particle on Track 1 is on site $i$ and is in state \arrRone. Furthermore,
the Track 4 state for particle $i$ is $q_L$ for some $q_L \in Q_L$.
\item The control particle on Track 1 is on site $i$ and is in state \arrRtwo. Furthermore,
for some particle other than $i$, its Track 5 state is $q'_R$ for some $q'_R \in Q'_R$.
\item The control particle on Track 1 is on site $i$ and is in state \arrRtwo. Furthermore,
the Track 5 state for particle $i$ is $q_L$ for some $q_L \in Q_L$.
\end{enumerate}
\end{definition}

We say that a standard basis state is {\em valid} if it is not invalid.
We can add illegal pairs so that any invalid state has an energy penalty of at least one.
The subspace of all valid states
is closed under the action of the transition rules.
As before, we call a state well-formed if its clock state (Track 1 and 2) is well-formed.
For any valid, well-formed state, the transition rules apply non-trivially to only
one pair of particles. Furthermore, the transition rules are reversible and norm-preserving
over the space spanned by  valid, well-formed states.

\subsection{Combining the Tracks}
\label{sec:combiningtracks}

We have described a set of transition rules for Tracks 1 and 2 which advance a clock
through $T=4(N-2)^2$ states.
We have also described a set of transition rules which implement two Turing Machines
on tracks 3 through 6. In each set of transition rules, the control particle of the clock
on Track 1 advances one space. Thus, when they are combined, each transition rule
advances the clock and implements part of a step on one of the Turing Machines.
Consider any state such that the states for Tracks 1 and 2 are clock states
and the states for Tracks 3 through 6 are a quantum state within the subspace of
valid states. If the clock state is not the final state, then the transition rules
apply to exactly one pair of particles and
take the state to a unique new state with the same norm in which the clock has
advanced by one tick. Similarly, if the clock state is not the initial state, then
the transition rules applied in the reverse direction apply to exactly one
pair of particles and take the state to a new
state with the same norm in which the clock has gone back by one tick.

We argued in Subsection \ref{sec:clock} that a zero eigenstate of the Hamiltonian
must be a superposition
of the sequence of clock states. Now that we are considering the states of the other tracks
as well, this is a higher-dimensional space. In order to define a basis for the zero eigenspace
of the Hamiltonian defined so far, we specify a standard basis state for Tracks 3 through 6
for the initial clock state. It will be convenient to separate the states of Tracks 3 through 5
from the state of Track 6. So let $X$ be a standard basis state for Tracks 3 through 5 and $Y$ be a
standard basis state for Track 6.
Let $\ket{t}$ be the state of the clock after $t$ transitions. Define $\ket{\phi_{t,X,Y}}$ denote
the state of Tracks 3 through 6 after $t$ applications of  the transition rules assuming that
Tracks 3 through 5 start in state $X$ and Track 6 starts in state $Y$. Define
$\ket{\phi_{X,Y}} = \sum_{t=0}^{T-1} \ket{t} \ket{\phi_{t,X,Y}}$.
The $\ket{\phi_{X,Y}}$ satisfy all the constraints due to transition rules and illegal states
for Tracks 1 and 2.
We will use $I$ to denote the desired initial configuration for Tracks 3 through 5.
That is, $\blankR^*$ on Track 3 and $\stateS \blankR^*$ and $q_0 \blankR^*$ on Tracks 4 and 5.
The energy for any $\ket{\phi_{X,Y}}$ where $X \neq I$ will be at least $\Omega(1/N^2)$.
This is because if $X \neq I$,
at least one state in the initialization phase will have an energy penalty
as the $\arrRzero$ control state sweeps to the right. So the subspace spanned by
$\{ \ket{\phi_{I,Y}} \}$ for all $Y$ form a basis of the zero eigenspace for the
Hamiltonian terms defined so far, restricted to $\calS_{br}$,
the subspace of all bracketed states. The next highest eigenstate in $\calS_{br}$
has energy $\Omega(1/N^3)$.

\subsection{Boundary Conditions}
\label{sec:quantumperiodic}

\subsubsection{The Finite Chain}

First we address the case of a finite chain and add a term that will enforce
that the ground state is in $\calS_{br}$.
In order to penalize states that are not bracketed, we weight $H$ by a factor of $3$
and add in the term $I - \ketbra{\rightend}{\rightend} - \ketbra{\leftend}{\leftend}$
to every particle. This has the effect of adding energy $N-2$ to every state
in $\calS_{br}$. The ground space is still spanned by the $\ket{\phi_{I,Y}}$
and the spectral gap remains $\Omega(1/N^3)$.
Any standard basis state outside $\calS_{br}$ will have a cost of at least $N-1$:
If there are $a$ particles in state $\rightend$ or $\leftend$\
in the middle of the chain, they will save
$a$ from the $I - \ketbra{\rightend}{\rightend} - \ketbra{\leftend}{\leftend}$
term. However,  there will be at least $\lceil a/2 \rceil$ illegal pairs
because the \leftend\ particles in the middle will each form an illegal pair with the
particle to their left and the \rightend\ particles in the middle will each
form an illegal pair
with the particles to their right. We can lower bound the additional cost by
counting the type (\leftend\ or \rightend) which is more plentiful,
yielding at least $\lceil a/2 \rceil$ distinct illegal pairs.
Thus, there will be a net
additional energy cost of at least $a/2$ when $a$ is even, or $(a+3)/2$ when $a$ is
odd.

\subsubsection{The Cycle}

If we instead have periodic boundary conditions, the problem is still $\QMAEXP$-complete.
We start by removing the penalty for the pair $\rightend \leftend$.
The set of legal and well-formed
states is exactly the same as it was for the finite chain
except that we can now have
more than one segment around the cycle.
For example, we could have
the following state wrapped around a cycle:
$$\underbrace{\leftend ~~~~ \cdots ~~~~  \rightend}_{\mbox{Segment~} 1}
\underbrace{\leftend ~~~~  \cdots ~~~~  \rightend}_{\mbox{Segment~} 2}
\underbrace{\leftend ~~~~  \cdots ~~~~  \rightend}_{\mbox{Segment~} 3}.$$
Note that the transition rules do not change the locations or numbers of $\leftend$ or
$\rightend$ sites, so we can restrict our attention to subspaces in which the
segments are fixed.
If a standard basis state is well-formed then every occurrence of $\rightend$
has a $\leftend$ to its immediate right and every occurrence of $\leftend$
has a $\rightend$ it its immediate left. Thus, we can assume that a standard basis
state in the support of a ground state can be divided into valid segments. Of course, it
is possible that there are no segments in which case the state could simply be a
 single particle state repeated
around the entire cycle.
We know that the states within a segment must be ground states for a chain
of that length. Otherwise, the energy of the state is at least $\Omega(1/l^3)$,
where $l$ is the length of the segment.
We need to add additional terms and states which give an energy penalty to states
with no segments or more than one segment.

It is most straightforward to achieve this by
using odd $N$ and adding an additional Track $0$.
There will be no transition rules that apply to Track 0, so the state of the Track 0
remains fixed even as transition rules apply to the state.
We will add some extra terms which are
diagonal in the standard basis for Track 0 and introduce energy penalties for certain
pairs of particles.
Fix a standard basis state for Track 0 and a set of segments and consider the space
spanned by these states. The Hamiltonian is closed over this space, so we just
need to examine the eigenstates restricted to each such subspace.
Track $0$ will
have three possible states $\stateA$, $\stateB$, and $\stateLine$.
Even the \rightend\ or \leftend\ particles have a Track 0 state, so the state
of a particle is either in $\{ \stateA, \stateB, \stateLine \} \times \{\leftend, \rightend\}$ or
it is a $7$-tuple denoting states for Tracks 0 through 6.

We add the terms $\ketbra{\stateA\stateA}{\stateA\stateA}$
and $\ketbra{\stateB\stateB}{\stateB\stateB}$, which gives an energy penalty of $1$ for pairs $\stateA\stateA$ or $\stateB\stateB$ on
Track 0. Since $N$ is odd, the state will have energy at least one unless there is
at least one particle whose Track 0 state is $\stateLine$.
Now we add an energy penalty of $1$ for any particle whose Track 0 state is $\stateLine$ which
is not $[\stateLine,\leftend]$. We also add an energy penalty for $[\stateA, \leftend]$ or $[\stateB,\leftend]$.
The only way for a state to have energy less than one is for there to be at least one
segment and to have a $\stateLine$ on Track 0 exactly at the locations where there is a $\leftend$
on the other tracks. Finally, we add an energy penalty of $1/2$ for any particle whose
Track 0 state is $\stateLine$. Thus, the only way for a state to
have energy less than one is for there to be exactly one
segment and to have a $\stateLine$ on Track 0 exactly at the single location where there is a $\leftend$
on the other tracks.

Consider the set of all basis states with exactly one segment and
whose Track 0 state has a $\stateLine$ co-located with the \leftend\ site and alternating $\stateA$'s and
$\stateB$'s elsewhere on Track 0.  Any eigenstate outside this space has energy at least one.
Note that there are $2N$ different low-energy ways to arrange the Track 0 states, since there are
$N$ possible locations for the $\stateLine$ site and the string of alternating $\stateA$'s and $\stateB$'s can
begin with $\stateA$ or $\stateB$.
Each such choice gives rise to a subspace that is closed over $H$ and otherwise has
identical behavior with respect to the rest of the terms, so we will make an arbitrary
choice and examine the resulting subspace.
Since there is one segment, the state $\ket{\phi_{I,Y}}$ from the previous
subsection is well defined. This state is a ground state for $H$ and has energy $1/2$.
All other eigenstates within the subspace have energy at least $1/2 + \Omega(1/N^3)$.

\subsection{Accepting States}

Finally, we add a term for  those $\ket{\phi_{I, \psi}}$ that do not lead to an
accepting computation for $V$. In the last configuration of the clock,
the two right-most particles are in state $[\leftend][\arrLtwo,\oneB]$.
We assume that an accepting computation will end with the head in the
left-most place in state $q_A$.
We add a penalty for a pair whose Track 1 and 2  state is $[\leftend][\arrLtwo,\oneB]$
such that the right particle in the pair does not have a Track 5 state of $q_A$.
Thus, an accepting computation will have an energy penalty of at most
$\epsilon/N^2$ for this term where $\epsilon$ can be made arbitrarily small.
If no witness leads to an accepting computation, all states will have energy
at least $(1-\epsilon)/N^2$ from this term.

\section{Quantum Case With Reflection Symmetry}
\label{sec:quantumreflection}

Many of the Hamiltonian terms in section~\ref{sec:quantum} had a left-right asymmetry,
allowing us to, for instance, start the Turing machine head at the left end of the
line of particles.  When we add reflection symmetry to the translational invariance,
this is no longer as straightforward.  However, by increasing the number of states
and with an appropriate choice of Hamiltonian, we will be able to spontaneously break
the reflection symmetry in the vicinity of the arrow states, showing that \rTIH{1}\ with
reflection symmetry is $\QMAEXP$-complete.  In this section, we present the construction
for periodic boundary conditions, but a similar construction works with open boundary
conditions.

The basic idea of the construction is to force the arrow states to have $\stateA$ on
one side and $\stateB$ on the other.  That creates an asymmetry between the two directions,
which we can use to define left transitions and right transitions.  At one time step, we
should consider $\stateA$ to be on the left and $\stateB$ to be on the right.  However,
in the next time step, the arrow has moved, and now $\stateB$ should be on the left and
$\stateA$ should be on the right.  Therefore, we will split the arrow state into pairs of
states that keep track of which direction $\stateA$ should be in.

Now for the details of the construction.  We assume that $N$, the number of particles in the cycle, is odd. There will be a
total of seven tracks, numbered 0 through 6. A state is then specified either as
$\stateLine$ or by a $7$-tuple, specifying the state for each Track. The state $\stateLine$
spans all seven tracks.
Some of the energy penalties we introduce will be a multiple of some constant
$c$ to be chosen later.

The states for each of the Tracks 1, 2, 4, and 5 will be divided into control
states and non-control states. The non-control states
will be divided into two types: A-states and B-states.
\begin{definition}
For any particular standard basis state, we
say that its {\em configuration} is defined by the following properties:
\begin{enumerate}
\item the number of control particles
on Tracks 1, 2, 4, and 5
\item the sequence of A-states and B-states (with control states removed)
on Tracks 1, 2, 4, and 5
\item the entire state of Track 0
\item  the locations
of the particles in state $\stateLine$.
\end{enumerate}
\end{definition}

We will select the transition rules so that they never change the configuration
of a standard basis state. Thus, if we consider the subspace spanned by all the
standard basis states with a particular configuration, the Hamiltonian will
be closed on that subspace. We can therefore analyze each such subspace
separately. We will specifically be interested in the following properties:

\begin{claim}
\label{claim:qreflectionproperties}
With an appropriate choice of Hamiltonians, we can ensure that the ground state
is a superposition of basis states whose configurations satisfy the following properties:
\begin{enumerate}
\item There is one particle in $\stateLine$.
\item For Track 0, the rest of the particles alternate $\stateA$ and $\stateB$.
\item Tracks 1, 2, 4, and 5 each have exactly one control particle.
\item The non-control particles for tracks 1, 2, 4, and 5 alternate between
A-states and B-states with an A-state or control state on either side of the $\stateLine$.
\item The control particle is flanked by an A-state on one side and a B-state
or a $\stateLine$ on the other.
\end{enumerate}
\end{claim}
We will therefore restrict our attention to
configurations that have these properties.

\begin{proof}[ of claim]
Towards this end, we add the following constraints:
Track 0 has states $\stateA$, $\stateB$ and $\stateLine$. There is a penalty of
$2c$ for the pairs
$\stateA \stateA$ and $\stateB \stateB$. There is a penalty of $c$ for any occurrence of $\stateLine$.

Any state that has no $\stateLine$ or more than one $\stateLine$ will have a cost of at least $2c$.
A configuration which satisfies properties 1 and 2
will have a cost of $c$.

For Tracks 1, 2, 4, and 5, we will enforce that there is exactly one particle
that is in a control state. This control state can move around, but there
should be just one.
This is enforced as follows. The control states for a track all cost 1.
There will be two states for each non-control state, an A version and
a B version. The state $\stateLine$ spanning Tracks 0 through 6 must go next to
an A-state or a control state on Tracks 1, 2, 4, and 5. The A-states cannot be next to
each other and the B-states cannot be next to each other. The
control states can be next to the $\stateLine$, A-states, or B-states, but
not next to each other. All of these forbidden pairs have a cost
of $c$.  These rules are summarized in table~\ref{table:qreflectionAB}.  The
transitions may cause the control particle to shift over by one
site by swapping places with a non-control particle. In these
transitions, the non-control particle will remain the same with
respect to whether it is an A or B state although it may change
other aspects of its state.

\begin{table}
\begin{centering}
\begin{tabular}{c|cccc}
               & $\stateLine$ &  A  &  B  & control \\
 \hline
$\stateLine$   &       c      &  0  &  c  &    0    \\
A              &       0      &  c  &  0  &    0    \\
B              &       c      &  0  &  c  &    0    \\
control        &       0      &  0  &  0  &    c    \\
\end{tabular}
\caption{Summary of the rules for A and B-states in Tracks 1, 2, 4, 5.  Also, all control
states have an additional cost of $1$.}
\label{table:qreflectionAB}
\end{centering}
\end{table}

Property 4 of the claim follows immediately from these constraints.

If there are no control states on a Track, the total cost of the Track must be
at least $c$ (in addition to the $c$ cost of having the one $\stateLine$).
This is a consequence of odd $N$ --- alternating A-states and B-states would
then produce an A-state on one side of $\stateLine$, but a B-state on the
other side.

Any configuration that avoids the cost of $c$ must have at least
one control state on each of Tracks 1, 2, 4, and 5.
Each control state incurs a cost of one.
Thus, the minimal-cost configuration has exactly one
control state on each track of 1, 2, 4, and 5 with a total cost of $c+4$.
Any other configuration will have a cost of at least $2c$. The minimum cost
configuration will have a control particle somewhere within a
sequence of alternating B's and A's with an A on each end next
to the $\stateLine$. In order to maintain the parity, property 5 must hold.

We can choose $c=5$ to guarantee that there is an energy gap of at least
one between configurations which obey properties 1 through 5 in the claim and
those which do not.
\end{proof}

\subsection{Track 1 Transitions}

We are now ready to describe the transitions for Track 1. Since these
rules satisfy reflective symmetry, whenever there is a rule:
$ab \rightarrow cd$, there is also a rule $ba \rightarrow dc$.
For the sake of brevity, we do not always give the reflected
version of the rule with the understanding that it is implicit.

There are four types of control states: $\arrRA$, $\arrRB$, $\arrLA$ and $\arrLB$.
When we discuss Track 2, these will each be expanded into three varieties as
was done in the construction without reflection symmetry.
There are only two non-control states. These are labelled $\stateA$ and $\stateB$.
We have the following transitions:

\begin{description}
\item $\arrRA \stateA \rightarrow \stateA \arrRB$
\item $\arrLA \stateA \rightarrow \stateA \arrLB$
\item $\arrRB \stateB \rightarrow \stateB \arrRA$
\item $\arrLB \stateB \rightarrow \stateB \arrLA$
\end{description}
That is, the control state swaps its location with an adjacent non-control state, and
the arrow stays pointing the same direction while the A/B label on the arrow switches.

We have the following transitions at the ends

\begin{description}
\item $\arrLB \stateLine \rightarrow \arrRA \stateLine$
\item $\arrRB \stateLine  \rightarrow \arrLA \stateLine$
\end{description}
That is, the arrow switches direction, and the A/B label on the arrow switches
as well.

It will be useful to fix a direction for the cycle by disconnecting
the cycle at the $\stateLine$ particle and stretching out from left to right.
\begin{definition}
Once this direction is fixed, we say that a standard basis state is
\emph{correctly oriented} if one of the properties hold:
\begin{description}
\item Control particle is $\arrRA$ and the particle to its right is in state $\stateA$.
\item Control particle is $\arrLA$ and the particle to its left is in state $\stateA$.
\item Control particle is $\arrRB$ and the particle to its right is in state $\stateB$ or $\stateLine$.
\item Control particle is $\arrLB$ and the particle to its left is in state $\stateB$ or $\stateLine$.
\end{description}
In other words, the arrow points to a state which matches the A/B label on the arrow.
\end{definition}

There are $2(n-1)$ standard basis states for Track 1 that satisfy properties 1 through 4
of claim~\ref{claim:qreflectionproperties}
which are also correctly oriented. The reflection of each of these states is
a standard basis state which  satisfies properties 1 through 4 but is not correctly
oriented.
The transitions form an infinite loop over the correctly oriented Track 1 states as
shown below. The $\stateLine$ particle is duplicated on each end to illustrate the
transitions.
\begin{center}
\begin{tabular}{cc}

\begin{tabular}{c}
\stateLine \arrRA \stateA \stateB $\cdots$ \stateB \stateA \stateLine \\
\stateLine \stateA \arrRB \stateB $\cdots$ \stateB \stateA \stateLine \\
\stateLine \stateA \stateB \arrRA $\cdots$ \stateB \stateA \stateLine \\
$\cdots$\\
\stateLine \stateA \stateB $\cdots$ \arrRB \stateB \stateA \stateLine \\
\stateLine \stateA \stateB $\cdots$ \stateB \arrRA \stateA \stateLine \\
\stateLine \stateA \stateB $\cdots$ \stateB \stateA \arrRB \stateLine \\
\end{tabular}

&

\begin{tabular}{c}
\stateLine \stateA \stateB $\cdots$ \stateB \stateA \arrLA \stateLine \\
\stateLine \stateA \stateB $\cdots$ \stateB \arrLB \stateA \stateLine \\
\stateLine \stateA \stateB $\cdots$ \arrLA \stateB \stateA \stateLine \\
$\cdots$\\
\stateLine \stateA \stateB \arrLB $\cdots$ \stateB \stateA \stateLine \\
\stateLine \stateA \arrLA \stateB $\cdots$ \stateB \stateA \stateLine \\
\stateLine \arrLB \stateA \stateB $\cdots$ \stateB \stateA \stateLine \\
\end{tabular}

\end{tabular}
\end{center}

Note that we always have a two-fold degeneracy which results from reflecting
each state in this sequence.
This results in a set of states which are not correctly oriented.
The transitions are closed on the two sets of states.
Furthermore,  the structure of the transitions is identical on each set,
so we can select one orientation and analyze that subspace.
For ease of exposition then we will restrict our attention to
states on Track 1 which are correctly oriented.

We no longer need distinct endmarkers $\rightend$ and $\leftend$ since the
state of the control particle next to an endmarker indicates whether it
will sweeping to the other end or changing directions first.
For example any transition  or illegal pair involving the pair
$\leftend \arrL$ would be replaced by $\stateLine \arrLB$.
The pair $\arrLB \stateLine$ will never occur in a state which
is correctly oriented. The table below shows the pairs involving
end states in the non-reflective construction and the corresponding
pairs in the reflecting construction that we are describing in this section.

\begin{center}
\begin{tabular}{|c|c|}
\hline
$\leftend \arrL$ & $\stateLine \arrLB$,  $\arrLB \stateLine$\\
\hline
$\leftend \arrR$ & $\stateLine \arrRA$,  $ \arrRA \stateLine$ \\
\hline
$\arrL \rightend $ & $\stateLine \arrLA $, $ \arrLA \stateLine$\\
\hline
$\arrR \leftend $ & $\stateLine \arrRB $, $ \arrRB \stateLine$\\
\hline
\end{tabular}
\end{center}

In the non-reflective construction, the Track 1 control symbol $\arrL$,
has three varieties $\arrLzero$, $\arrLone$ and $\arrLtwo$.
Now in  the construction with reflection symmetry, these in turn come
in an A-version and B-version giving six control symbols:
$\arrLAzero$, $\arrLAone$, $\arrLAtwo$, $\arrLBzero$, $\arrLBone$, $\arrLBtwo$.
We have a similar set of the right-pointing control states.

\subsection{Track 2}

As discussed above, we will enforce that Track 2 will have exactly one
particle in a control state. For Track 2, the control states are
$\zeroB$ and $\oneB$. The other states $\zero$, $\one$ and $\two$ will each have an
A-version and a B-version. The A-version will be denoted by a heavier outline
as in $\zeroVerA$ and the B-version will be denoted by a lighter outline as in $\zeroVerB$.
In addition to the constraints described below, the states satisfy the constraints
given in table~\ref{table:qreflectionAB}.

We need to have a transition for every possible combination of A's and B's.
Thus the transition:
$$\fourcells{\blankL}{\zero}{\arrLone}{\zero} \rightarrow \fourcells{\arrLone}{\zero}{\blankR}{\zero}$$
is replaced by four transitions:
$$\fourcells{\stateA}{\zeroVerA}{\arrLAone}{\zeroVerB} \rightarrow \fourcells{\arrLBone}{\zeroVerA}{\stateA}{\zeroVerB},
\hspace{.25in}
\fourcells{\stateA}{\zeroVerB}{\arrLAone}{\zeroVerA} \rightarrow \fourcells{\arrLBone}{\zeroVerB}{\stateA}{\zeroVerA},
\hspace{.25in}
\fourcells{\stateB}{\zeroVerA}{\arrLBone}{\zeroVerB} \rightarrow \fourcells{\arrLAone}{\zeroVerA}{\stateB}{\zeroVerB},
\hspace{.25in}
\fourcells{\stateB}{\zeroVerB}{\arrLBone}{\zeroVerA} \rightarrow \fourcells{\arrLAone}{\zeroVerB}{\stateB}{\zeroVerA}.$$
Similarly, the transition
$$\fourcells{\blankL}{\zeroB}{\arrLone}{\zero} \rightarrow \fourcells{\arrLone}{\one}{\blankR}{\zeroB}$$
is replaced by
$$\fourcells{\stateA}{\zeroB}{\arrLAone}{\zeroVerA} \rightarrow \fourcells{\arrLBone}{\oneVerA}{\stateA}{\zeroB},
\hspace{.25in}
\fourcells{\stateB}{\zeroB}{\arrLBone}{\zeroVerA} \rightarrow \fourcells{\arrLAone}{\oneVerA}{\stateB}{\zeroB},
\hspace{.25in}
\fourcells{\stateA}{\zeroB}{\arrLAone}{\zeroVerB} \rightarrow \fourcells{\arrLBone}{\oneVerB}{\stateA}{\zeroB},
\hspace{.25in}
\fourcells{\stateB}{\zeroB}{\arrLBone}{\zeroVerB} \rightarrow \fourcells{\arrLAone}{\oneVerB}{\stateB}{\zeroB}.$$
Finally, the transitions involving an end particle will be expressed using the appropriate control
character as expression in the table above. The rule
$$\threecellsL {\arrLtwo}{\one}   \rightarrow \threecellsL {\arrRtwo}{\one}$$
will become
$$\threecellsLRefl {\arrLBtwo}{\oneVerA}   \rightarrow \threecellsLRefl {\arrRAtwo}{\oneVerA}.$$

Although all of the illegal pairs and transitions have corresponding rules obtained by reflection, only
one will apply since the states are correctly oriented and the direction of the Track 0 control
particle is determined.

In the construction without reflection symmetry, we defined a Track 2 state to be
well-formed if it had the following form:
$\leftend \one^* (\zeroB \zero^* +\oneB \two^*) \rightend$. With the new construction,
we use the following definition instead:
\begin{definition}
A Track 2 state is \emph{well-formed} if it has the form:
$$\leftend (\oneVerA \oneVerB)^* (\zeroB (\zeroVerA \zeroVerB)^*\zeroVerA +
\oneB (\twoVerA \twoVerB)^* \twoVerA) \rightend \mbox{~~~~or~~~~}
\leftend (\oneVerA \oneVerB)^* \oneVerA (\zeroB (\zeroVerB \zeroVerA)^* +
\oneB (\twoVerB \twoVerA)^* ) \rightend.$$
\end{definition}

The set of well-formed basis states are all closed under the transition
rules. If we restrict our attention to Tracks 1 and 2 and
the subspace spanned by all well-formed
clock states, the results from the previous construction apply and we know
that the unique ground state is the uniform superposition over valid clock states.

The construction without reflection symmetry makes use of illegal pairs
with a particular orientation in enforcing that states must be well-formed.
We need a means of enforcing this same property with only illegal pairs
that obey reflection symmetry.
Since we have restricted our attention to states with only one control
particle, we know that the Track 2 state must have the form
$\stateLine (\zero + \one + \two)^* (\zeroB+\oneB) (\zero + \one + \two)^* \stateLine$.
We will call a state of this form {\em proper}.
Furthermore, the set of transition rules is closed over the subspace spanned by
this set.

We will first observe  that for any proper standard basis state,
at most one transition rule applies in the forward direction and at most
one applies in the reverse direction. This follows from the fact that the
control particle on Track 1 and the location of the A-state and B-state
on either side of the control particle uniquely determine which pair the
transition will apply to. Then the argument for the previous construction
(Lemma 4.4) carries over to establish that the transitions that apply in each
direction are unique.

This means that when we look at the state graph (graph of all standard basis states with
a directed edge for each transition), it forms a set of disjoint paths over all
proper standard basis states. Since the transition rules are closed over well-formed
states, we know that a path is either composed entirely of states that are
well-formed or states that are not well-formed.
Then we will add a set of illegal pairs which all have the desired reflection symmetry
and show that for every path that is not composed of well-formed states, a fraction of
$1/2n$ of the states must have a state with an illegal pair.
This will imply that any state in the subspace formed by the closure of standard basis
states long a path composed of states that are not well-formed
will have energy $1/poly(n)$ more than the uniform superposition of the valid clock states.

We add a set of illegal pairs, each with a cost of one.
For each pair specified, its reflection is also illegal.
This first group are
\zero \one, \zero \two, \one \two.
Each of these states have an A-version and a B-version and we implicitly
disallow any combination of these. That is the pair $\zero \one$
implicitly denotes the two pairs $\zeroVerA \oneVerB$ and $\zeroVerB \oneVerA$.
Recall that it is already illegal to have two A-states or two B-states next
to each other.
Furthermore, the control states can not be next to each other, so $\oneB \zeroB$ is an
illegal pair.
These rules enforce that any state without an illegal pair must be of the form:
$$\stateLine (\zero^* + \one^* + \two^*) (\zeroB + \oneB) (\zero^* + \one^* + \two^*).$$
We also add pairs $\zeroB \two$ and $\oneB \zero$ since these never appear in
a well-formed state.
Finally, we want to have the $\one^*$ on the left end of the chain
and the $\zero^*$ or $\two^*$ on the right end of the chain.
To do this, we add illegal pairs:
$$\threecellsLRefl {\arrLB}{\zero}  \hspace{.5in}
\threecellsLRefl {\arrLB}{\two} \hspace{.5in}
\threecellsRRefl {\arrRB}{\one}
.$$
We also add illegal pairs:
$$\threecellsLRefl {\arrRA}{\zero}  \hspace{.5in}
\threecellsLRefl {\arrRA}{\two} \hspace{.5in}
\threecellsRRefl {\arrLA}{\one}.$$
These are duplicated for the A-version and B-version of \zero, \one, and \two\
as well as  \arrLBzero, \arrLBone, and \arrLBtwo\ for \arrLB\ and
\arrLAzero, \arrLAone, and \arrLAtwo\ for \arrLA.
If we have a proper state that is not well-formed and has no illegal pair,
it must either have a $\one^*$ on the right end of the chain
or a $\zero^*$ or $\two^*$ on the left end of the chain.
What will happen, is  that the control symbol on Track 1 will
cycle around to the end of the chain where the violation occurs
and reach a state with an illegal pair.

\subsection{Turing Machine Initialization}

We have already described a way to enforce that there is one control state
on Tracks 4 and 5. The particle states which correspond to
Turing Machine states will be the control states for Tracks 4 and 5.
There is some designated start state $\stateS$ for each Turing Machine
and we want to ensure that the Turing Machine starts with the head in the
leftmost location in state \stateS.
Similar to the non-reflection case, we make any state illegal in which the Track 1
state is $\arrRAzero$ or $\arrRBzero$ and the Track 4 state is anything
except \stateS\ or \blankR.
As described above, we have already enforced that there is one state corresponding
to the head location.
Finally, we make any pair illegal in which one particle is in state $\stateLine$
and the other particle has a Track 1 state of $\arrRAzero$ and the Track 4 state
is not \stateS.
This ensures that the unique head location is at the left end of the chain.
The same is done for Track 5.

\subsection{Turing Machine Implementation}

All the transitions are expended as for Track 2 to include the A-version and
the B-version of the non-control states.
Transitions involving \rightend\ and \leftend\ are replaced by rules which use
\stateLine\ and the appropriate control particle for Track 1 which will
determine whether the other particle is to the left or the right of the
\stateLine.

All the transition rules for the execution of the Turing Machines are governed
by the direction in which the $\arrLA$ and $\arrLB$ control particles move.
We can add the reflection of each of these rules and since the direction
of the control particle is fixed, only the correct set of rules will apply.

\section{The Infinite Chain}
\label{sec:infinite}

One common approach to studying translationally-invariant systems in physics is to take the
limit of the number of particles going to $\infty$.  In this case, we have to alter our approach slightly in
order to define a sensible notion of complexity.  For one thing, we have to change the quantity we
wish to evaluate.  In general, the ground state energy will be infinite when there are
infinitely many particles.  Instead, we should study the energy density in the ground state:
the energy per particle.

There is an additional complication. Throughout the rest of the paper, we have fixed a single Hamiltonian
term and used it for all values of $N$, but if we were to do that for the infinite chain, the desired answer (the
ground state energy) would be just a single number, and there would be no language for us to study the
complexity of.

The solution is to vary the Hamiltonian term. We still constrain it to interact two particles (i.e., act on
a $d^2$-dimensional Hilbert space), but now we let the entries of the matrix vary. In particular, if we specify
terms in the Hamiltonian to $n = \Theta(\log N)$ bits of precision, we can cause the ground state of the infinite
chain to break up into segments each of size $N$. Within a single segment, the Hamiltonian will act like the
usual 1D Hamiltonian discussed in section~\ref{sec:quantum}.

One might be concerned that the problem is not well-defined on an infinite chain.  After all, the hardness
constructions used in this paper cause the ground state to change considerably as $N$ varies, suggesting that
there is no well-defined limit as $N \rightarrow \infty$.  This concern is largely addressed by separating
$N$ and the number of particles, since the ground state is now the tensor product of segments of finite size,
each of which is completely well-defined.  To be even more careful, one can take a finite number of particles $M \gg N$,
using either periodic or open boundary conditions.  If we write down the restriction to $M$ particles of the
infinite ground state considered below, there is a constant-size correction to the energy due to the boundary
conditions.  However, we wish to determine the energy \emph{per particle}, and the correction to that is only $O(1/M)$.
Thus, in the limit $M \rightarrow \infty$, the energy per particle of the ground state is well-defined.

The first step is to alter the Hamiltonian for an $N$-particle chain slightly so that
we can assume that if a chain of length $N$ has a zero energy state in $\calS_{br}$, then
it must be the case that $N=3 \bmod 4$.
This can be achieved by having an additional track with states A, B, C and D.
We will use illegal pairs to enforce that this additional track has states
that are a substring of $(ABCD)^*$. Then we add illegal pairs to enforce that
only an A can be adjacent to a \leftend\ or a \rightend. Therefore, in order to
avoid additional cost, the additional track must be of the form
$\leftend A (BCDA)^* \rightend$. In order to achieve the same expressive power
as we originally had, the verifier Turing Machine $V$ will begin by erasing
the two least significant bits of its input (the count from $M_{BC}$).
Thus, the new construction will yield a state with energy $\epsilon$ on a chain of
length $N$ if and only if $N = 3 \bmod 4$ and there is a quantum witness
which causes the verifier to accept a binary
representation of $ (N-3)/4 $ with probability $(1 - \epsilon)$
in $N-2-x$ steps, where $x$ is the number of steps
required for the subroutine which deletes the last two bits of the input.
$x$ is approximately $2 \log N$.
If either $N \neq 3 \bmod 4$ or every witness causes $V$ to reject with
high probability on
input $(N-3)/4 $, the ground state energy will be at least $\Omega(1/N^3)$.
For the remainder of this section, we will assume  that $N = 3 \bmod 4$.

Now we address the case of the infinite chain in earnest.
For the next step, we make a second small change to the set of illegal pairs to allow the pair \rightend \leftend.
\begin{definition}
For a particular state, the sequence of sites extending
from a  $\leftend$ site through the next  $\rightend$ site is
a \emph{segment}.
\end{definition}
The set of bracketed and well-formed
states is exactly the same as it was for the finite chain
except that we can now have
more than one segment along the chain.
If a standard basis state contains no illegal pairs,
there must be a control symbol on Tracks 1 and 2 between \leftend\ and
\rightend\ which means we cannot have segments of  length 2.
Also, every occurrence of $\rightend$
has a $\leftend$ to its immediate right and every occurrence of $\leftend$
has a $\rightend$ to its immediate left.
In addition, the illegal pairs exclude the possibility of sequences of sites containing two \leftend\ sites with no \rightend\ in between, or sequences with two \rightend\ sites without a \leftend.  We therefore get the following lemma:
\begin{lemma}
A standard basis state with no illegal pairs either contains no $\leftend$ or $\rightend$ sites anywhere in the chain, or it can be divided up into valid segments with no extra sites in between segments.
\end{lemma}

Let $H_{TM}$ be the resulting two-particle term that includes illegal pairs and
transition rules. We omit for now the term which penalizes for a rejecting computation.
A state thus has zero energy for $H_{TM}$ if it corresponds to a
correctly executed computation that begins with the correct initial
configuration regardless of whether it accepts or not.
We also omit the terms for the boundary conditions described for the finite chain
and cycle. Instead we are going to add a new term $H_{size,N}$ which will be energetically favorable
to segments of length $N$.

Let $\calS$ be the subspace spanned by all well-formed states in the standard basis which have segments and cannot evolve via forward or backwards transitions to a state which does have an illegal pair.  Lemmas~\ref{lem:well-formed} and \ref{lem:clairvoyance} still apply since we have not changed the transition rules.
Thus, $\calS$ is the set of states which can be obtained
by starting each segment in the correct initial configuration (with an arbitrary
quantum state on the witness tape) and applying any number of transition
rules to each segment.
Since the transition rules do not alter the segments,
$H_{TM}$ is closed over $\calS$ as it was for the chain.
The final Hamiltonian $H$
will be the sum of $H_{TM}$ and a set of terms which
are all diagonal in the standard basis, which means that
$\calS$ will be closed under $H$ as well.
We will study $H|_{\calS}$.

We will add in another Hamiltonian $H_{size,N}$ that will
 be designed to be energetically favorable to segments for some specific size $N$.
We first construct a Hamiltonian $H$ of the form $H_{TM} + H_{size,N} /p(N)$
for some polynomial in $N$. In the final Hamiltonian, there will also be a term
which penalizes rejecting computations.
$H_{size,N}$ will have a fixed form (i.e., constant size to
specify) except  there will be several
coefficients which are inverse polynomials in $N$ and therefore require $O( \log N)$
bits to specify.

Since all two-particle terms are zero on the pair $\rightend \leftend$,
we can omit the two-particle terms which span two segments when considering
$H|_{\calS}$. Now $H$ can be divided into a sum of
terms, each of which acts on particles entirely within a segment.
Let $H^i$ be the sum of the terms which act on particles within segment $i$.
Let $l$ denote the length of the segment.
We can define $H^i_{size}$ and $H^i_{TM}$ similarly.
An eigenstate of $H$ in $\calS$ is then a tensor product of eigenstates
of each $H^i$ acting on the particles in segment $i$. The energy is the sum of the energies
of each $H^i$ on their corresponding eigenstate. The eigenstates for a segment
of length $l$ are exactly the same as the bracketed eigenstates for a chain of the same length.
We know that if $l \neq 3 \bmod 4$, then the ground state has energy at
least $\Omega(1/l^3)$.

We are now ready to define the final component of $H$.
Define $T_l$ to be the number of clock states for a finite chain
or segment of length $l$.
We determined in section~\ref{sec:quantum} that $T_l = 4(l-2)^2$.
Note that the symbol  \arrRzero\ appears in exactly $l-2$ clock states.
\begin{equation}
H_{size,N} = \frac{1}{N} I - 2
\ketbra{\leftend}{\leftend} + \frac{T_N}{N-2} \left( \ketbra{ \arrRzero }{ \arrRzero } \right).
\end{equation}
We will analyze the ground state energy of a segment as a function of its length.
We will need
to use the Projection Lemma from \cite{KempeKitaevRegev06}
which will allow us to focus on the
ground space of $H^i_{TM}$. Define $\ket{\phi_Y^l}$ to be the state
that corresponds to the uniform superposition of all states that can
be obtained by applying forward transition rules to the correct initial
state with quantum state $Y$ on the witness tape.
The ground space of $H^i_{TM}$ is within the span of all the $\ket{\phi_Y^l}$.
Note that $ \bra{\phi_Y^l} H^i_{size,N} \ket{\phi_Y^l} $ is the same
for all $Y$ since $H_{size,N}$ depends only on the clock tracks.

\begin{lemma}
\label{lem:proj}
\cite{KempeKitaevRegev06}
Let $H = H_1 + H_2$ be the sum of two Hamiltonians acting on a Hilbert space
$\calH = \calT + \calT^{\perp}$. The Hamiltonian $H_2$ is such that
$\calT$ is a zero eigenspace for $H_2$ and the eigenvectors in $\calT^{\perp}$
have value at least $J > 2 \lVert H_1 \rVert$. Then
$$\lambda(H_1 |_{\calT}) - \frac{ \lVert H_1 \rVert^2}{J - 2 \lVert H_1 \rVert}
\le \lambda(H) \le \lambda(H_1 |_{\calT}),$$
where $\lambda(H)$ is the ground energy of $H$.
\end{lemma}

\begin{coro}
\label{cor:proj}
There is a polynomial $p(N)$ such that $p(N)$ is $O(N^{7})$ and for any
segment of size $l \le 2N$ and
$H^i = H^i_{TM} + H^i_{size,N}/p(N)$,
$$p(N) \lambda(H^i|_{\calS}) \ge \bra{\phi_Y^l} H^i_{size,N} \ket{\phi_Y^l} - 1/2 N^2.$$
\end{coro}

\begin{proof}
We use the projection lemma with
$H_2 = p(N) H^i_{TM}$ and $H_1 = H^i_{size, N}$. Note that $H_1$ need not
be positive, although it does need to be positive on
$\calT$ in order to yield a non-trivial lower bound.
We need to establish that $\lVert H^i_{size} \rVert = O(N)$.
Since $l \le 2N$, the first term is $O(1)$.
The Hilbert space $\calS$ is the set of all well-formed, bracketed states for that
segment, so there can be at most one \arrRzero\ site.
Thus the second two terms in $H^i_{size,N}$
are at most $T_N/N$ for any state in $\calS$,
which is also $O(N)$.
The spectral gap of $H^i_{TM}$ is $\Omega(1/N^3)$, so we can choose
 $p(N)$ so that $p(N)$ is $O(N^{7})$
and $J$ (the spectral gap of $p(N) H^i_{TM}$) is at least
$2 N^2 \lVert H_1 \rVert^2 + 2 \lVert H_1 \rVert$. Using Lemma \ref{lem:proj}, we
can lower bound $p(N) \lambda(H^i|_{\calS})$ by $\bra{\phi_Y^l} H^i_{size,N} \ket{\phi_Y^l}   - 1/2 N^2$.
\end{proof}

Note that we are not able to use the projection lemma for very large $l$ because the
$\Omega(1/N^3)$ gap will not be large enough. In the lemma below, we
determine the ground state energy of a segment
as a function of its length. Large $l$ (greater than $2N$, including infinite $l$) are dealt with
separately with an argument that does not require the projection lemma.

\begin{lemma}
\label{lem:size}
The operator $H^i$ acting on the $l$
particles of segment $i$
restricted to well-formed bracketed states will
have ground state energy at most $0$ and spectral gap $\Omega(1/N^4)$ per particle for $l=N$.
The ground state
energy is $\Omega(1/N^9)$ per particle for any other value of $l$.
\end{lemma}

\begin{proof}
We have assumed throughout that $N = 3 \bmod 4$.
For any $l$ such that $l \neq 3 \bmod 4$, the ground energy is at least $\Omega(1)$
which is at least $\Omega(1/N)$ per particle. We will restrict our
attention therefore to $l$ such that $l = 3 \bmod 4$, so for $l \neq n$,
we know that $|l-n| \ge 4$.

We consider four different cases based on the size of the segment $l$.
\begin{description}
\item { $\mathbf {l=N}$:}

Consider a state $\ket{\phi_Y^N}$.  Then $\bra{\phi_Y^N} H^i_{TM} \ket{\phi_Y^N} = 0$.
Recall that $\ket{\phi_Y^N}$ is a uniform superposition of states.
There are $T_N$ distinct configurations represented in the support of
$\ket{\phi_Y^N}$. Exactly $N-2$ of these contain a \arrRzero\ site.
All of them contain one particle in state \leftend.
Therefore
\begin{equation}
\bra{\phi_Y^N} H^i_{size,N} \ket{\phi_Y^N} = \frac N N - 2 +
\frac{T_N}{(N-2)} \frac{(N-2)}{T_N} = 0.
\end{equation}
Any state $\ket{\psi}$ that is orthogonal to the subspace
spanned by the $\ket{\phi_Y^N}$
and is also in the subspace spanned by the well-formed
states for segments of length $N$
will have
$\bra{\psi} H_{TM} \ket{\psi} \ge \Omega(1/N^3)$ and
$\bra{\psi} H_{size}/p(N) \ket{\psi} \ge -1/N^7$. Thus, the spectral gap of $H^i$ will be $\Omega(1/N^4)$
per particle.
\item { $\mathbf {l > 2N}$:}

Let $\psi$ be a state in the standard basis that is well-formed, bracketed
and has length $l$.
We will only lower bound $\bra{\psi} H^i_{size}/p(N) \ket{\psi}$. Since
$H^i_{TM}$ is non-negative, the lower bound will hold for all
of $H^i$. Furthermore, we will omit the last term in $H_{size}$ because this
only adds to the energy.
Every standard basis state in a bracketed well-formed segment of length $l$ has
exactly one occurrence of $\leftend$. Therefore the energy of a segment of length
$l$ will be at least $(l/N - 2)/p(N)$. This is at least $(1/N-2/l)/p(N)$ per particle. Since $l \ge 2N+1$, this will be at least
$\Omega(1/N^8)$ per particle. Note that this holds for infinite $l$ as well.
\item {$\mathbf {2N \ge l > n}$:}

Since we can assume that $l = 3 \bmod 4$, we know that $l \ge N+4$.
We
 will use the projection lemma for this case and show that
$\bra{\phi_Y^l}H^i_{size}\ket{\phi_Y^l} \ge 1/N^2$, which by
Corollary \ref{cor:proj}
will be enough to lower bound $\lambda(H^i)$ by
$1/(2 N^2p(N))$.
\begin{eqnarray*}
\bra{\phi_Y^l}H^i_{size}\ket{\phi_Y^l} & = &
(l-N)\left( \frac 1 {N} -\frac{1}{l-2} \right) \\
& \ge & (l-N)\left( \frac 1 {N} -\frac{1}{N+2} \right) \ge
\frac{1}{N^2}\\
\end{eqnarray*}
\item {$\mathbf {l < N}$:}

Since we can assume that $l = 3 \bmod 4$, we know that $l \le N-4$.
Now we will use the projection lemma for this case and show that
$\bra{\phi_Y^l}H^i_{size}\ket{\phi_Y^l} \ge 1/N^2$, which will be enough to lower bound $\lambda(H)$ by
$1/(2 N^2 p(N))$.
\begin{eqnarray*}
\bra{\phi_Y^l}H^i_{size}\ket{\phi_Y^l} & = &
(N-l)\left( \frac 1 {l-2} -\frac{1}{N} \right) \\
& \ge & (N-l)\left( \frac 1 {N-2} -\frac{1}{N} \right) \ge \frac{4}{N(N-2)} \ge
\frac{1}{N^2}\\
\end{eqnarray*}
\end{description}
\end{proof}

Finally, we add a term which adds energy $1$ to any state which is
not an accepting computation.
If $N = 3 \bmod 4$ and there is a quantum witness that causes the verifier to
accept with probability $(1-\epsilon)$, then there is a state with energy at most $\epsilon/N^3$ per particle.
This is the state which corresponds to correct computations, using the good witness,
within consecutive segments
of length $N$. We can assume the acceptance probability has been amplified sufficiently to make $\epsilon$
small --- smaller than $1/p(N)$ is sufficient.
Note that the ground space has an $N$-fold degeneracy which
corresponds to translations of this state.

Now suppose that there is no quantum witness which
causes the verifier to accept on input $(N-3)/4$ with probability greater than $\epsilon$.
Consider a possibly infinite set of locations for the \leftend\ and \rightend\ particles
and the subspace spanned by the standard basis states corresponding to  these assignments.
We will analyze the subspace spanned by these standard basis states and argue that the
ground state energy of the Hamiltonian restricted to each such subspace is
at least $\Omega(1/N^3)$ per particle. Therefore, we assume now that the locations of
the particles in the \leftend\ or \rightend\ state are fixed.
We can partition the infinite chain into maximal pieces whose leftmost particle is
in state \leftend\ and such that the piece has no other particles in state \leftend.
Note that it may be that the leftmost piece extends infinitely to the left.
We know that any piece that is longer than $2N$ incurs a cost of at least
$\Omega(1/N^8)$ per particle from $H_{size,N}$ alone. This follows from the same analysis
for proper segments of length more than $2N$ given above.
Now suppose that a piece has length at most $2N$.
If it is not a proper segment (i.e., beginning with a \leftend, ending with a \rightend,
and no \leftend\ or \rightend\ particles in between), then there must be at least one
illegal pair. The illegal pair will cause an energy cost of $1$. Since the energy from $H_{size,N}/p(N)$
will be at least $-2/p(N)$, the total energy per particle will be $\Omega(1)$.
We also know that if the piece is a proper segment and
the length of the segment is some $l \neq N$
or $l=N$ and $N \neq 3 \bmod 4$, then there will be a cost of
at least $\Omega(1/N^9)$ per particle on that segment.
Finally, even if $l=N$ and $N = 3 \bmod 4$, every quantum witness will cause
the verifier
to reject on input $(N-3)/4$ with probability at least $(1-\epsilon)$.
When $\epsilon$ is small enough (polynomially small is sufficient),
using standard $\QMA$-completeness proof techniques (see, e.g.~\cite{QMA1D}), we can show that there will be an energy penalty of at least $\Omega(1/T_N^2)$ for that segment,
which is $\Omega(1/N^5)$ per particle.

\section{Conclusion}

We have shown that a class of classical tiling problems and $1$-dimensional quantum Hamiltonian problems can be proven hard, even when the rules are
translationally-invariant and the only input is the size of the problem.  While this result was motivated by the desire to see if it could be hard to
find the ground state in some physically interesting system, it is true that the tiling problem and Hamiltonian problem for which we prove hardness are
not themselves particularly natural.  Still, given that very simple cellular automata can be universal, it seems quite possible that even some very
simple tiling and Hamiltonian problems are complete for $\NEXP$ and $\QMAEXP$ respectively.  Finding one would be very interesting.

More generally, one can ask which Hamiltonians are computationally hard and which are easy.  Our approach provides a framework to investigate the question, although our techniques are not powerful enough to answer it except in special cases.  Perhaps hard Hamiltonians are pervasive for particles of large enough dimension, or perhaps they are always outlying exceptional points.  The computational hardness structure of the space of Hamiltonians is clearly somewhat complicated.  For instance, for the infinite chain discussed in section~\ref{sec:infinite}, the hard class of Hamiltonians considered form a sequence (as $N \rightarrow \infty$) that converges to a Hamiltonian that is not itself hard: It lacks the $H_{size}$ term, so one can have an infinite chain with no segments or control particles, and the ground state energy is $0$.

It would also be nice to understand the finite-temperature behavior of these systems better.  The hard constructions should have a very complicated landscape of energy states, and the actual ground states are highly sensitive to the exact number of particles.  This suggests that even at finite temperature, these systems should have complicated dynamics and behave like a spin glass, but it is not completely clear.

There remain many additional variants of tiling and Hamiltonian problems that we have not studied.  Probably for most variations (e.g., triangular lattice instead of square lattice), we can expect similar results.  However, there are some borderline cases for which the answer is not clear.  The case of WEIGHTED \tiling\ with reflection symmetry and constant allowed cost remains open for the four-corners or open boundary conditions.  When we have translational invariance but no additional symmetry, \rTIH{2} is as hard as \rTIH{1}, but the construction for \rTIH{1} with reflection invariance does not generalize to \rTIH{2} with rotational invariance.  Can rotationally- and translationally-invariant Hamiltonians be hard in two or more dimensions?  Also, the hardness result for \ITIH\ uses quantum properties, and we do not have a similar classical result.

Another interesting avenue to pursue would be to apply a similar idea to other problems.  For instance, the game of go produces a $\PSPACE$-complete problem~\cite{papa95}.  However, the computational problem GO is defined by asking whether black can force a win given a particular board configuration as input.  However, there is no guarantee that these board configurations would appear in a regular game of go, which starts from a blank board.  A more natural problem arising from GO is to ask who wins with optimal play when starting from an empty $N \times N$ board.  As with our tiling and Hamiltonian problems, the only thing that varies is the size; the rules and initial configuration are fixed.  Thus, we would wish to show that this variant of GO is $\EXPSPACE$-complete.  Our techniques will not solve this problem, but at least our result points the way to ask the right question.

\section*{Acknowledgements}
We would like to thank John Preskill for pointing out the connection to $N$-REPRESENTABILITY, Cris Moore and Igor Pak for some information about
classical tiling problems, and Xiao-Gang Wen and Ignacio Cirac for suggesting some of the open problems mentioned in the conclusion.  D.G.\ was supported by Perimeter Institute for Theoretical Physics, CIFAR, and NSERC.  Research at Perimeter Institute is
supported by the Government of Canada through Industry Canada and by the Province of Ontario through the Ministry of Research \& Innovation.
S.I. was supported in part  by NSF Grant CCR-0514082.  Part of this work was done
while both authors were visiting the Institute for Quantum Information
at Caltech.


\end{document}